\documentclass[12pt]{article} 
\usepackage[utf8]{inputenc}
\usepackage{thm-restate}
\usepackage[round]{natbib}
\usepackage{array,epsfig,fancyheadings,rotating}
\usepackage{todonotes}
\usepackage[]{hyperref}  
\usepackage{sectsty, secdot}

\usepackage[margin=1in]{geometry}
\usepackage{setspace}
\usepackage{mathtools}
  
\usepackage{amsmath}
\usepackage{amssymb}
\usepackage{amsfonts}
\usepackage{multirow}
\usepackage{amsthm}
\usepackage{color}

\usepackage{mycommands}
\usepackage{afterpage}

\begin{document}


\renewcommand{\baselinestretch}{2}

\markright{ \hbox{\footnotesize\rm
}\hfill\\[-13pt]
\hbox{\footnotesize\rm
}\hfill }

\markboth{\hfill{\footnotesize\rm FIRSTNAME1 LASTNAME1 AND FIRSTNAME2 LASTNAME2} \hfill}
{\hfill {\footnotesize\rm FILL IN A SHORT RUNNING TITLE} \hfill}

\renewcommand{\thefootnote}{}
$\ $\par


\fontsize{12}{14pt plus.8pt minus .6pt}\selectfont \vspace{0.8pc}
\centerline{\large\bf 
{ Sufficient and Necessary Conditions for the Identifiability of   }}
\centerline{\large\bf 
{ DINA  Models with Polytomous Responses }}
\vspace{.4cm} 
\centerline{Mengqi Lin and Gongjun Xu} 
\vspace{.4cm} 
\centerline{\it University of Michigan}
 \vspace{.55cm} \fontsize{9}{11.5pt plus.8pt minus.6pt}\selectfont

\begin{abstract}
Cognitive Diagnosis Models (CDMs) provide a powerful statistical and psychometric tool for researchers and practitioners to learn fine-grained diagnostic information about respondents' latent attributes. 
There has been a growing interest in the use of CDMs for polytomous response data, as more and more items with multiple response options become widely used. Similar to many latent variable models, the identifiability of CDMs is critical for accurate parameter estimation and valid statistical inference. However, the existing identifiability results are primarily focused on binary response models and have not adequately addressed the identifiability of CDMs with polytomous responses.
This paper addresses this gap by presenting sufficient and necessary conditions for the identifiability of the widely used DINA model with polytomous responses, with the aim to provide a comprehensive understanding of the identifiability of CDMs with polytomous responses and to inform future research in this field.

\vspace{9pt}
\end{abstract}
\par

\def\thefigure{\arabic{figure}}
\def\thetable{\arabic{table}}

\fontsize{12}{14pt plus.8pt minus .6pt}\selectfont

\newpage 
\afterpage{\cfoot{\thepage}}
\section{Introduction}


Cognitive Diagnosis Models (CDMs), which serve as a powerful tool to infer subjects' latent attributes such as skills, knowledge, or psychological disorders based on their responses to some designed diagnostic items in the cognitive diagnosis assessment, have drawn increasing attention over the years.
As a family of  discrete latent variable models, its popularity is not limited to educational assessments 
\citep*{Junker,davier2008general,henson2009defining,rupp2010diagnostic,de2011generalized,wang2018tracking}, psychiatric diagnosis of mental disorders \citep*{Templin,dela2018}, and epidemiological and medical measurement studies \citep*{wu2017nested,o2019causes}.

Various CDMs have been developed with different diagnostic assumptions and modeling goals, among which the Deterministic Input Noisy output ``And" gate model \citep[DINA;][]{Junker}, which assumes that subjects are expected to complete an item correctly only when they possess all required attributes, is one of the most popular ones. Furthermore, the DINA model also serves as a basis for a larger range of more general CDMs, including the general diagnostic model \citep{davier2008general}, the log linear CDM \citep[LCDM;][]{henson2009defining}, and the generalized DINA model \citep[GDINA;][]{de2011generalized}. As tests with polytomous responses appear more frequently in practice, the study of CDMs with polytomous responses has also grown in popularity \citep{culpepper2021inferring}.
Specifically, 
several models concerning polytomous responses were proposed, such as General Diagnostic Models \citep[GDM;][]{davier2008general}, General Polytomous Diagnosis Models \citep[GPDM;][]{GPDINA}, and Sequential Cognitive Diagnosis Models \citep[Sequential CDM;][]{sequential}.

As is the case with many statistical methods, ensuring the models applied in the cognitive diagnosis are statistically \emph{identifiable} is fundamental
to achieve reliable and valid diagnostic assessment. Additionally, this is also a necessity for consistent estimation of the model parameters of interest and  valid statistical inferences.  
 The study of identifiability issue for CDMs has long been considered, such as \citet{DiBello}, \citet{MarisBechger}, \citet{TatsuokaC09}, \citet{deCarlo2011}, and \citet{davier2014dina}. Considerable identifiability developments have been added to the CDM literature, such as DINA model and its generalizations in recent years. For instance, \citet{xu2016} and \citet{id-dina} discussed the sufficient and necessary condition for DINA model with binary responses. \citet{xu2017},  \citet{slam,partial}, \citet{chen2020sparse} and \cite{culpepper2022note} discussed identifiability for more generally restricted latent class models. However,
these results are targeted for dichotomous responses specifically, and the requirements for the identifiability of models with polytomous responses have sparingly been taken into consideration. For instance, \citet{Culpepper2019} and
\citet{fang2019}
discussed the sufficient condition for the identifiability of general CDMs with polytomous responses, while the necessity of those conditions remains an open problem.

Our paper fills this gap by providing sufficient and necessary conditions for the identifiability of CDMs with polytomous responses. 
In particular, we focus on two commonly used polytomous responses models under the DINA model setting: the GPDM \citep{GPDINA} under the DINA model, which we refer as \emph{GPDINA}, and the sequential CDM \citep{sequential} under the DINA model, which we refer as \emph{Sequential DINA model}. There are several challenges in developing the identifiability of the polytomous responses models. Firstly,
in binary responses DINA models, 
the uncertainty of each item is charaterized by two item parameters, whereas in polytomous responses models, each item generally involves more than two parameters. Therefore, polytomous responses models have more parameters to identify, which makes its identifiability more challenging. 
What is more intricate is that the dependency structure between these parameters is different from that of the binary response models. This is because, in addition to accounting for dependencies across items, polytomous models must also consider the dependency of parameters within a single item.
Moreover, the technical tool, $\TT$-matrix \citep{JLGXZY2011,xu2017}, which has been widely used in the identifiability literature, is restricted to binary responses models currently, to our knowledge. 

To address these challenges, we generalize the $\TT$-matrix framework to the more complex polytomous model settings. 
Based on different dependency structure of the parameters of the two models, the generalizations of the $\TT$-matrix for the two considered models (i.e., GPDINA and sequential settings) are also different. In particular, there is a significant difference in the structure of the $\TT$-matrix for Sequential DINA model, as compared to the $\TT$-matrix for binary DINA models, since the sequential modeling introduces more complex and challenging structure than the binary DINA case.
With this powerful tool,
we establish   sufficient and necessary conditions for the identifiability of the GPDINA and the Sequential DINA models.
Our proposed conditions   ensure  the identifiability and also specify the practical requirements that the two models need to process to be identifiable. 
Through the duality of the DINA and DINO models \citep{chen2015statistical}, the identifiability finding can be immediately applied to the two models under the  DINO setting.
Moreover, our results not only extend many existing results aimed at binary data to the polytomous case, but also shed light on the study of more general polytomous CDMs, which cover the considered  DINA models as submodels. 
Practically, the sufficient and necessary condition solely depends on the $\QQ$-matrix structure, and this easily verifiable requirement would serve as a practical guideline for developing cognitive tests that are both statistically valid and estimable. 

The rest of the paper is organized as follows. Section 2 introduces the model setup and brings up the definition of identifiability. Section 3 introduces a powerful tool $\TT$-matrix, specific to the polytomous responses models and develops the identifiability results, examples are also provided for illustration. Section 4 gives further discussion, and the supplementary material provides the proofs for the main results.

\section{Model Setup}
Before we present our results, we first introduce some notations. 
Let $\ee_j = (0, \ldots, 1, 0, \ldots, 0)^\top$ denote the vector where only the $j$-th entry is 1. Let $\one = (1,\ldots,1)^\top$ denote the vector of all ones and $\zero = (0, \ldots, 0)^\top$ denote the vector of all zeros. Let $\cI_K$ denote the $K$-dimensional identity matrix. For a positive integer $m$, we denote $[m] = \{1, \ldots, m\}$. Let $\circ$ denote the Hadamard product (element-wise product) of vectors. For instance, for $\aa = (a_1, \ldots, a_m)^\top$ and $\bb = (b_1, \ldots, b_m)^\top$, $\aa \circ \bb = (a_1b_1, \ldots, a_m b_m)^\top$.
Let $\otimes$ denote the Kronecker product between matrices. For example, for $\cc = (c_1, \ldots, c_n)^\top \in \R^n$, 
\begin{equation*}
    \aa \otimes \cc = \begin{pmatrix}
        a_1\cc\\
        a_2\cc\\
        \vdots\\
        a_m\cc
    \end{pmatrix} \in\R^{mn\times1},\;\;\aa \otimes \cI_K = \begin{pmatrix}
    a_1\cI_K\\
    a_2\cI_K\\
    \vdots\\
    a_m\cI_K
\end{pmatrix} \in \R^{mK \times K}.
\end{equation*}

Assume we have $J$ polytomous items to measure $K$ unobserved binary latent attributes, 
and a binary latent attribute profile can be written as $\aaa=(\alpha_1,\ldots,\alpha_K)^\top$, where $\alpha_k \in \{0, 1\}$. So there are $2^K$ attribute profiles in total.
For $j \in [J]$, 
define positive integer $H_j$ to be the number of non-zero categories (levels) 
the $j$-th polytomous item has, therefore, item $j$ has $H_j+1$ categories in total, i.e., 0, 1, $\ldots, H_j$. 
Accordingly, we define the observed random variable response $\RR =(R_1,\ldots,R_J)^\top$, with $R_j \in \{0, 1, \ldots, H_j\}$, and denote the set of all possible responses as   $\mcs = \{\rr = (r_1, \ldots r_J): r_j \in \{0, 1, \ldots, H_j\}\}$.

In the CDM literature, the relationships  between attributes and  items are characterized by the $\QQ$-matrix, which was proposed by \citet{Tatsuoka1983}. Different from CDMs with binary responses, for polytomous responses, 
the
interpretations of the entries in the $\QQ$-matrix differ according to different modelings. In the following, we focus on two popular models under the DINA assumption, the \emph{general polytomous diagnosis model} (GPDINA) by \citet{GPDINA} and the \emph{Sequential DINA model} by \citet{sequential} separately and introduce different ways of specifying the $\QQ$-matrix for polytomous CDMs. 

\subsection{The GPDINA model}\label{sec-setup-GPDINA}
In GPDINA \citep{GPDINA} (the GPDM under the DINA assumption), for models with $J$ items and $K$ attributes, we define a $J \times K$ binary $\QQ$-matrix. 
The entry $q_{jk}$ of the $\QQ$-matrix is interpreted as follows: $q_{jk} = 1$ means completing (responding) any non-zero category of item $j$ requires attribute $k$, and $q_{jk} = 0$ means completing any non-zero categories does not require attribute $k$. So the $j$-th row of the $\QQ$-matrix, $\qq_j$, denotes the attributes required to complete any non-zero categories for item $j$. Therefore, any non-zero category of the same item requires the same attributes and shares the same $\qq$-vector. In other words, non-zero categories of an item are indistinguishable, can be exchanged.

We consider the DINA assumption  under the GPDINA framework. As in the DINA model for binary data, we denote the ideal response $\xi_{j, \aaa} = I(\aaa \succeq \qq_{j})$. To further quantify the  uncertainty of the responses, define the item parameters as:
\begin{align}\label{param:GPDINA}
    \theta_{j,l}^+ :=  P(R_{j}=l\mid \xi_{j,\aaa} =1),\;\;l \in [H_j],\\
    \theta_{j,l}^- :=  P(R_{j}=l\mid \xi_{j,\aaa} =0),\;\;l \in [H_j],
\end{align}
where $\theta_{j,l}^+$ means the probability of completing category $l$ of item $j$ given the attribute profile $\aaa$ is capable of completing it and $\theta_{j,l}^-$ means the probability of completing category $l$ of item $j$ given the attribute profile $\aaa$ is not able to complete it.
Then $1-\theta_{j,l}^+$ can be interpreted as slipping parameter and $ \theta_{j,l}^-$ interpreted as the guessing parameter \citep{Junker}, and we assume that $\theta_{j,l}^+ > \theta_{j,l}^-$ for 
$l \in [H_j]$ and $j \in [J]$. As we can see,
although the attributes required by different categories of the same item are the same, here we allow the response uncertainty to be heterogeneous, i.e., $\theta_{j,l}^+$ and $\theta_{j,l}^-$ can be different across $l$. 
So in total we have $2\sum_{j=1}^J H_j$ item parameters, and the multiplicity of the item parameters is one of the aspects that makes polytomous responses models different from the binary DINA models.
For notation convenience, we also let
\begin{align}
    P(R_{j}=0\mid \xi_{j,\aaa} =1) = 1-\sum_{l=1}^{H_j}\theta_{j,l}^+\;:= \theta_{j,0}^+,\\
    P(R_{j}=0\mid \xi_{j,\aaa} =0) = 1-\sum_{l=1}^{H_j}\theta_{j,l}^- \;:= \theta_{j,0}^-.
\end{align}

When $\qq_{j} = \zero$, $\xi_{j,\aaa} \equiv 1$ for all $\aaa$, then $\theta_{j,l}^-$ is not defined for all $l \in [H_j]$. In the following Proposition~\ref{propGPDINA}, we  will show that excluding these zero $\qq$-vectors does not affect our analysis. 
Let 
\[
\ttt_{j}^+ = (\theta_{j,1}^+,\; \theta_{j,2}^+, \;\ldots \theta_{j,H_{j}}^+)^\top \; \text{and} \;\;\ttt_{j}^- = (\theta_{j,1}^-,\; \theta_{j,2}^-, \;\ldots \theta_{j,H_{j}}^-)^\top,
\]
$\ttt^+ = (\ttt_{j}^+)_{j=1}^J$ and  $\ttt^- = (\ttt_{j}^-)_{j=1}^J$, where there are $\sum_{j=1}^J H_j$ entries in both $\ttt^+$ and $\ttt^-$. Denote $p_{\aaa}$ as the proportion of attribute profile $\aaa$ in the population and $\pp:= (p_{\aaa}:\aaa\in\{0, 1\}^K)^\top$, which satisfies  $\sum_{\aaa \in \{0, 1\}^K}p_{\aaa}=1$, and we assume that $p_{\aaa}>0$ for all $\aaa$.
Given the attribute profile $\aaa$, assume that a subject's responses to the $J$ items are independent. For $\rr= (r_1,\ldots,r_J)^\top \in \mcs$, we have
 \begin{equation}\label{prob}
 P(\RR=\rr\mid \QQ,\ttt^+,\ttt^-,\pp) 
 =  \sum_{\aaa\in\{0, 1\}^K}p_{\aaa} \prod_{j=1}^J (\theta_{j, r_j}^+)^{\xi_{{j},\aaa} } (\theta_{j, r_j}^-)^{(1-\xi_{{j},\aaa})} .
\end{equation}
We use the following example to further the illustration of the model setup.
\begin{example}\label{ex1}
Suppose there are two polytomous items, each with two non-zero categories, so then $J = 2$ and $H_1 = H_2 = 2$. Suppose only two attributes $\alpha_1$ and $\alpha_2$ are involved, and the $\QQ$-matrix takes the following formula:
\[
\QQ = 
\begin{pmatrix}
item \;1 \begin{cases}
 \;\;\;\qq_{1} = [1\;\; 0] \\
\end{cases}\\
    \hdashline
item \;2 \begin{cases}
\;\;\;\qq_{2} = [0\;\; 1]\\
\end{cases}\\
\end{pmatrix}.
\]
The dashline ``- - -" is used to separate different items.
Therefore, the first and the second categories of the first item both require solely $\alpha_1$, and the first and the second categories of the second item both require solely $\alpha_2$. In particular, attribute profile $\aaa = (1, 0)$ has $\xi_{1,\aaa} = 1$ and $\xi_{2,\aaa} = 0$. Thus,
\[
P(R_1=1\mid \aaa) = \theta_{1,1}^+;\;\;
P(R_1=2\mid \aaa) = \theta_{1,2}^+;\;\;P(R_2=1\mid \aaa) = \theta_{2,1}^-;\;\;P(R_2=2\mid \aaa) = \theta_{2,2}^-,
\]
whereas for attribute profile $\aaa = (0, 1)$,  $\xi_{1,\aaa} = 0$,  $\xi_{2,\aaa} = 1$, and
\[
P(R_1=1\mid \aaa) = \theta_{1,1}^-;\;\;P(R_1=2\mid \aaa) = \theta_{1,2}^-;\;\;P(R_2=1\mid \aaa) = \theta_{2,1}^+;\;\;P(R_2=2\mid \aaa) = \theta_{2,2}^+.
\]
Therefore, attribute profile $\aaa = (1, 0)$ has higher probability of completing the two non-zero categories of the
first item but lower probability of completing the two non-zero categories of the second item. Distributions for profiles with $\aaa = (1, 1)$ and $\aaa = (0, 0)$ can be similarly obtained as well. 
\end{example}
Under the above GPDINA model setup, the model parameters include $(\ttt^-, \ttt^+, \pp)$. To study the identifiability of these parameters, we formally introduce the definition in the following, and we defer the identifiability result in Section~\ref{sec-id}.

\paragraph{Identifiability.}
We say that the GPDINA parameters are identifiable if there  is no $(\bar{\ttt}^+,\bar{\ttt}^-,\bar{\pp})$ $ \neq (\ttt^+,\ttt^-,\pp)$ such that
\begin{equation}\label{eq-orig}
P(\RR=\rr\mid \QQ,\ttt^+,\ttt^-,\pp) = P(\RR=\rr\mid \QQ,\bar{\ttt}^+,\bar{\ttt}^-,\bar{\pp}) \mbox{ for all }  \rr\in \mcs. 
\end{equation}

To simplify our discussion of the identifiability issue, we 
assume that $\qq_{j} \neq 0$ for all $j \in [J]$ without compromising the validity of the analysis, thanks to  the following proposition.
\begin{restatable}{proposition}{propGPDINA}\label{propGPDINA}
    Let $\Delta = \{j \in [J]: \qq_j = \zero\}$ denote the set of items whose $\qq$-vectors are zero, then the GPDINA model parameters with $\QQ$-matrix are identifiable if and only if the GPDINA model parameters with $\QQ_{-\Delta}$-matrix are identifiable, where $\QQ_{-\Delta}$ is obtained by removing the $\qq$-vectors in $\QQ$ corresponding to the
    items in $\Delta$.
\end{restatable}
\subsection{The Seqential DINA model}\label{sec:seq}
Another popular modeling approach for polytomous responses is the Sequential DINA model, proposed by \citet{sequential}.
In the Sequential DINA model with $J$ items and $K$ attributes, we assume that the subject's response to item $j \in [J]$, 
with $R_j = 0$ indicates that the subject fails to complete the first category, and $R_j = r_j >0$ indicates that the subject
has completed categories $1, \ldots, r_j$ successfully and failed to complete category $r_j+1$. $R_j=H_j$ simply means the subject successfully completed all the categories.  
Therefore, categories within one item are not exchangeable, and such ordered categories make it different from the previous GPDINA model setup.

Due to the sequential hierarchy of the categories, different categories could require different attributes. What's worth noticing is that, though response categories are assumed to be attained sequentially, there is no required structure for the attributes required by different categories. 
For each item $j$, its different categories should have their corresponding $\qq$-vectors. In \citet{sequential}, they refer such $\QQ$-matrix as \emph{restricted} $\QQ$-matrix.
As defined, the polytomous item $j$ has $H_j$ non-zero categories, so for the associations between the attributes and the polytomous item $j$, we have $H_j$ rows in the $\QQ$-matrix to characterize such information. With each row having $K$ entries indicating which attributes are required by the category, the $\QQ$-matrix can be summarized as a $(\sum_{j=1}^J H_j) \times K$ binary matrix.
Specifically,  we index the $\QQ$-matrix in the following way: for $l \in [H_j]$, we define the $(j,l)$-th row of the matrix, $\qq_{j,l}$, as a $K$ dimensional binary vector indicating the association between the  
 category $l$ of
item $j$ and the $K$ attributes. 
According to our model construction, the $\qq_{j, l}$ vector indicates the attributes required to complete category $l$ of item $j$, \emph{given that} the subject has successfully completed the previous categories $1, \ldots, l-1$. To further illustrate the model setup, we present an example in the following.

\begin{example}\label{ex2}
Suppose there are two polytomous items with $H_1 = H_2 = 2$ and  two attributes $\alpha_1$ and $\alpha_2$, and
\begin{equation}
\QQ = 
\begin{pmatrix}
item \;1 \begin{cases}
category \;1 \;\;\;\qq_{1,1} = [1\;\; 0] \\
category \;2\;\;\;\qq_{1,2} = [0\;\; 1]\\
\end{cases}\\
\hdashline
item \;2 \begin{cases}
category \;1 \;\;\;\qq_{2,1} = [0\;\; 1] \\
category \;2\;\;\;\qq_{2,2} = [1\;\; 1]\\
\end{cases}
\end{pmatrix}.
\end{equation}
Therefore, to complete the first category of the first item, a subject needs to require the first attribute, and given that the subject has completed the first category, he/she needs to require the second attribute to complete the second category of the first item.
\end{example}
Since different categories require different attributes, the ideal response needs to be specified accordingly to different categories. We define the ideal response as $\xi_{j, l, \aaa} = I(\aaa \succeq \qq_{j,l})$ for category $l$ of item $j$. This is also different from the setup in GPDINA, for which  we only need to define item-wise ideal response.
To quantify the uncertainty of the response to different categories, we define the item parameters specific to the Sequential DINA model as:
\begin{align}\label{seq:param}
    \beta_{j,l}^+ &:=  P(R_{j}\geq l\mid R_{j}\geq l-1,\; \xi_{j, l, \aaa} =  1),\;\;\;\; l \in [H_j],\\
    \beta_{j,l}^- &:=  P(R_{j}\geq l\mid R_{j}\geq l-1,\; \xi_{j, l, \aaa} = 0),\;\;\;\; l \in [H_j],
\end{align}
and we assume that $0 \leq \beta_{j,l}^- < \beta_{j,l}^+ \leq 1$. Note that the inequality $\beta_{j,l}^- <\beta_{j,l}^+$ is assumed to respect the monotonicity assumption of the latent attributes \citep{xu2016}, which is also needed to avoid the label switching issue of the DINA model.  Consequently, $\beta_{j,l}^-$ is permitted to take on values within the range $[0, 1)$ while  $\beta_{j,l}^+$ can take on values within the range $(0, 1]$.
These parameters characterize the probability of completing category $l$ of item $j$ \emph{given} a subject with attributes $\aaa$ has completed the previous categories. 
Furthermore, $1- \beta_{j,l}^+$ can be interpreted as the slipping parameter and $ \beta_{j,l}^-$  interpreted as the guessing parameter \citep{Junker}.
Also notice that
\begin{align*}
    P(R_{j}\geq 0\mid \xi_{j, l, \aaa} =  1) = 1,\;\;\;\; l \in [H_j],\\
      P(R_{j}\geq 0\mid \xi_{j, l, \aaa} = 0)  = 1,\;\;\;\; l \in [H_j],
\end{align*}
and we let $\beta_{j,H_j+1}^+ = \beta_{j,H_j+1}^- = 0$.

To see how these item parameters are related to the model setup in \citet{sequential}, we formulate several concepts in their paper as the following.
The \emph{processing function} 
$S_j(l\,|\,\aaa)$
in \citet{sequential}, which denotes the probability of completing category $l$ of item $j$ provided that they have already completed category $l-1$ successfully, given the attribute profile $\aaa$, can be written as
\[
S_j(l\,|\,\aaa) = (\beta_{j,l}^+)^{\xi_{j,l,\aaa}}
(\beta_{j,l}^-)^{1-\xi_{j,l,\aaa}}= P(R_{j}\geq l\mid R_{j}\geq l-1,\; \aaa), \;\;\; l \in [H_j].
\]
Let $S_j(0\,|\,\aaa) = P(R_j \geq 0\; |\; \aaa)= 1$ and $S_j(H_j+1\,|\,\aaa) = 0$.
Then noticing that 
\begin{align*}
    P(R_j \geq r_{j}\, |\; \aaa) &= \prod_{l=1}^{r_j} P(R_j \geq l\; |\; R_j \geq l-1,\; \aaa)\cdot P(R_j \geq 0\; |\; \aaa)\\
    &= \prod_{l=1}^{r_j} S_j(l\,|\,\aaa)\\
    &= \prod_{l=1}^{r_j} (\beta_{j,l}^+)^{\xi_{j,l,\aaa}}
(\beta_{j,l}^-)^{1-\xi_{j,l,\aaa}},
\end{align*}
given the attribute profile $\aaa$, the probability of $R_j = r_j$ can be written as
\begin{align*}
    P(R_j = r_j\, |\; \aaa) &= P(R_j \geq r_{j}\, |\; \aaa) - P(R_j \geq r_j+1 \,|\; \aaa)\\
    &= \left[1-S_j(r_j+1\,|\,\aaa)\right] \prod_{l=0}^{r_j} S_j(l\,|\,\aaa).
\end{align*}

Similar to   GPDINA, when $\qq_{j,l} = \zero$, $\xi_{j,l,\aaa} \equiv 1$ for all $\aaa$, and then $\beta_{j,l}^-$ is not defined. We will show later in Proposition~\ref{propseq} that excluding these categories with $\qq_{j,l} = \zero$ does not affect our analysis. Note that when $\beta_{j,l}^- = 0$ ($\qq_{j,l}$ is not necessarily $\zero$), some model parameters may not be well-defined. 
Suppose category $l^*$ is the first category in item $j$ which appears to have $\beta_{j,l^*}^- = 0$, i.e., $\beta_{j,l}^- > 0$ for $l < l^*$.
If we denote $\Gamma_{j,l^*}^- := \{\aaa: \xi_{j,l^*,\aaa} = 0\}$ as the set of attribute profiles that are not able to complete the $l^*$-th category of item $j$ ideally, and if the probability of  guessing correctly category $l^*$ of item $j$ is also 0, then there's no way for the subject to complete  higher categories of item $j$. So we define for $\aaa \in \Gamma_{j,l^*}^-$,
\begin{equation}\label{eq:beta zero}
    \beta_{j,l}^+ = \beta_{j,l}^- = 0,\;\;\text{for}\;\; l > l^* .
\end{equation}

Assume that a subject’s responses
to the $J$ items are conditionally independent given the attribute profiles. We let \[\bbb_j^+ = \left(\beta_{j,1}^+,\; \beta_{j,1}^+\beta_{j,2}^+,\; \ldots \;\prod_{l=1}^{H_j}\, \beta_{j,l}^+\right) \;\text{and}\; \bbb_j^- = \left(\beta_{j,1}^-,\; \beta_{j,1}^-\beta_{j,2}^-,\; \ldots \;\prod_{l=1}^{H_j}\, \beta_{j,l}^-\right),\; \text{for}\; j \in [J] \]
and $\bbb^+ = (\bbb_1^+, \bbb_2^+, \ldots, \bbb_J^+)$, $\bbb^- = (\bbb_1^-, \bbb_2^-, \ldots, \bbb_J^-)$, then 
 \begin{equation}
 P(\RR=\rr\mid \QQ,\bbb^+,\bbb^-,\pp) 
 =  \sum_{\aaa\in\{0, 1\}^K}p_{\aaa} \prod_{j=1}^J P(R_j = r_j |\; \aaa).
\end{equation} 
The Sequential DINA model parameters consist of $(\bbb^+, \bbb^-, \pp)$. Following the  literature, we formally define the identifiability for the Sequential DINA model in the following.
\paragraph{Identifiability.} We say that the Sequential DINA model parameters are identifiable if there is no $(\bar{\bbb}^+,\bar{\bbb}^-,\bar{\pp})\neq (\bbb^+,\bbb^-,\pp)$ such that
\begin{equation}\label{eq-orig-seq}
P(\RR=\rr\mid \QQ,\bbb^+,\bbb^-,\pp) = P(\RR=\rr\mid \QQ,\bar{\bbb}^+,\bar{\bbb}^-,\bar{\pp}) \mbox{ for all }  \rr\in \mcs. 
\end{equation}

Similar to   GPDINA, in the following proposition, we show that excluding categories with $\qq_{j,l} = \zero$ does not influence our analysis of the identifiability.  Therefore, for simplicity, we  
assume that $\qq_{j,l} \neq 0$ for all $l \in [H_j],\; j \in [J]$ in this paper. 

\begin{restatable}{proposition}{propseq}\label{propseq}
    Let $\Delta^s = \left\{(j, l) :  \qq_{j,l} = \zero\right\}$ denote the set of categories whose $\qq$-vectors are zero, then the Sequential DINA model parameters with $\QQ$-matrix are identifiable if and only if the Sequential DINA model  parameters with $\QQ_{-\Delta^s}$-matrix are identifiable, where $\QQ_{-\Delta^s}$ is obtained by removing the $\qq$-vectors in $\QQ$ corresponding to the
    categories in $\Delta^s$.
\end{restatable}

\subsection{Relationship between  GPDINA and  Sequential DINA models}
In this section, we briefly discuss the relation between the GPDINA model and the Sequential DINA model.

Fundamentally, GPDINA and   Sequential DINA models differ by the hierarchy of the categories of items. In GPDINA, different non-zero categories of the same item can be exchanged and share the same $\qq$-vector. Whereas in Sequential DINA model, different non-zero categories are generally not exchangeable and need to be completed sequentially, and different non-zero categories are allowed to have arbitrarily different $\qq$-vectors. However, when all the non-zero categories of an item share the same $\qq$-vector, the Sequential DINA model becomes equivalent to GPDINA.

Formally, in Sequential DINA model, when $\qq_{j,1} = \ldots = \qq_{j, H_j}$, such $\QQ$-matrix is referred to as \emph{unrestricted} $\QQ$-matrix \citep{sequential}, we have
 $\xi_{j,1,\aaa} = \ldots = \xi_{j,H_j,\aaa}$ for all $\aaa$ and $j \in [J]$. Under this $\QQ$-matrix, attribute profile $\aaa$ is either capable of completing all the non-zero categories of an item or unable to complete any non-zero category. 
With these constraints, such $\QQ$-matrix is also applicable to GPDINA, and we show that the two models are equivalent by presenting a bijective mapping from the item parameters of GPDINA to the parameters of the Sequential DINA model when the parameters are well-defined. Specifically, for each item $j \in [J]$, we have the following relation between the two models' parameters.\\
\textbf{From Sequential DINA model to GPDINA}: 
\begin{align*}
\begin{cases}
P(R_j = l \,|\,\xi_{j,\aaa} = 1) = \theta_{j,l}^+ = (1-\beta_{j,l+1}^+) \prod_{h=1}^l\beta_{j,h}^+,\; \;\textnormal{for}\;
    l \geq 1;\\
P(R_j = l \,|\,\xi_{j,\aaa} = 0) = \theta_{j,l}^- = (1-\beta_{j,l+1}^-) \prod_{h=1}^l\beta_{j,h}^-,\; \;\textnormal{for}\;
    l \geq 1.\\
\end{cases}
\end{align*}
\textbf{From GPDINA to Sequential DINA model}:
\begin{align*}
\begin{cases}
        \displaystyle 
        P(R_j \geq l \,|\, R_j \geq l-1, \xi_{j,\aaa} = 1) =  \beta_{j,l}^+ = \frac{\sum_{h = l}^{H_j} \theta_{j,h}^+}{\sum_{h = l-1}^{H_j} \theta_{j,h}^+},\; \;\textnormal{for}\;
    l \geq 1;\\
        \displaystyle P(R_j \geq l \,|\, R_j \geq l-1, \xi_{j,\aaa} = 0) = \beta_{j,l}^- = \frac{\sum_{h = l}^{H_j} \theta_{j,h}^-}{\sum_{h = l-1}^{H_j} \theta_{j,h}^-},\; \;\textnormal{for}\;
    l \geq 1.
\end{cases}
\end{align*}
By examining the above equations, it becomes apparent that there is a one-to-one correspondence between the parameters of the two models, demonstrating the equivalence of the two models under the considered $\QQ$-matrix constraints.

\section{Identifiability}\label{sec-id}
This section introduces our identifiability results for the GPDINA model and the Sequential DINA model. To provide a foundation for these results, we first generalize the $\TT$-matrix, a powerful tool in the literature to establish the identifiability of CDMs with binary responses \citep{JLGXZY2011,xu2016,xu2017}, to polytomous models in Section 3.1. Since the structure of the two polytomous models differ, the $\TT$-matrix generalizations also differ, and we provide examples to illustrate this. We then formally present our identifiability results for the two models in Sections 3.2 and 3.3, respectively. 

\subsection{Generalized $\TT$-matrix for CDMs with polytomous responses}\label{sec:tmat}
Directly working on the equations (\ref{eq-orig}) and (\ref{eq-orig-seq}) from the definitions of identifiability is challenging. Alternatively, we work on the marginal probability matrix, the $\TT$-matrix, firstly introduced by \citet{JLGXZY2011}, which sets up a linear dependence between attribute distribution and the response distribution. However, under the DINA model setting, most existing literature only focuses on the $\TT$-matrix for binary responses. 
For polytomous response DINA models, there are more parameters involved for each item, and these parameters can not be naively treated separately.
Our aim in this section is to generalize this powerful $\TT$-matrix tool to polytomous response models adjusted accordingly to the model setup.

\subsubsection{$\TT$-matrix for GPDINA}
The $\TT$-matrix for GPDINA $\TT(\ttt^+,\ttt^-)$ is a $\prod_{j=1}^J(H_j+1) \, \times \, 2^K $ matrix, where the entries are indexed by row index $\rr \in \mcs$ with $r_j \in \{0, 1, \ldots, H_j\}$ and column index $\aaa \in \{0, 1\}^K$. 
The $\rr$-th row and $\aaa$-th column entry of $\TT(\ttt^+,\ttt^-)$, denoted by $t_{\rr,\aaa}(\ttt^+,\ttt^-)$, is defined as
\[t_{\rr,\aaa}(\ttt^+,\ttt^-) = P\left(\underset{j: r_j \neq 0}{\bigcap} \{R_j = r_j\}\mid \QQ,\ttt^+,\ttt^-,\aaa\right) = \underset{j: r_j \neq 0}{\prod} P\left( R_j = r_j\mid \QQ,\ttt^+,\ttt^-,\aaa\right). 
\]
When $\rr=\mathbf 0$, 
$t_{\mz,\aaa}(\ttt^+,\ttt^-) = 1 \mbox{ for any } \aaa.$
When $\rr= r_j \cdot \ee_j$, $$t_{\rr,\aaa}(\ttt^+,\ttt^-) =P(R_j=r_j \mid \QQ,\ttt^+,\ttt^-,\aaa).$$  
Let $\TT_{\rr}( \ttt^+,\ttt^-)$ be the row vector in the $\TT$-matrix corresponding to $\rr$.
Then for any $\rr \neq \mz$, we can write
$\TT_{\rr}( \ttt^+,\ttt^-) = 
\underset{j: r_j\neq0}{\circ}  \TT_{r_j \cdot \ee_j}( \ttt^+,\ttt^-),$ where $\circ$ is the element-wise product of the row vectors. 
Since there exists a one-to-one mapping between $\TT_{\rr}$ and $P(\RR=\rr\mid \QQ,\ttt^+,\ttt^-,\pp)$ for all $\rr \in \mcs$, we may substitute the original identifiability problem with an equivalent statement as follows.
\begin{lemma}\label{lem:t-matrix}
Following the definition in (\ref{eq-orig})  and letting the attribute $\aaa$ index of $\pp$ be consistent with the $\aaa$ index in $\TT$, 
the GPDINA parameters are identifiable if and only if there is no $(\bar{\ttt}^+,\bar{\ttt}^-,\bar{\pp})\neq (\ttt^+,\ttt^-,\pp)$ such that 
\begin{equation}\label{lem:t-matrix-eq}
    \TT \pp = \bar{\TT}\bar{\pp}.
\end{equation}
\end{lemma}
To illustrate the construction of the $\TT$-matrix, we provide an example in the following.
\begin{example}
For the $\QQ$-matrix given in Example~\ref{ex1}, 
the $\TT$-matrix for this $\QQ$-matrix is
\begin{equation*}
\TT = 
    \begin{pmatrix}
 \aaa\;: & (0,0) & (1,0) & (0,1) & (1,1)\\
    \TT_{\rr = (0,0)}\;& 1 & 1 & 1 &1\\
    \TT_{\rr = (1,0)}& \theta_{1,1}^- & \theta_{1,1}^+ & \theta_{1,1}^- &\theta_{1,1}^+\;\\
    \TT_{\rr = (2,0)}& \theta_{1,2}^- & \theta_{1,2}^+ & \theta_{1,2}^- &\theta_{1,2}^+\;\\
    \TT_{\rr = (0,1)}\;& \theta_{2,1}^- &  \theta_{2,1}^-& \theta_{2,1}^+ &\theta_{2,1}^+\\
    \TT_{\rr = (0,2)}\;& \theta_{2,2}^- &  \theta_{2,2}^-& \theta_{2,2}^+ &\theta_{2,2}^+\\
    \TT_{\rr = (1,1)} \;& \theta_{1,1}^-\theta_{2,1}^- &  \theta_{1,1}^+\theta_{2,1}^-& \theta_{1,1}^-\theta_{2,1}^+ &\theta_{1,1}^+\theta_{2,1}^+\\
    \TT_{\rr = (2,1)} \;& \theta_{1,2}^-\theta_{2,1}^- &  \theta_{1,2}^+\theta_{2,1}^-& \theta_{1,2}^-\theta_{2,1}^+ &\theta_{1,2}^+\theta_{2,1}^+\\
    \TT_{\rr = (1,2)} \;& \theta_{1,1}^-\theta_{2,2}^- &  \theta_{1,1}^+\theta_{2,2}^-& \theta_{1,1}^-\theta_{2,2}^+ &\theta_{1,1}^+\theta_{2,2}^+\\
    \TT_{\rr = (2,2)} \;& \theta_{1,2}^-\theta_{2,2}^- &  \theta_{1,2}^+\theta_{2,2}^-& \theta_{1,2}^-\theta_{2,2}^+ &\theta_{1,2}^+\theta_{2,2}^+
\end{pmatrix},
\end{equation*}
where $\TT_{\rr = (1,1)} = \TT_{\rr = (1,0)} \circ \TT_{\rr = (0,1)}$, $\TT_{\rr = (2,1)} = \TT_{\rr = (2,0)} \circ \TT_{\rr = (0,1)}$, $\TT_{\rr = (1,2)} = \TT_{\rr = (1,0)} \circ \TT_{\rr = (0,2)}$, $\TT_{\rr = (2,2)} = \TT_{\rr = (2,0)} \circ \TT_{\rr = (0,2)}$. We can see that the $\TT$-matrix's structure is the same as the classic $\TT$-matrix for binary DINA model, where the entries of the $\TT$-matrix involve at most two parameters.
\end{example}

\subsubsection{$\TT$-matrix for Sequential DINA model}
Similarly, we generalize the $\TT$-matrix for the Sequential DINA model. However, due to the special structure of the Sequential DINA model, the generalization of the $\TT$-matrix here is slightly different from the literature, which we denote as $\TT^s$-matrix, where the ``s" stands for Sequential DINA model.
Let the entries of $\TT^s$-matrix $\TT^s(\bbb^+,\bbb^-)$
be indexed by row index $\rr \in \mcs$ and column index $\aaa \in \{0,1\}^K$. The $\rr$-th row and $\aaa$-th column entry of $\TT^s(\bbb^+,\bbb^-)$, denoted by $t_{\rr,\aaa}^s(\bbb^+,\bbb^-)$, is defined as
\begin{align*}
    t_{\rr,\aaa}^s(\bbb^+,\bbb^-) &= P\left(\underset{j: r_j \neq 0}{\bigcap} \{R_j \geq r_j \} \mid \QQ,\bbb^+,\bbb^-,\aaa\right)\\
    &= \underset{j: r_j \neq 0}{\prod} P\left( R_j \geq r_j\mid \QQ,\bbb^+,\bbb^-,\aaa\right) \\
    &= \underset{j: r_j \neq 0}{\prod} \prod_{l=1}^{r_j} (\beta_{j,l}^+)^{\xi_{j,l,\aaa}}(\beta_{j,l}^-)^{1-\xi_{j,l,\aaa}}.
\end{align*}
Apparently, $t_{\mz,\aaa}^s(\bbb^+,\bbb^-) = \one \mbox{ for any } \aaa.$
When $\rr= r_j \cdot \ee_j$,
\[\displaystyle t_{r_j \cdot \ee_j,\aaa}^s(\bbb^+,\bbb^-) =P(R_j \geq r_j \mid \QQ,\bbb^+,\bbb^-,\aaa) = \prod_{l=1}^{r_j} (\beta_{j,l}^+)^{\xi_{j,l,\aaa}}(\beta_{j,l}^-)^{1-\xi_{j,l,\aaa}} .\]
Let $\TT^s_{\rr}( \bbb^+,\bbb^-)$ be the row vector in the $\TT^s$-matrix corresponding to $\rr$.
Then for any $\rr \neq \mz$, we can write
$\TT^s_{\rr}( \bbb^+,\bbb^-) =
\underset{{j: r_j\neq0}}{\circ} \TT^s_{r_j \cdot \ee_j}( \bbb^+,\bbb^-)$. Similarly, due to the one-to-one mapping between $\TT_{\rr}^s$ and $P(\RR \geq \rr\mid \QQ,\ttt^+,\ttt^-,\pp)$ for all $\rr \in \mcs$, we may substitute the original identifiability problem using the $\TT^s$-matrix technique, we state this consequence in the following lemma.
\begin{lemma}\label{lem:ts-matrix}
Following the definition in (\ref{eq-orig-seq}) and  letting the attribute $\aaa$ index in $\pp$ be consistent with the $\aaa$ index in $\TT$, 
the Sequential DINA model parameters are identifiable if and only if there is no $(\bar{\bbb}^+,\bar{\bbb}^-,\bar{\pp})\neq (\bbb^+,\bbb^-,\pp)$ such that 
\begin{equation}\label{id-T}
    \TT^s \pp = \bar{\TT}^s\bar{\pp}.
\end{equation}
\end{lemma}
In the following, we present the $\TT^s$-matrix for the model given in Example~\ref{ex2}. Due to the unique structure of the Sequential DINA model, the $\TT^s$-matrix is designed in a very different way from a standard $\TT$-matrix for  the DINA model. 
\begin{example}
For the $\QQ$-matrix given in Example~\ref{ex2}, 
the $\TT^s$-matrix for this $\QQ$-matrix is
\[
\TT^s = 
\begin{pmatrix}
    \aaa\;: & (0,0) & (1,0) & (0,1) & (1,1)\\
    \TT^s_{\rr = (0,0)}\;& 1 & 1 & 1 &1\\
    \TT^s_{\rr = (1,0)}& \beta_{1,1}^- & \beta_{1,1}^+ & \beta_{1,1}^- &\beta_{1,1}^+\;\\
    \TT^s_{\rr = (2,0)}\;& \beta_{1,1}^-\beta_{1,2}^- & \beta_{1,1}^+ \beta_{1,2}^-& \beta_{1,1}^-\beta_{1,2}^+ &\beta_{1,1}^+\beta_{1,2}^+\\
    \TT^s_{\rr = (0,1)}\;& \beta_{2,1}^- &  \beta_{2,1}^-& \beta_{2,1}^+ &\beta_{2,1}^+\\
    \TT^s_{\rr = (0,2)}\;& \beta_{2,1}^- \beta_{2,2}^-&  \beta_{2,1}^-\beta_{2,2}^-& \beta_{2,1}^+\beta_{2,2}^- &\beta_{2,1}^+\beta_{2,2}^+\\
    \TT^s_{\rr = (1,1)}& \beta_{1,1}^-\beta_{2,1}^- & \beta_{1,1}^+\beta_{2,1}^- & \beta_{1,1}^-\beta_{2,1}^+ &\beta_{1,1}^+\beta_{2,1}^+\;\\
    \TT^s_{\rr = (2,1)}& \beta_{1,1}^-\beta_{1,2}^-\beta_{2,1}^- & \beta_{1,1}^+ \beta_{1,2}^-\beta_{2,1}^-& \beta_{1,1}^-\beta_{1,2}^+\beta_{2,1}^+ &\beta_{1,1}^+\beta_{1,2}^+\beta_{2,1}^+\\
    \TT^s_{\rr = (1,2)} & \beta_{1,1}^-\beta_{2,1}^- \beta_{2,2}^-&  \beta_{1,1}^+\beta_{2,1}^-\beta_{2,2}^-& \beta_{1,1}^-\beta_{2,1}^+\beta_{2,2}^- &\beta_{1,1}^+\beta_{2,1}^+\beta_{2,2}^+\\
    \TT^s_{\rr = (2,2)}\;& \beta_{1,1}^-\beta_{1,2}^-\beta_{2,1}^- \beta_{2,2}^-& \beta_{1,1}^+ \beta_{1,2}^-\beta_{2,1}^-\beta_{2,2}^-& \beta_{1,1}^-\beta_{1,2}^+\beta_{2,1}^+\beta_{2,2}^- &\beta_{1,1}^+\beta_{1,2}^+\beta_{2,1}^+\beta_{2,2}^+
\end{pmatrix}
\]
where $\TT^s_{\rr = (1,1)} = \TT^s_{\rr = (1,0)} \circ \TT^s_{\rr = (0,1)}$,\;\; $\TT^s_{\rr = (2,1)} = \TT^s_{\rr = (2,0)}\;\circ \TT^s_{\rr = (0,1)}$,
$\TT^s_{\rr = (1,2)} = \TT^s_{\rr = (1,0)} \circ \TT^s_{\rr = (0,2)}$, \;and \;$\TT^s_{\rr = (2,2)} = \TT^s_{\rr = (2,0)} \circ \TT^s_{\rr = (0,2)}$.
\end{example}
Unlike the $\TT$-matrix for GPDINA, the entries of the $\TT^s$-matrix for the Sequential DINA model usually involve more than two parameters, making identifying them  technically  more challenging. For instance, $\TT^s_{\rr = (2,2)}$ in the Sequential DINA model has four parameters in each entry, whereas $\TT_{\rr = (2,2)}$ in GPDINA only has two parameters in each entry. The following sections give  a more detailed discussion of the identifiability issue.

\subsection{Identifiability of the GPDINA model}
In this section, we develop the sufficient and necessary condition for the identifiability of the GPDINA model. 
To begin with, we introduce the terminology ``completeness" for $\QQ$-matrix, which was firstly proposed by \citet{chiu2009}. A $\QQ$-matrix is said to be complete if it can differentiate all latent attribute profiles. 
Under the DINA model with binary responses, it requires that for each attribute, there exists some item requiring solely that attribute, that is,  a complete $\QQ$-matrix must contain an identity matrix $\cI_K$ up to some row permutations, which can be written as 
\begin{equation}\label{eq-formq}
\QQ=\left(\begin{array}{c}
\cI_K \\
\QQ^*
\end{array}\right)_{J\times K}.
\end{equation}
Similar to the binary response case \citep{xu2016}, the completeness of the $\QQ$-matrix is necessary for the identifiability of the population proportion parameter $\pp$. Additionally, each attribute must be required by a certain amount of items, and formally we state these conditions as follows.
\renewcommand{\thecondition}{C\arabic{condition}}
\begin{condition}\label{C1}
    The $\QQ$-matrix must be complete, taking the form (\ref{eq-formq}).
\end{condition}
\begin{condition}\label{C2}
    Each of the $K$ attributes is required by at least three items.
\end{condition}
\begin{condition}\label{C3}
    Any two different columns of the sub-matrix $\QQ^*$ in (\ref{eq-formq}) are distinct.
\end{condition}

\begin{restatable}{theorem}{nesuGPDINA}\label{thm:nesu1}
Conditions C1-C3 are sufficient and necessary for the identifiability of the parameters of the GPDINA model.
\end{restatable}

\begin{remark}
When $H_j = 1$ for all $j \in [J]$, the model is reduced to binary DINA model, and the result we develop here is consistent with the result in \cite{id-dina}.
\end{remark}

\begin{remark}
 While the identifiability conditions are the same as those for the DINA model with binary responses \cite{id-dina}, we would like to emphasize several significant distinctions. 
In the case of the DINA model with binary responses, the uncertainty of each item is characterized by two parameters -- the slipping and guessing parameters. In contrast, the GPDINA model with polytomous responses introduces more than two parameters for each item, significantly complicating the models and rendering the study of identifiability more challenging.
In particular, as discussed in Section \ref{sec:tmat}, one crucial theoretical tool commonly employed in the literature to investigate the identifiability of the DINA model is the $\TT$-matrix, which is primarily designed for binary response models (Xu, 2027; Gu and Xu, 2019). However, when extending our focus to the polytomous response scenario such as the GPDINA model, it cannot be directly applied and a generalization of this tool becomes necessary. The first contribution of our work, detailed in Section 3.1, lies in this generalization, extending the applicability of these analytical techniques to a broader class of cognitive diagnosis models.
Moreover, with the newly developed $\TT$-matrix tool, significant efforts and new techniques are involved in the establishment of our new results. From the sufficient condition perspective, although conditions C1-C3 are also the counterparts of those of the DINA model, it is not immediately evident if these conditions, transposed from the binary model, are still capable of capturing the complexity and ensuring the identifiability of the more parameter-rich GPDINA model. Addtionally, from the necessary condition perspective, evaluating the necessity of conditions C1-C3 for the GPDINA model is more challenging than that of the DINA model with binary responses, due to the increased complexity of the GPDINA model, as illustrated in the following example and our proof of the theorem.  
\end{remark}

The  completeness of the  $\QQ$-matrix is necessary for the identifiability of the population proportion parameters, which follows from a similar argument as the binary DINA model \citep{id-dina}. See our proof in Supplementary Material for more details. To illustrate the necessity of the second condition C2 and the third condition C3, we consider a simple case when $K = 2$ in the following example. 
\begin{example}
We illustrate the necessity of the  conditions C2 and C3 with an example with $K = 2$. 
We first consider the necessity of the second condition.  
Suppose the $\QQ$-matrix is complete, but does not satisfy condition C2. i.e., there exists some attribute which is required by at most two items. Without loss of generality (WLOG), assume that this is the first attribute. According to Proposition~\ref{propGPDINA}, $\qq_{j} \neq \zero$ for all $j \in [J]$, so the $\QQ$-matrix can be written as one of the following: 
\begin{equation}\label{exC1}
\QQ =
\begin{pmatrix}
1 \quad 0  \\
\hdashline[2pt/2pt]
1\quad 0\\
\hdashline[2pt/2pt]
0\quad 1\\
\hdashline[2pt/2pt]
\vdots\quad \vdots\\
\hdashline[2pt/2pt]
0\quad 1\\
\end{pmatrix}_{J\times 2} \;\text{or}\quad
\QQ =
\begin{pmatrix}
1 \quad 0  \\
\hdashline[2pt/2pt]
1\quad 1\\
\hdashline[2pt/2pt]
0\quad 1\\
\hdashline[2pt/2pt]
\vdots\quad \vdots\\
\hdashline[2pt/2pt]
0\quad 1\\
\end{pmatrix}_{J\times 2},
 \end{equation}
 where the dashline ``- - -'' indicates the separation of different items.
For simplicity, we may assume that the $\QQ$-matrix takes the first formula. The case when the $\QQ$-matrix takes the other formula can be similarly obtained. So then only the first and the second item require $\alpha_1$.
Under this $\QQ$-matrix, we show that the model parameters $(\ttt^+,\ttt^-,\pp)$ are not identifiable by constructing a set of parameters   $(\bar\ttt^+,\bar\ttt^-,\bar\pp)\neq (\ttt^+,\ttt^-,\pp)$ which satisfy \eqref{eq-orig}. 
Take $\bar\ttt^+ = \ttt^+$ and $\bar\ttt^-_j = \ttt^-_j$ for $j>2$,
and $\bar p_{(11)} + \bar p_{(01)}=p_{(11)} + p_{(01)}$. Next we show that the remaining parameters $(\ttt^-_1, \ttt^-_2, p_{(00)}, p_{(01)}, p_{(10)})$ are not identifiable. 
Using the $\TT$-matrix tool, it can be shown that the non-identifiability occurs if the following equations hold:
$P\big((R_1,R_2)=(r_1,r_2)\mid \QQ, \bar\ttt^+,\bar\ttt^-,\bar\pp\big) =  P\big((R_1,R_2)=(r_1,r_2)\mid \QQ, \ttt^+,\ttt^-,\pp\big)$ for all $r_1 \in \{0, 1, \ldots H_1\}$, $r_2 \in \{0, 1, \ldots, H_2\}$,
where $(R_1,R_2)$ are the first two entries of the random response $\RR$.
These equations can be further expressed as the following equations: 
\begin{equation}\label{aa2}
\resizebox{1\hsize}{!}{$
\begin{cases}
\bar p_{(00)} + \bar p_{(10)} + p_{(01)} + p_{(11)}
    = p_{(00)} + p_{(10)} + p_{(01)}+ p_{(11)};\\
\bar \theta_{1,l_1}^-[\bar p_{(00)} +  \bar p_{(01)}] + \theta_{1,l_1}^+[\bar p_{(10)} +   \bar p_{(11)}] 
    = \theta_{1,l_1}^-[ p_{(00)} +   p_{(01)}] + \theta_{1,l_1}^+[ p_{(10)} +   p_{(11)}];\\
\bar \theta_{2,l_2}^-[\bar p_{(00)} +  \bar p_{(01)}] + \theta_{2,l_2}^+[\bar p_{(10)} +  \bar p_{(11)}] 
    = \theta_{2,l_2}^-[ p_{(00)} +   p_{(01)}] + \theta_{2,l_2}^+[ p_{(10)} +   p_{(11)}];\\
\bar\theta_{1,l_1}^-\bar \theta_{2,l_2}^-[\bar p_{(00)} +  \bar p_{(01)}] +  \theta_{1,l_1}^+\theta_{2,l_2}^+[\bar p_{(10)} +  \bar p_{(11)}]
    =  \theta_{1,l_1}^-\theta_{2,l_2}^-[ p_{(00)} +   p_{(01)}] +  \theta_{1,l_1}^+\theta_{2,l_2}^+[ p_{(10)} +   p_{(11)}];
\end{cases}$}
\end{equation}
where $l_1 \in [H_1]$, $l_2 \in [H_2]$. Then there are $(1+H_1+H_2+H_1H_2)$ equations above in total.
If we further let some $\kappa \in (0,1)$ s.t.
\begin{align}\label{kappa}
\begin{pmatrix}
    \theta_{1, l_1}^-\\
    \bar{\theta}_{1, l_1}^-\\
    \theta_{1, l_1}^+\\
    \bar{\theta}_{1, l_1}^+
\end{pmatrix} = \kappa^{l_1-1}
\begin{pmatrix}
    \theta_{1, 1}^-\\
    \bar{\theta}_{1, 1}^-\\
    \theta_{1, 1}^+\\
    \bar{\theta}_{1, 1}^+
\end{pmatrix}, \;\text{and}\;
\begin{pmatrix}
    \theta_{2, l_2}^-\\
    \bar{\theta}_{2, l_2}^-\\
    \theta_{2, l_2}^+\\
    \bar{\theta}_{2, l_2}^+
\end{pmatrix} = \kappa^{l_2-1}
\begin{pmatrix}
    \theta_{2, 1}^-\\
    \bar{\theta}_{2, 1}^-\\
    \theta_{2, 1}^+\\
    \bar{\theta}_{2, 1}^+
\end{pmatrix} \;\;\;\text{for}\;\;l_1 \in [H_1],\;l_2 \in [H_2],
\end{align} 
then equations (\ref{aa2})  can be reduced to four equations
\begin{equation}\label{aa3}
\resizebox{1\hsize}{!}{$
\begin{cases}
 \bar p_{(00)} + \bar p_{(10)} + p_{(01)} + p_{(11)}
    = p_{(00)} + p_{(10)} + p_{(01)}+ p_{(11)};\\
\bar \theta_{1,1}^-[\bar p_{(00)} +  \bar p_{(01)}] + \theta_{1,1}^+[\bar p_{(10)} +   \bar p_{(11)}] 
    = \theta_{1,1}^-[ p_{(00)} +   p_{(01)}] + \theta_{1,1}^+[ p_{(10)} +   p_{(11)}];\\
\bar \theta_{2,1}^-[\bar p_{(00)} +  \bar p_{(01)}] + \theta_{2,1}^+[\bar p_{(10)} +  \bar p_{(11)}] 
    = \theta_{2,1}^-[ p_{(00)} +   p_{(01)}] + \theta_{2,1}^+[ p_{(10)} +   p_{(11)}];\\
\bar\theta_{1,1}^-\bar \theta_{2,1}^-[\bar p_{(00)} +  \bar p_{(01)}] +  \theta_{1,1}^+\theta_{2,1}^+[\bar p_{(10)} +  \bar p_{(11)}]
    =  \theta_{1,1}^-\theta_{2,1}^-[ p_{(00)} +   p_{(01)}] +  \theta_{1,1}^+\theta_{2,1}^+[ p_{(10)} +   p_{(11)}].\\
\end{cases}
$}
\end{equation}
\doublespacing
For any  $(\ttt^+,\ttt^-,\pp)$, there are 4 constraints in \eqref{aa3}
but   5 parameters 	$(\bar \theta_{1,1}^-,\bar \theta_{2,1}^-,\bar p_{(00)}, \bar p_{(10)},\bar p_{(01)})$ to solve. Therefore there are infinitely many solutions and $(\ttt^+,\ttt^-,\pp)$ are non-identifiable. 
As for the case when the $\QQ$-matrix takes the other formula, the proof can be easily obtained with minor change of notation.

Next we prove the necessity of the third condition C3. Suppose the $\QQ$-matrix is complete, according to Proposition~\ref{propGPDINA}, we may assume that
the $\QQ$-matrix has the following form up to some permutation:
\begin{equation}\label{exC1}
\QQ =\begin{pmatrix}
 {\cI}_2  \\
\one_{J-2} \quad \one_{J-2}
\end{pmatrix}_{J\times 2}.
 \end{equation}
 Take $\bar\ttt^+ = \ttt^+$ and $\bar\ttt^-_j = \ttt^-_j$ for $j>2$,
and $\bar p_{(11)}=p_{(11)}$. Next we show the remaining parameters $(\ttt^-_1, \ttt^-_2, p_{(00)}, p_{(10)},p_{(01)})$ are not identifiable. 
Using the $\TT$-matrix tool, again we can show that the non-identifiability occurs if the following equations hold:
$P\big((R_1,R_2)=(r_1,r_2)\mid \QQ, \bar\ttt^+,\bar\ttt^-,\bar\pp\big) =  P\big((R_1,R_2)=(r_1,r_2)\mid \QQ, \ttt^+,\ttt^-,\pp\big)$ for all $r_1 \in \{0, 1, \ldots H_1\}$, $r_2 \in \{0, 1, \ldots, H_2\}$,
where $(R_1,R_2)$ are the first two entries of the random response $\RR$.
These equations can be further expressed 
into $(1+H_1+H_2+H_1H_2)$ equations similar to equations~(\ref{aa2}) with minor notation modification.
Similarly, if we further let some $\kappa \in (0,1)$ s.t. equations~(\ref{kappa}) hold, then these equations can be reduced to only four equations.
\doublespacing

For any  $(\ttt^+,\ttt^-,\pp)$, there are four constraints 
but five parameters 	$(\bar \theta_{1,1}^-,\bar \theta_{2,1}^-,\bar p_{(00)}, \bar p_{(10)},\bar p_{(01)})$ to solve. Therefore there are infinitely many solutions and $(\ttt^+,\ttt^-,\pp)$ are non-identifiable.
Thus we have shown that that the conditions C2 and C3 are indeed necessary.
For the proofs of more general cases and the sufficiency of the conditions, see   Supplementary Material for more details.
\end{example}

\subsection{Identifiability of the Sequential DINA model}\label{sec:seq-id}
To study the identifiability of the Sequential DINA model, different techniques need to be developed.
From the discussion in  Sections 3.1 and 3.2, the structure of the $\TT^s$-matrix for the Sequential DINA model is different from the $\TT$-matrix defined for the GPDINA model, since the rows of the $\TT^s$-matrix of Sequential DINA corresponding to higher response categories generally involve more than two item parameters, making it different from the usual DINA model structure. 

To address this issue,  note that
\begin{equation}\label{eq:T-ej}
    \TT_{\ee_j}^s = P( R_j \geq 1\mid \QQ,\bbb^+,\bbb^-,\aaa) = (\beta_{j,1}^+)^{\xi_{j,1,\aaa}} (\beta_{j,1}^-)^{1-\xi_{j,1,\aaa}}
\end{equation}
only involves two parameters $\beta_{j,1}^+$ and $\beta_{j,1}^-$, which has a similar algebraic structure to  that for the DINA model with binary responses, and thus working on these parameters firstly would be a good strategy to consider. 
The focus of these quantities can be interpreted as follows:
consider ``binary" responses  of the form $I(\text{item j} \geq 1)$, the Sequential DINA model is then reduced to a binary DINA model. According to equation (\ref{eq:T-ej}), the uncertainty parameters for this model are $\{\beta_{j,1}^+,\;\beta_{j,1}^-\}_{j \in [J]}$. The corresponding $\TT$-matrix for this reduced model consists of exactly vectors $\left(\TT^s_{\ee_j}\right)_{j \in [J]}$ and their element-wise products. That is, let $\TT^1$ denote the $\TT$-matrix for the reduced model (here we compress the notation ``s" in $\TT^s$), which is a submatrix of $\TT^s$-matrix, then
\[
\TT^1 = \Bigg(\overset{i_j}{ \underset{l=i_1}{\circ}}
\TT^s_{\ee_l}\Bigg) \;\text{ for }\;{i_1 < \ldots <i_j,\;j \in [J]}.
\]
Furthermore,   let $\QQ^1$ denote the submatrix of the $\QQ$-matrix for the first category of each item, i.e., $\QQ^1 = (\qq_{j,1})_{j \in J}$. Then the $\QQ$-matrix for the above reduced model is $\QQ^1$, as only the attributes required for completing the first categories are in scope.
For notation convenience, we  let $\QQ^1_{1:K}$ denote the submatrix of the $\QQ^1$-matrix that consists of the $\qq$-vectors for the first categories of the first $K$ items,  and $\QQ^1_{K+1:J}$ denote the submatrix of $\QQ^1$ that consists of the $\qq$-vectors for the first categories of items $(K+1), \ldots, J$, i.e., 
\[
\QQ^1 = 
\begin{pmatrix}
    \QQ^1_{1:K}\\
    \QQ^1_{K+1:J}
\end{pmatrix}.
\]
To better illustrate this idea, we present an example in the following.
\begin{example}
The $\QQ^1$-matrix and the $\TT^1$-matrix for the reduced model of Example~\ref{ex2} are:
\begin{equation*}
    \QQ^1 = 
    \begin{pmatrix}
        1 & 0\\
        \hdashline[2pt/2pt]
        0 &1 
    \end{pmatrix},\;\;\;
    \TT^1 = 
    \begin{pmatrix}
            \aaa\;: & (0,0) & (1,0) & (0,1) & (1,1)\\
    \TT^1_{\rr = (0,0)}\;& 1 & 1 & 1 &1\\
    \TT^1_{\rr = (1,0)}& \beta_{1,1}^- & \beta_{1,1}^+ & \beta_{1,1}^- &\beta_{1,1}^+\;\\
    \TT^1_{\rr = (0,1)}\;& \beta_{2,1}^- &  \beta_{2,1}^-& \beta_{2,1}^+ &\beta_{2,1}^+\\
    \TT^1_{\rr = (1,1)}& \beta_{1,1}^-\beta_{2,1}^- & \beta_{1,1}^+\beta_{2,1}^- & \beta_{1,1}^-\beta_{2,1}^+ &\beta_{1,1}^+\beta_{2,1}^+\;
    \end{pmatrix},
\end{equation*}
where $\TT^1_{\rr = (1,1)} = \TT^1_{\rr = (1,0)} \circ \TT^1_{\rr = (0,1)}$.
\end{example}
It turns out that the first category of each item plays a crucial role in the identifiability of the Sequential DINA model.
Provided the first categories of the items are informative enough, based on the identifiability results for  the DINA model with binary responses, we can identify the item parameters of the first categories and the population proportion parameters. 
More interestingly, we can show that the item parameters of the other categories can be identified based on these identified parameters without additional requirements.
Motivated by this, we have
the following sufficient condition for the identifiability of the Sequential DINA model.
\begin{restatable}{theorem}{suSeq}\label{suSeq}
The Sequential DINA model parameters are identifiable if the $\QQ^1$ matrix satisfies the following conditions S1-S3.
\renewcommand{\thecondition}{S\arabic{condition}}
\begin{condition}\label{cond-s1}
    $\QQ^1$-matrix is complete, i.e., under some permutation, $\QQ^1_{1:K} = \cI_K$.
\end{condition}
\begin{condition}\label{cond-s2}
    Each of the K attributes is required by at least three items’ first categories.
\end{condition}
\begin{condition}\label{cond-s3}
    Suppose $\QQ^1_{1:K} = \cI_K$, then any two different columns of $\QQ^1_{K+1:J}$  are distinct.
\end{condition}
\end{restatable}

\begin{remark}
    Conditions S1-S3 are similar to conditions  C1-C3, with different target. S1-S3 are stated for $\QQ^1$-matrix in the Sequential DINA model, whereas C1-C3 are stated for $\QQ$-matrix in GPDINA. When $H_j \equiv 1$, both polytomous models are reduced to binary DINA model, and conditions C1-C3 are  equivalent to S1-S3.
\end{remark}
The conditions S1-S3, as sufficient conditions for identifying the Sequential DINA model, provide guidelines for practitioners to design $\QQ$-matrix that validates identifiability. Based on the theorem, it is suggested to design $\QQ$-matrix with informative first categories (satisfying S1-S3) to ensure identifiability.

On the other hand, sufficient these conditions are, their requirements only rely on the model's first categories. With polytomous response data involving more categories, it is natural to ask whether other categories can aid in relaxing these conditions. 
It turns out that relaxing these conditions necessitates careful consideration. In the following, we examine the necessity of each condition, and our primary finding is that while these conditions are challenging to relax, with certain constraints that allow for other informative categories to help, they might be possible to be relaxed.
The finding that these conditions are challenging to relax comes from the intrinsic sequential structure of the model. Specifically, we will show that condition S1 can not be relaxed and conditions S2 and S3 are hard to relax as non-identifiable examples do exist with the absence of these conditions.

Our first claim is that without additional constraints, the first condition S1 can not be relaxed, i.e., S1 is  necessary.
\begin{restatable}[Necessity of Condition S1]{proposition}{neSeq1}\label{neSeq1}
    Condition S1 is necessary for the identifiability of the parameters of the Sequential DINA model.
\end{restatable}
\noindent For the convenience of the following discussion, we present the proof of Proposition~\ref{neSeq1} below.
\begin{proof}[Proof of Proposition~\ref{neSeq1}]
Suppose that the $\QQ$-matrix does not satisfy condition S1, i.e., $\QQ^1$ is not complete, then there exists some attribute  that is not solely required by any item's first category.  WLOG, assume that this is the first attribute, and thus any item's first category that requires the first attribute  also requires some other attributes. 
We claim that the model parameters are not identifiable for such an incomplete $\QQ^1$-matrix. Specifically, take $\beta_{j,1}^- \equiv 0$, for $j \in [J]$. Then subjects with attribute profiles $\zero$ and $\ee_1$ are not able to complete the first categories of all the items. Since $\beta_{j,1}^- \equiv 0$, according to the model construction in Section~\ref{sec:seq}, subjects with  $\zero$ and $\ee_1$ cannot complete the other categories either, and
 for  $\zero$ and $\ee_1$,
 $\beta_{j,l}^+ \equiv \beta_{j,l}^- \equiv 0$ for $l > 1$.
Therefore, the two profiles $p_{\zero}$ and $p_{\ee_1}$
share the same probability of completing all the categories of all the items, which is zero, i.e.,
$
t_{\rr, \zero} = t_{\rr, \ee_1} \equiv 0, \forall \rr.
$
Thus, 
parameters $p_{\zero}$ and $p_{\ee_1}$ are not identifiable.
\end{proof}
In the above proof, we constructed a Sequential DINA model with $\beta_{j,1}^- \equiv 0$ so that the parameters of higher categories are defined to be zero for attribute profiles $\zero$ and $\ee_1$. 
Note that the identifiability definition requires any set of the parameters in the parameter space to be identifiable. With the model parameters space including  $0 \leq \beta_{j,l}^- <\beta_{j,l}^+ \leq 1$,   in the proof of Proposition 3, showing the nonidentifiability of the case \(\beta_{j, 1}^{-} = 0\) would be enough to establish our claim on the necessity of the completeness condition.

However, 
this example is tender and may no longer be valid if we add additional constraints for the model parameters, that is, we only focus on the identifiability of a subset of the model parameters space.
 For instance,
if we restrict our model parameters to the subset $0 < \beta_{j,l}^- < \beta_{j,l}^+ \leq 1$, then the necessity of S1 may not hold anymore.
This is because by constraining $0 < \beta_{j,l}^- < \beta_{j,l}^+ \leq 1$, we allow more categories to help identifying the parameters.
The following gives an example of the model with identifiable parameters whose $\QQ$-matrix does not satisfy condition S1 under the assumption that $0 < \beta_{j,l}^- < \beta_{j,l}^+ \leq 1$.

\begin{example}\label{counterex1}
Assume that $0 < \beta_{j,l}^- < \beta_{j,l}^+ \leq 1$, and
consider the case when $K = 2$ where the $\QQ$-matrix takes the following form: 
\begin{equation}
    \QQ = \begin{pmatrix}
    item 1 \begin{cases}
        1 & 1\\
        0 & 1
    \end{cases}\\
        \hdashline[2pt/2pt]
        item 2 \begin{cases}
        1 & 1\\
        \end{cases}\\
        \hdashline[2pt/2pt]
        item 3 \begin{cases}
        1 & 1\\
        \end{cases}\\
        \hdashline[2pt/2pt]
        item 4 \begin{cases}
        1 & 0\\
        \end{cases}\\
        \hdashline[2pt/2pt]
        item 5 \begin{cases}
        1 & 0\\
        \end{cases}\\
        \hdashline[2pt/2pt]
        item 6 \begin{cases}
        1 & 0\\
        \end{cases}\\
\end{pmatrix} \text{and}\;\;
    \QQ^1 = \begin{pmatrix}
         1 & 1\\
        \hdashline[2pt/2pt]
        1 & 1\\
        \hdashline[2pt/2pt]
        1 & 1\\
        \hdashline[2pt/2pt]
        1 & 0\\
        \hdashline[2pt/2pt]
        1 & 0\\
        \hdashline[2pt/2pt]
        1 & 0\\
    \end{pmatrix}.
\end{equation}
Clearly, the $\QQ^1$-matrix does not satisfy the completeness condition, but the model parameters with this $\QQ$-matrix are identifiable, whose proof is presented in the Supplementary Material.
\end{example}

\begin{remark}\label{rm:strict id}
    Through the above analysis, we can see that condition S1 is necessary in a strict sense, which may impose overly stringent requirements for practical cognitive diagnostic tests.  Statistically, ``strict sense" in this context refers to the  standard identifiability definition of the model parameters that requires any set of the parameters in the parameter space to be identifiable \citep{partial}. Contrary to the notion of strict identifiability is the notion of generic identifiability \citep{allman2009,partial}, where we allow for non-identifiability to happen within a zero-measure set.
    This slightly weaker notion can often suffice for real data analysis purposes \citep*{allman2009, partial} and is therefore useful in practice. The extent to which our necessary conditions can be relaxed for generic identifiability of the Sequential DINA model needs further explorations in the future, and the above case with $\beta_{j,l}^- = 0$ in the Sequential DINA model is one of such example. 
\end{remark}
Next we study the necessity of conditions S2 and S3. It turns out that the analysis for conditions S2 and S3 is more complicated. We start by presenting two examples to illustrate that.

\begin{example}\label{notid1}
Consider the case when $K = 2$ with two attributes $\alpha_1$ and $\alpha_2$, $J = 4$ items, and the Q-matrix takes the following form:
\begin{equation}
    \QQ = 
    \begin{pmatrix}
    item \; 1
    \begin{cases}
    1 & 0\\   
    \end{cases}\\
    \hdashline[2pt/2pt]
    item \; 2
    \begin{cases}
    0 & 1\\       
    \end{cases}\\
    \hdashline[2pt/2pt]
    item \; 3
    \begin{cases}
    0 & 1\\        
    \end{cases}\\
    \hdashline[2pt/2pt]
    item \; 4
    \begin{cases}
    1 & 1\\     
        1 & 0\\
    \end{cases}
    \end{pmatrix} \;\text{and}\;\;\QQ^1 = \begin{pmatrix}
    1 & 0\\
    \hdashline[2pt/2pt]
    0 & 1\\
    \hdashline[2pt/2pt]
    0 & 1\\
    \hdashline[2pt/2pt]
    1 & 1
    \end{pmatrix}. 
\end{equation}
The above $\QQ$-matrix satisfies conditions S1 and S3, but does not satisfy condition S2, and the model parameters are not identifiable.
\end{example}

\begin{example}\label{notid2}
Consider the case when $K = 2$ with two attributes $\alpha_1$ and $\alpha_2$, $J = 4$ items, and the Q-matrix takes the following form:
\begin{equation}
    \QQ = 
    \begin{pmatrix}
    item \; 1
    \begin{cases}
    1 & 0\\   
    \end{cases}\\
    \hdashline[2pt/2pt]
    item \; 2
    \begin{cases}
    0 & 1\\       
    \end{cases}\\
    \hdashline[2pt/2pt]
    item \; 3
    \begin{cases}
    1 & 1\\        
    1 & 0\\
    \end{cases}\\
    \hdashline[2pt/2pt]
    item \; 4
    \begin{cases}
    1 & 1\\        
    \end{cases}
    \end{pmatrix} \;\text{and}\;\;\QQ^1 = \begin{pmatrix}
    1 & 0\\
    \hdashline[2pt/2pt]
    0 & 1\\
    \hdashline[2pt/2pt]
    1 & 1\\
    \hdashline[2pt/2pt]
    1 & 1
    \end{pmatrix}. 
\end{equation}
The above $\QQ$-matrix satisfies conditions S1 and S2, but does not satisfy condition S3, and the model parameters are not identifiable.
\end{example}
We defer the proofs of the non-identifiability of the above two examples in Supplementary Material.
The preceding examples illustrate the difficulty in relaxing conditions S2 and S3, even in simple cases such as $J=4$ and $K=2$, where non-identifiable examples exist when these conditions are violated. 
For more general cases, relaxing these conditions could be even more challenging.

However, the existence of these examples does not necessarily mean that conditions S2 and S3 are always necessary. In fact, we construct two identifiable examples that do not satisfy conditions S2 and S3 in the following, which indicates that conditions S2 and S3 may not be necessary in general. The identifiability of the following two examples relies on other additional categories, which carry relevant information in place of the first categories. This is also aligned with intuition, as we expect other categories to contribute to the identification of the model parameters. In other words, with the help of other categories, the model parameters could possibly be identified.

\begin{example}\label{counterex2}
Consider the case when $K=2$ with two attributes $\alpha_1$ and $\alpha_2$, and $J=4$ items. Each item contains two categories and the $\QQ$-matrix takes the following form:
\begin{equation}
    \QQ = \begin{pmatrix}
    item \;1
    \begin{cases}
        1 & 0\\
    0 & 1\\
    \end{cases}\\
    \hdashline[2pt/2pt]
    item \;2
    \begin{cases}
    1 & 0\\
    0 & 1\\        
    \end{cases}\\
    \hdashline[2pt/2pt]
    item \;3
    \begin{cases}
    0 & 1\\
    1 & 0\\
    \end{cases}\\
    \hdashline[2pt/2pt]
    item \; 4
    \begin{cases}
    0 & 1\\
    1 & 0\\        
    \end{cases}
    \end{pmatrix} \;\text{and}\;\;\QQ^1 = \begin{pmatrix}
    1 & 0\\
    \hdashline[2pt/2pt]
    1 & 0\\
    \hdashline[2pt/2pt]
    0 & 1\\
    \hdashline[2pt/2pt]
    0 & 1\\
    \end{pmatrix}. 
\end{equation}
The above $\QQ^1$ matrix does not satisfy the condition S2, yet the model parameters are identifiable, whose proof is deferred to the Supplementary Material.
\end{example}
\begin{remark}
    Condition S2 assumes each attribute is required by three items' first categories. In the above example, both attributes $\alpha_1$ and $\alpha_2$ are required by only two items' first categories, yet the two attributes are also required by the second categories of other items, which provides additional information and eventually makes the model parameters identifiable. This suggests that the information provided by higher categories would also be helpful for the  model identifiability.  
\end{remark}

Similarly, as illustrated in the following example, the role of the first category in condition S3 could also be replaced by other categories, which may make the model identifiable as well.

\begin{example}\label{counterex3}
Consider the case when $K=2$ with two attributes $\alpha_1$ and $\alpha_2$, and $J=5$ items, and the $\QQ$-matrix takes the following form:
\begin{equation*}
    \QQ = 
    \begin{pmatrix}
    item \; 1
    \begin{cases}
    1 & 0\\
    0 & 1\\        
    \end{cases}\\
    \hdashline[2pt/2pt]
    item \; 2
    \begin{cases}
    0 & 1\\
    1 & 0\\        
    \end{cases}\\
    \hdashline[2pt/2pt]
    item \; 3
    \begin{cases}
    1 & 1\\        
    \end{cases}\\
    \hdashline[2pt/2pt]
    item \; 4
    \begin{cases}
    1 & 1\\        
    \end{cases}\\
    \hdashline[2pt/2pt]
    item \; 5
    \begin{cases}
    1 & 1\\        
    \end{cases}
    \end{pmatrix} \;\text{and}\;\;\QQ^1 = \begin{pmatrix}
    1 & 0\\
    \hdashline[2pt/2pt]
    0 & 1\\
    \hdashline[2pt/2pt]
    1 & 1\\
    \hdashline[2pt/2pt]
    1 & 1\\
    \hdashline[2pt/2pt]
    1 & 1\\
    \end{pmatrix}. 
\end{equation*}
The above $\QQ^1$ matrix does not satisfy the condition S3, yet the model parameters are identifiable, whose proof is presented in the Supplementary Material.
\end{example}

While the above two examples imply that the conditions S2 and S3 may not be necessary for the identifiability of the parameters for the Sequential DINA model,   the following weaker versions of S2 and S3 (denoted as conditions S2$^*$ and S3$^*$)  are necessary  for the model identifiability. This proposition is summarized as follows.

\begin{restatable}[Necessity of Conditions S2$^*$ and S3$^*$]{proposition}{neSeq}\label{neSeq}
The Sequential DINA model parameters are identifiable only if the $\QQ$-matrix satisfies the following conditions S2$^*$ and S3$^*$.

\noindent{\bf Condition S2$^*$ } Each of the K attributes is required by at least three categories (not necessarily the first categories), and the three categories must come from at least two different items. 

\noindent{\bf Condition S3$^*$ } Suppose $\QQ$-matrix satisfies S1, i.e., $\QQ_{1:K}^1 = \cI_K$, and
    any two different columns of the following matrix
    (which removes the identity matrix of $\QQ_{1:K}^1$ from $\QQ$)
    \[
    \begin{pmatrix}
        \QQ_{1:K}^{-1}\\
        \QQ_{K+1:J}
        \end{pmatrix}
    \]
    are distinct,
    where $ \QQ_{1:K}^{-1}$ denotes the remaining  submatrix of $\QQ_{1:K}$ after removing $\QQ_{1:K}^1$.
\end{restatable}

We can see that conditions S2 and S3 are stronger versions of S2$^*$ and S3$^*$, which means that any $\QQ$-matrix satisfying condition S2 (S3) will satisfy condition S2$^*$ (S3$^*$). We can also see that the two identifiable models in Example~\ref{counterex2} and Example~\ref{counterex3} that do not satisfy conditions S2 and S3 both satisfy condition S2$^*$ and condition S3$^*$. For instance, the $\QQ$-matrix in Example~\ref{counterex2}, does not satisfy condition S2 since there are only two items' first categories require $\alpha_1$ and only two items' first categories require $\alpha_2$.
However, it does satisfy condition S2*, since there are two other items' second categories require $\alpha_1$ and other two items' second categories require $\alpha_2$. Similarly, the $\QQ$-matrix in Example~\ref{counterex3}, not satisfying condition S3, does satisfy condition S3*, as the second category of the first item requires only $\alpha_2$ and the second category of the second item requires only $\alpha_1$.

In summary, from the above discussions, we conclude that the sufficient conditions S1-S3 are challenging to relax. Specifically, condition S1 can not be relaxed unless additional constraints are imposed. While conditions S2 and S3 are also difficult to relax, we found that other categories may assist in identifying the parameters.

In spite of the fact that the sufficient condition and the necessary condition   proposed in this section are different, filling the gap is not an easy task, as the model structure is more subtle and the interactions between parameters are more complex. For instance, the $\TT^s$-matrix structure is different from the $\TT$-matrix structure for the binary DINA model except for the first categories. The $\TT^s_{\rr}$-vectors for higher categories behave more similar to the $\TT_{\rr}$-vectors for G-DINA model \citep{de2011generalized}, as the uncertainty for these categories are characterized by more than two parameters. Therefore, to study the identifiability of the Sequential DINA model requires more techniques beyond the DINA setting.

 \section{ Data Examples}
\label{sec-app}
In this section, we demonstrate the application of our proposed results by examining two   educational assessment datasets: a PISA 2000 reading assessment dataset using the GPDINA model  \citep{GPDINA} and a TIMSS 2007 fourth-grade mathematics assessment dataset using the Sequential DINA model \citep{sequential}.


\paragraph{Identifiability of the GPDINA model: a PISA 2000 data example.}

We consider a dataset from the PISA 2000 reading assessment, which was previously studied in \cite{GPDINA}. This assessment, released by the \citet*{oecd1999, oecd2006}, comprised both polytomous and binary items. 
 The dataset for this application comprises responses from 1,039 English examinees to 20 specific items from a designated test booklet. Out of these 20 items, five are polytomous.
Following \cite{GPDINA}, the attribute definitions for the PISA dataset are given in Table 1 and the
  $\QQ$-matrix for this application is presented in Table 2. Since in the GPDINA model, different categories within the same item share the same $\qq$-vectors, it suffices to provide one $\qq$-vector for each item.

\begin{table}[h]
\centering
\caption{Attribute Definitions for the PISA data \citep{GPDINA}}
\begin{tabular}{|c|p{10cm}|}
\hline
Symbol & Description \\
\hline
$c$ & Number of categories \\
\hline
$\alpha_1$ & Retrieving information \\
\hline
$\alpha_2$ & Forming a broad general understanding \\
\hline
$\alpha_3$ & Developing an interpretation \\
\hline
$\alpha_4$ & Reflecting on and evaluating the content of a text \\
\hline
$\alpha_5$ & Reflecting on and evaluating the form of a text \\
\hline
\end{tabular}
\end{table}

\begin{table}[h]
\centering
\caption{Items and $\QQ$-matrix for the PISA data \citep{GPDINA}}
\begin{tabular}{|c|c|c|c|c|c|c|c|c|c|c|c|c|c|c|c|}
\hline No. & Item Code & \( c \) & \( \alpha_1 \) & \( \alpha_2 \) & \( \alpha_3 \) & \( \alpha_4 \) & \( \alpha_5 \) & No. & Item Code & \( c \) & \( \alpha_1 \) & \( \alpha_2 \) & \( \alpha_3 \) & \( \alpha_4 \) & \( \alpha_5 \) \\
\hline 1 & R040Q02 & 2 & 1 & 0 & 1 & 0 & 0 & 11 & R088Q04T & 3 & 1 & 0 & 1 & 0 & 0 \\
\hline 2 & R040Q03A & 2 & 1 & 0 & 1 & 1 & 0 & 12 & R088Q05T & 2 & 0 & 1 & 1 & 1 & 0 \\
\hline 3 & R040Q04 & 2 & 0 & 1 & 1 & 1 & 0 & 13 & R088Q07 & 2 & 0 & 1 & 0 & 0 & 1 \\
\hline 4 & R040Q06 & 2 & 1 & 0 & 1 & 0 & 0 & 14 & R216Q01 & 2 & 0 & 1 & 0 & 0 & 0 \\
\hline 5 & R077Q03 & 3 & 0 & 1 & 0 & 1 & 1 & 15 & R216Q02 & 2 & 1 & 0 & 0 & 0 & 1 \\
\hline 6 & R077Q04 & 2 & 1 & 1 & 1 & 0 & 0 & 16 & R216Q03T & 2 & 0 & 1 & 1 & 0 & 0 \\
\hline 7 & R077Q05 & 3 & 0 & 1 & 1 & 1 & 0 & 17 & R216Q04 & 2 & 0 & 1 & 1 & 0 & 0 \\
\hline 8 & R077Q06 & 2 & 0 & 1 & 0 & 0 & 1 & 18 & R216Q06 & 2 & 0 & 1 & 0 & 1 & 0 \\
\hline 9 & R088Q01 & 2 & 0 & 1 & 1 & 0 & 0 & 19 & R236Q01 & 2 & 1 & 0 & 1 & 0 & 0 \\
\hline 10 & R088Q03 & 3 & 1 & 0 & 1 & 0 & 0 & 20 & R236Q02 & 3 & 0 & 0 & 1 & 1 & 0 \\
\hline
\end{tabular}
\footnotesize{}
\end{table}

According to our Theorem~\ref{thm:nesu1}, this $\QQ$-matrix does not contain an identity matrix, and thus the model parameters are not identifiable. 
Specifically, since the matrix does not contain $\ee_1^\top$, $\ee_3^\top$, $\ee_4^\top$ and $\ee_5^\top$, attribute profiles $\zero$, $\ee_1$, $\ee_3$, $\ee_4$ and $\ee_5$ have the same conditional response distributions. Therefore, the parameters $p_{\zero}$, $p_{\ee_1}$, $p_{\ee_3}$, $p_{\ee_4}$ and $p_{\ee_5}$  can not be identified.

\paragraph{Identifiability of the Sequential DINA model: a TIMSS 2007 data example.}

We consider the dataset in \cite{sequential}, which is derived from booklets 4 and 5 of the TIMSS 2007 fourth-grade mathematics assessment. This subset, originally utilized by \citet{Lee2011CognitiveDiagnostic}, includes responses from 823 students to 12 items, which are linked to eight of the original 15 attributes. Notably, items 3 and 9 are constructed-response items scored polytomously across three response categories (0, 1, and 2). The dataset also features items like 7a and 7b which, due to their heavy interdependence, can be treated as a single polytomous item. We consider the Sequential DINA model in this example. 
Following \citet{sequential}, the attribute definitions for the TIMSS data are given in Table 3 and the $\QQ$-matrix is in Table 4. The corresponding   $\QQ^1$-matrix is also presented below.

    \begin{table}[h]
\centering
\caption{Attribute definitions for TIMSS 2007 data \citep{sequential}}
\begin{tabular}{|c|p{12cm}|}
\hline
Attribute & Description \\
\hline
$\alpha_1$ & Representing, comparing, and ordering whole numbers as well as demonstrating knowledge of place value \\
\hline
$\alpha_2$ & Recognizing multiples, computing with whole numbers using the four operations, and estimating computations \\
\hline
$\alpha_3$ & Solving problems, including those set in real-life contexts 
\\
\hline
$\alpha_4$ & Finding the missing number or operation and modelling simple situations involving unknowns in number sentence or expression \\
\hline
$\alpha_5$ & Describing relationships in patterns and their extensions; generating pairs of whole numbers by a given rule and identifying a rule for every relationship given pairs of whole numbers \\
\hline
$\alpha_6$ & Reading data from tables, pictographs, bar graphs, and pie charts \\
\hline
$\alpha_7$ & Comparing and understanding how to use information from data \\
\hline
$\alpha_8$ & Understanding different representations and organizing data using tables, pictographs, and bar graphs \\
\hline
\end{tabular}
\end{table}

    \begin{table}[h]
\centering
\caption{$\QQ$-matrix for TIMSS 2007 data \citep{sequential}}
\begin{tabular}{|c|c|c|c|c|c|c|c|c|c|c|}
\hline \multirow[b]{2}{*}{ Item } & \multirow[b]{2}{*}{ TIMSS item no. } & \multirow[b]{2}{*}{ Category } & \multicolumn{8}{|c|}{ Attributes } \\
\hline & & & $\alpha_1$ & $\alpha_2$ & $\alpha_3$ & $\alpha_4$ & $\alpha_5$ & $\alpha_6$ & $\alpha_7$ & $\alpha_8$ \\
\hline 1 & M041052 & 1 & 1 & 1 & 0 & 0 & 0 & 0 & 0 & 0 \\
\hline 2 & M041281 & 1 & 0 & 1 & 1 & 0 & 1 & 0 & 0 & 0 \\
\hline $3 \mathbf{a}$ & M041275 & 1 & 1 & 0 & 0 & 0 & 0 & 1 & 0 & 1 \\
\hline $3 \mathbf{b}$ & M041275 & 2 & 1 & 0 & 0 & 0 & 0 & 1 & 0 & 1 \\
\hline 4 & M031303 & 1 & 0 & 1 & 1 & 0 & 0 & 0 & 0 & 0 \\
\hline 5 & M031309 & 1 & 0 & 1 & 1 & 0 & 0 & 0 & 0 & 0 \\
\hline 6 & M031245 & 1 & 0 & 1 & 0 & 1 & 0 & 0 & 0 & 0 \\
\hline $7 \mathbf{a}$ & M031242A & 1 & 0 & 1 & 1 & 0 & 1 & 0 & 0 & 0 \\
\hline $7 \mathbf{b}$ & M031242B & 2 & 0 & 0 & 0 & 0 & 0 & 0 & 1 & 0 \\
\hline 8 & M031242C & 1 & 0 & 1 & 1 & 0 & 1 & 0 & 1 & 0 \\
\hline $9 \mathbf{a}$ & M031247 & 1 & 0 & 1 & 1 & 1 & 0 & 0 & 0 & 0 \\
\hline 9b & M031247 & 2 & 0 & 1 & 1 & 1 & 0 & 0 & 0 & 0 \\
\hline 10 & M031173 & 1 & 0 & 1 & 1 & 0 & 0 & 0 & 0 & 0 \\
\hline 11 & M031172 & 1 & 1 & 1 & 0 & 0 & 0 & 1 & 0 & 1 \\
\hline
\end{tabular}
\end{table}

\begin{equation*}
    \QQ^1 = 
\begin{pmatrix}
1 & 1 & 0 & 0 & 0 & 0 & 0 & 0 \\
0 & 1 & 1 & 0 & 1 & 0 & 0 & 0 \\
1 & 0 & 0 & 0 & 0 & 1 & 0 & 1 \\
0 & 1 & 1 & 0 & 0 & 0 & 0 & 0 \\
0 & 1 & 1 & 0 & 0 & 0 & 0 & 0 \\
0 & 1 & 0 & 1 & 0 & 0 & 0 & 0 \\
0 & 1 & 1 & 0 & 1 & 0 & 0 & 0 \\
0 & 1 & 1 & 0 & 1 & 0 & 1 & 0 \\
0 & 1 & 1 & 1 & 0 & 0 & 0 & 0 \\
0 & 1 & 1 & 0 & 0 & 0 & 0 & 0 \\
1 & 1 & 0 & 0 & 0 & 1 & 0 & 1 \\
\end{pmatrix}
\end{equation*}


According to   Proposition~\ref{neSeq1}, since the $\QQ^1$-matrix does not contain an identity matrix, the model parameters are not identifiable. 
Specifically, since the matrix does not contain any $\ee_j$ for $j = 1, 2, \ldots, 8$, if we take $\beta_{j,1}^- = 0$ for $j = 1, 2, \ldots 20$, 
then subjects with attribute profiles $\zero$ and $\ee_j$ for $j = 1, 2, \ldots, 8$ are not able to complete the first categories of all the items. Since $\beta_{j,1}^- \equiv 0$, according to the model construction in Section~\ref{sec:seq}, these attribute profiles cannot complete other categories either.
Therefore, attribute profiles $\zero$, $\ee_j$ for $j = 1, 2, \ldots, 8$  have the same probability of completing all the categories of all the items, which is zero. Therefore, the parameters $p_{\zero}$, $p_{\ee_j}$ for $j = 1, 2, \ldots, 8$   can not be identified.


\begin{remark}
For the above educational assessment examples, while the analysis shows nonidentifiability issues for the  two considered models, this should not overshadow the potential for analyzing these data using polytomous DINA or more general cognitive diagnosis models.
First, as discussed in Section~\ref{sec:seq-id}, although the two models in our application data fail to satisfy the completeness condition, if we consider the more relaxed generic identifiability of the model parameters, that is, allowing nonidentifiability of parameters in a negligible zero-measure set of the parameter space, the stringent completeness condition may not be necessary, as discussed in \cite{partial}. 
Second, the investigation of partial identifiability, as proposed by \citet{partial}, could also be extended to the current situation.  Specifically, when the completeness condition is violated, 
partial identifiability may be established to partially identify the nonidentifiable 
proportion parameters $\pp$ up to their equivalent classes. 
For example, in the first example, since attribute profiles $\zero$, $\ee_1$, $\ee_3$, $\ee_4$ and $\ee_5$ have the same conditional response distributions, they can be grouped and considered as an equivalent latent class. Partial identifiability then seeks to identify parameter $(p_{\zero} + p_{\ee_1} + p_{\ee_3} + p_{\ee_4} + p_{\ee_5})$ as a whole, instead of treating each proportion parameter separately. Under such relaxation, the models applied to the data examples may be partially identifiable. 
Finally, beyond the DINA models considered in this paper,   general cognitive diagnosis models \citep{GPDINA,sequential} may be more appropriate for the two datasets, and studying the identifiability \citep*{partial}  of these models could be also of great interest.
Further explorations of these interesting extensions are promising future research directions.
\end{remark}

\section{Discussion}\label{sec-disc}
This paper presents the sufficient and necessary conditions for the identifiability of CDMs with polytomous responses. 
Our results focus on two popular models under the DINA assumption: the GPDINA model and the Sequential DINA model. 
For both models, we provide the sufficient and necessary conditions for their identifiability. 
The results can be easily extended to the DINO (deterministic input; noisy ``or" gate) model 
\citep{Templin} 
through the duality between the DINA and DINO models.
While the minimum requirements for more general CDMs are still unknown, our proposed necessary conditions remain necessary for them since our polytomous DINA models are submodels of the general CDMs. Therefore our
results would also shed light on the study of their identifiability.


The popularity of polytomous data is not restricted to response data, and polytomous attributes data is also receiving more and more attention \citep*{Haberman, davier2008general, pGDINA, qiu19}. 
Yet the discussion on the identifiability of such models has sparingly been considered.
More interestingly, we may further study the identifiability results under the general CDM framework with polytomous responses and polytomous attributes. 

The $\QQ$-matrix in this paper is assumed to be correctly specified. In practice, the $\QQ$-matrix
is usually constructed by the designers, which can be subjective and may not be accurate. For this reason, researchers have proposed to estimate and validate the design $\QQ$-matrix based on the response data, which motivates the study of the  identifiability of the $\QQ$-matrix
\citep[e.g.,][]{JLGXZY2011, chen2015statistical, xushang,Culpepper2019,chen2020sparse,gu21}. Nevertheless, most of these existing works focus on dichotomous responses, and only few have explored the identifiability of $\QQ$-matrix in the polytomous data setting, which would also be an interesting future research topic. 
 

\bibliographystyle{apalike}
\bibliography{ref.bib}

\newpage
\appendix

\section{Supplementary Material}\label{sec:proof}
This supplementary material provides the proofs of the main theoretical results and conclusions in Examples \ref{counterex1}-\ref{counterex3}. Specifically, Section~\ref{app:prop} provides the proofs of Propositions 1 and 2, while the identifiability results for the GPDINA model and the Sequential DINA model are presented in Sections~\ref{app:GPDINA} and~\ref{app:seq}, respectively. Additionally, the proofs of the examples are provided in Section~\ref{app:ex}.
\subsection{Proofs of Propositions 1 and 2}\label{app:prop}
This section deals with the zero $\qq$-vectors ($\qq = \zero$) in $\QQ$-matrix. Our propositions show that, for both the GPDINA model and the Sequential DINA model, excluding items or categories whose corresponding $\qq$-vectors are all zero does not affect the identifiability results. 
\propGPDINA*

\begin{proof}
    According to Lemma~\ref{lem:t-matrix},
    it suffices to show that the GPDINA models with $\QQ$-matrix and $\QQ_{\Delta}$-matrix yield the same equation system $\TT\pp = \bar \TT \bar \pp$. Let $\TT$ and $\TT'$ denote the $\TT$-matrix under the $\QQ$-matrix and $\QQ_{-\Delta}$-matrix separately, and let $\rr_{\Delta}$ denote the $\Delta$-coordinates of $\rr$. Then $\TT'$ is a submatrix of $\TT$ which excludes vectors $\TT_{\rr}$ in $\TT$ with $\rr_{\Delta} \neq \zero$. i.e., $\TT = \TT' \cup \{\TT_{\rr}: \rr_\Delta \neq \zero\}$. We now show that $\{\TT_{\rr}: \rr_\Delta \neq \zero\}$ does not add additional constraints to the equation system $\TT'\pp = \bar\TT'\bar\pp$. For $j \in \Delta$ and $l \in [H_j]$, recall that $\one^\top = (1, 1,\ldots, 1)$, 
    since $\qq_j = \zero$ and $\one^\top \pp = \one^\top \bar \pp = 1$, we have
     \[\TT_{l\ee_j}\pp = \theta_{j,l}^+\one^\top \pp = \theta_{j,l}^+, \;\; \bar \TT_{l\ee_j} \bar \pp = \bar \theta_{j,l}^+\one^\top \bar \pp = \bar \theta_{j,l}^+. \;\;
    \]
    So $\TT_{l\ee_j}\pp = \bar \TT_{l\ee_j} \bar \pp$ gives $\theta_{j,l}^+ = \bar \theta_{j,l}^+$ for $l \in [H_j]$. Therefore, for $j \in \Delta$, parameters $\theta_{j,l}^+$ are all identifiable. Furthermore, 
    for any $\rr$ s.t. $\rr_{\Delta} \neq 0$, write $\rr = \sum_{j \in \Delta} r_j \ee_j + \left(\rr - \sum_{j \in \Delta} r_j \ee_j\right)$, then $\TT_{\left(\rr - \sum_{j \in \Delta} r_j \ee_j\right)}$ is a vector in $\TT'$-matrix,
    and
    \[\TT_{\rr} = \left(\circ_{j \in \Delta} \TT_{r_j \ee_j}\right) \circ \TT_{\left(\rr - \sum_{j \in \Delta} r_j \ee_j\right)} = \prod_{j \in \Delta}\theta_{j, r_j}^+ \TT_{\left(\rr - \sum_{j \in \Delta} r_j \ee_j\right)} = \prod_{j \in \Delta}\bar \theta_{j, r_j}^+ \TT_{\left(\rr - \sum_{j \in \Delta} r_j \ee_j\right)}.\]
    Similarly, we have  $\bar \TT_{\rr} = \prod_{j \in \Delta} \bar \theta_{j, r_j}^+ \bar \TT_{\left(\rr - \sum_{j \in \Delta} r_j \ee_j\right)}$.
    Therefore, $\TT_{\rr}\pp = \bar \TT_{\rr} \bar \pp$ is equivalent to  $\TT_{\left(\rr - \sum_{j \in \Delta} r_j \ee_j\right)} \pp = \bar \TT_{\left(\rr - \sum_{j \in \Delta} r_j \ee_j\right)} \bar \pp$.
   Therefore, the model with $\QQ$-matrix and $\QQ_{\Delta}$-matrix give the same equation system $\TT\pp = \bar \TT \bar \pp$. 
\end{proof}

\propseq*
\begin{proof}
    Similar to the previous analysis, we can show that $\beta_{j,l}^+ = \bar \beta_{j,l}^+$ for $(j,l)\in \Delta^s$. If $l = 1$, since $\qq_{j,l} = \zero$ and $\one^\top \pp  = \one^\top \bar\pp = 1$, we have
    \[\TT_{\ee_j}\pp = \beta_{j,1}^+\one^\top \pp = \beta_{j,1}^+, \;\; \bar \TT_{\ee_j} \bar \pp = \bar \beta_{j,1}^+\one^\top \bar\pp = \bar \beta_{j,1}^+. \;\;
    \]
    Thus, $\TT_{\ee_j}\pp = \bar \TT_{\ee_j} \bar \pp$ gives $\beta_{j,1}^+ = \bar \beta_{j,1}^+$.
    For $l > 1$, since $\qq_{j,l} = \zero$,  $\TT_{l\ee_j} = \beta_{j,l}^+\TT_{(l-1)\ee_j}  $, and since $\TT_{(l-1)\ee_j} \pp = \bar \TT_{(l-1)\ee_j} \bar \pp$, we have 
    \[
    \TT_{l\ee_j}\pp = \beta_{j,l}^+\TT_{(l-1)\ee_j} \pp = \beta_{j,l}^+ \bar \TT_{(l-1)\ee_j} \bar \pp = 
    \bar \beta_{j,l}^+ \bar \TT_{(l-1)\ee_j} \bar \pp = \bar\TT_{l\ee_j}\bar\pp.
    \]
    Thus, $\beta_{j,l}^+ = \bar \beta_{j,l}^+$. So the item parameters in $\Delta^s$ are all identifiable.
    Furthermore, 
    write 
    \[
    \rr = \sum_{(j,l) \in \Delta^s} l \ee_j + \left(\rr - \sum_{(j,l) \in \Delta^s} l \ee_j\right), 
    \]
    then $\TT_{\left(\rr - \sum_{(j,l) \in \Delta^s} l \ee_j\right)} \pp = \bar \TT_{\left(\rr - \sum_{(j,l) \in \Delta^s} l \ee_j\right)} \bar \pp $, 
    and
    \[\TT_{\rr} = \left(\circ_{(j,l) \in \Delta^s} \TT_{l \ee_j}\right) \circ \TT_{\left(\rr - \sum_{(j,l) \in \Delta^s} l \ee_j\right)} = \prod_{(j,l) \in \Delta^s}\beta_{j, l}^+ \TT_{\left(\rr - \sum_{(j,l) \in \Delta^s} l \ee_j\right)} = \prod_{(j,l) \in \Delta^s}\bar \beta_{j, l}^+ \TT_{\left(\rr - \sum_{(j,l) \in \Delta^s} l \ee_j\right)},\]
    similarly, we have $\bar \TT_{\rr} = \prod_{j \in \Delta^s} \bar \theta_{j, r_j}^+ \bar \TT_{\left(\rr - \sum_{(j,l) \in \Delta^s} l \ee_j\right)}$.
    Therefore, $\TT_{\rr}\pp = \bar \TT_{\rr} \bar \pp$ gives 
    \[
    \TT_{\left(\rr - \sum_{(j,l) \in \Delta^s} l \ee_j\right)} \pp = \bar \TT_{\left(\rr - \sum_{(j,l) \in \Delta^s} l \ee_j\right)} \bar \pp.
    \]
   Therefore, the model with $\QQ$-matrix and $\QQ_{\Delta^s}$-matrix give exactly the same equation system $\TT\pp = \bar \TT \bar \pp$. 
\end{proof}

\subsection{Identifiability of GPDINA}\label{app:GPDINA}
\nesuGPDINA*
\begin{proof}[Proof of sufficiency]
Suppose the $\QQ$-matrix satisfies conditions C1-C3. Using Lemma~\ref{lem:t-matrix}, we show that
$\TT \pp = \bar{\TT}\bar{\pp}$ will give $ (\ttt^+,\ttt^-,\pp) = (\bar{\ttt}^+,\bar{\ttt}^-,\bar{\pp})$. Take one arbitrary non-zero category from each item, denoted by $c(j)$ ($1 \leq c(j) \leq H_j$), for $j \in [J]$. 
We show that $\theta_{j,c(j)}^+ = \bar{\theta}_{j,c(j)}^+$ and $\theta_{j,c(j)}^- = \bar{\theta}_{j,c(j)}^-$.

Our proof leverages the identifiability results of the binary DINA model. Consider constructing the following binary DINA model:
if we focus on solely the category $c(j)$ of each item $j$, for $j \in [J]$. If we dichotomize each item $j$ through category $c(j)$,
and reframe the response as binary response $I(R_j = c(j))_{j \in [J]}$, then the model is reduced to a binary DINA model. According to the model construction, the $\QQ$-matrix for the reduced model is equivalent to the $\QQ$-matrix for the original polytomous model, since every non-zero category of the same item requires the same attributes. The $\TT$-matrix for this reduced model, i.e., the marginal probability distribution for the dichotomized response $I(R_j = c(j))$, is simply a submatrix of the original $\TT$-matrix.  It is made up of the vectors that only involve category $c(j)$ of item $j$, i.e., $\left( \TT_{c(j)\cdot\ee_{j}}\right)_{j \in [J]}$ and their element-wise products.
The parameters for this reduced model are $\left(\left\{\theta_{j,c(j)}^+, \;\theta_{j,c(j)}^-\right\}_{j \in [J]}, \;\pp\right)$.
So the reduced model is completely a binary DINA model. Since the $\QQ$-matrix for the reduced model satisfies conditions C1-C3, as a direct result of \citet{id-dina}, $\bar{\TT}\bar{\pp} = \TT\pp$ will give $\bar{\pp} = \pp$, $\theta_{j,c(j)}^+ = \bar{\theta}_{j,c(j)}^+,\; \theta_{j,c(j)}^- = \bar{\theta}_{j,c(j)}^-, \;\text{for}\;j \in [J]$. This holds for any $c(j) \in [H_j]$, thus $(\bar{\ttt}^+,\bar{\ttt}^-,\bar{\pp}) =  (\ttt^+,\ttt^-,\pp)$ and we complete the proof.

\end{proof}

\begin{proof}[Proof of necessity]
We prove separately each condition is necessary.

\paragraph{Necessity of condition C1.} Suppose the $\QQ$-matrix is not complete, and WLOG, assume that $\ee_1^\top \notin (\qq_j)_{j=1}^J$. Then attributes profiles $\zero$ and $\ee_1$  have the same conditional response distributions. Therefore, the parameters $p_\zero$ and $p_{\ee_1}$ are exchangeable, and thus can not be identified. Therefore, condition C1 is necessary.
\paragraph{Necessity of condition C2.} Suppose the $\QQ$-matrix satisfies condition C1, but does not satisfy condition C2, i.e., there exists some attribute which is only required by at most two items. WLOG, assume this is the first attribute,
and it is the first and second items that require the first attribute, so the
$\QQ$-matrix can be written as follows:
\begin{equation}\label{eq:Q1}
    \QQ = 
    \begin{pmatrix}
    \text{item}\;1\;&
    1 \; & \zero^\top\\
    \hdashline[2pt/2pt]
    \text{item}\;2\;&
    1 \; & \vv^\top\\
    \hdashline[2pt/2pt]
    \text{item}\;(3:J)\;&
    \zero \; & \QQ'
    \end{pmatrix}.
\end{equation}
We partition $\aaa$ into two groups according to the first attribute:
\begin{align*}
    \bg^0 = \{\aaa: \alpha_1 = 0\} = \{\aaa = (0, \aaa^*),\;\aaa^* \in \{0, 1\}^{K-1}\},\\
    \bg^1 = \{\aaa: \alpha_1 = 1\} = \{\aaa = (1, \aaa^*), \;\aaa^* \in \{0, 1\}^{K-1}\},
\end{align*} 
so each group has $2^{K-1}$ attribute profiles, and we index the entries in each group by
\begin{align*}
    \bg^0_1 = (0, \zero), \; \bg^0_2 = (0, \ee_1), \ldots, \bg^0_K = (0, \ee_{K-1}), \; \bg^0_{K+1} = (0, \ee_1+\ee_2), \ldots, \bg^0_{2^{K-1}} = \left(0, \sum_{k=1}^{K-1} \ee_k\right), \\
    \bg^1_1 = (1, \zero), \; \bg^1_2 = (1, \ee_1), \ldots, \bg^1_K = (1, \ee_{K-1}), \; \bg^1_{K+1} = (1, \ee_1+\ee_2), \ldots, \bg^1_{2^{K-1}} = \left(1, \sum_{k=1}^{K-1} \ee_k\right),
\end{align*}
where $\ee_1, \ldots, \ee_{K-1} \in \{0,1\}^{K-1}$ have $K-1$ elements.
Therefore, the $k$-th ($k \in [2^{K-1}]$) entry of $\bg^0$ and $\bg^1$,
$\bg^0_k$ and $\bg^1_k$ share the same attributes except for the first one $\alpha_1$. Index the population proportion parameters $\pp$ in the following way:
\begin{equation}\label{eq:p}
  \pp = \begin{pmatrix}
\pp_{\bg^0}\\
\pp_{\bg^1}\\
\end{pmatrix},\;\;\;
\text{where}\;\;
\pp_{\bg^0} = \begin{pmatrix}
p_{\bg^0_1}\\
p_{\bg^0_2}\\
\vdots\\
p_{\bg^0_{2^{K-1}}}\\
\end{pmatrix}\;\;\;\text{and}\;\;
\pp_{\bg^1} = \begin{pmatrix}
p_{\bg^1_1}\\
p_{\bg^1_2}\\
\vdots\\
p_{\bg^1_{2^{K-1}}}
\end{pmatrix} .
\end{equation}
We now seek to construct $(\bar{\ttt}^+,\bar{\ttt}^-,\bar{\pp})\neq (\ttt^+,\ttt^-,\pp)$ such that (\ref{lem:t-matrix-eq}) holds: take $\bar{\ttt}_j^+ = \ttt_j^+$ for $j > 2 $ and $\bar{\ttt}_j^- = \ttt_j^-$  for $j \in [J] $. We claim that, with a simplified matrix $\mct$, for $\bar{\TT}\bar{\pp} = \TT \pp$ to hold, it suffices to have 
$\bar{\mct} \bar{\pp} = \mct \pp$. When the entries in $\pp$ are indexed according to (\ref{eq:p}), the $\mct$ is given as follows:
\begin{equation}\label{eq:mct1}
\mct = \begin{pmatrix}
1 & 1\\
\ttt_1^- & \ttt_1^+\\
\ttt_2^- & \ttt_2^+\\
\ttt_1^- \otimes \ttt_2^- & \ttt_1^+ \otimes \ttt_2^+
\end{pmatrix} \otimes \cI ,
\end{equation}
where $\cI = \cI_{2^{K-1}}$. We now prove the claim.
\begin{itemize}
\item For any $\rr$ s.t. $r_1 = r_2 = 0$, 
$\TT_{\rr}$ does not involve $\ttt_1^+$, $\ttt_1^-$, $\ttt_2^+$, and $\ttt_2^-$. According to the construction of the $\QQ$-matrix in (\ref{eq:Q1}),
the response $\rr$ does not require $\alpha_1$, so the response distributions in both groups are the same, i.e., $t_{\rr,\, \bg^0_k} \equiv t_{\rr, \,\bg^1_k} \; \text{for}\; k \in [2^{K-1}]$. Since $\bar{\ttt}_j^+ = \ttt_j^+$ and $\bar{\ttt}_j^- = \ttt_j^-$  for $j > 2$, we further have
\begin{equation}\label{eq:t1}
    t_{\rr,\, \bg^0_k} \equiv t_{\rr, \,\bg^1_k} \equiv \bar{t}_{\rr,\, \bg^0_k} \equiv \bar{t}_{\rr, \,\bg^1_k} \; \text{for}\; k \in [2^{K-1}].
\end{equation}
If we denote $\TT_{\rr}^{(\frac{1}{2})}$ as the first half of the vector $\TT_{\rr}$, 
i.e., the response of group $\bg^0$,
then equation (\ref{eq:t1}) indicates that 
$\TT_{\rr}^{(\frac{1}{2})}$ is also the response of group $\bg^1$. Moreover, 
$\bar{\TT}_{\rr}^{(\frac{1}{2})} = \TT_{\rr}^{(\frac{1}{2})}$ and
\begin{equation*}
    \TT_{\rr} = \left(\,\TT_{\rr}^{(\frac{1}{2})}\;\; \TT_{\rr}^{(\frac{1}{2})}\,\right) = \TT_{\rr}^{(\frac{1}{2})} \left(\,\cI\;\;\cI\,\right) 
= \bar{\TT}_{\rr}^{(\frac{1}{2})}  \left(\,\cI\;\;\cI\,\right) = \bar{\TT}_{\rr}.
\end{equation*}
Therefore, equations $\Big(\cI\;\;\cI\Big)\bar{\pp} = \Big(\cI\;\;\cI\Big){\pp}$ 
in (\ref{eq:mct1}) gives
\begin{equation*}
    \bar{\TT}_{\rr}\bar{\pp} = \bar{\TT}_{\rr}^{(\frac{1}{2})}  \left(\,\cI\;\;\cI\,\right) \bar{\pp} = \TT_{\rr}^{(\frac{1}{2})} \Big(\cI\;\;\cI\Big){\pp} = \TT_{\rr}\pp.
\end{equation*}
\item For any $\rr$ s.t. $r_1 \neq 0,  r_2 = 0$, write $\rr = r_1\ee_1 + (\rr - r_1\ee_1)$,
then $\TT_{\rr} = \TT_{(\rr - r_1\ee_1)} \circ \TT_{r_1\ee_1}$, and
the response $(\rr - r_1\ee_1)$ belongs to the case we analyzed previously. Therefore, $\bar{\TT}_{(\rr - r_1\ee_1)}^{(\frac{1}{2})} = {\TT}_{(\rr - r_1\ee_1)}^{(\frac{1}{2})}$ and
\begin{align*}
   \bar{\TT}_{(\rr - r_1\ee_1)} = \bar{\TT}_{(\rr - r_1\ee_1)}^{(\frac{1}{2})}\Big(\cI\;\; \cI\Big) = \TT_{(\rr - r_1\ee_1)}^{(\frac{1}{2})}\Big(\cI\;\; \cI\Big) &= \TT_{(\rr - r_1\ee_1)},
\end{align*}
which gives
\begin{align*}  
    \bar{\TT}_{\rr}\bar{\pp} = 
(\bar{\TT}_{(\rr - r_1\ee_1)} \circ \bar{\TT}_{r_1\ee_1})\;\bar{\pp} &= 
\bar{\TT}_{(\rr - r_1\ee_1)}\,( \bar{\TT}_{r_1\ee_1} \circ \bar{\pp} ) \\
&= {\TT}_{(\rr - r_1\ee_1)}^{(\frac{1}{2})}\Big(\cI\;\; \cI\Big)  \,( \bar{\TT}_{r_1\ee_1} \circ \bar{\pp} )\\
&= {\TT}_{(\rr - r_1\ee_1)}^{(\frac{1}{2})}
\left(\Big(\cI\;\; \cI\Big)  \circ  \bar{\TT}_{r_1\ee_1} \right) \bar{\pp} \\
&= {\TT}_{(\rr - r_1\ee_1)}^{(\frac{1}{2})}
\Big(\bar\theta_{1,r_1}^-\cI\;\;\; \bar\theta_{1,r_1}^+\cI\Big)  \bar{\pp}.
\end{align*}
Therefore, $\Big(\bar{\ttt}_1^- \otimes \cI\;\;\; \bar{\ttt}_1^+ \otimes \cI\Big)\bar{\pp} = \Big(\ttt_1^- \otimes \cI\;\;\; \ttt_1^+ \otimes \cI\Big){\pp}$ 
in equation (\ref{eq:mct1})
guarantees that for $\forall r_1 \in [H_1]$,
\[
\Big(\bar\theta_{1,r_1}^-\cI\;\;\; \bar\theta_{1,r_1}^+\cI\Big)  \bar{\pp} = \Big(\theta_{1,r_1}^-\cI\;\;\; \theta_{1,r_1}^+\cI\Big)  \pp.
\]
Therefore,
\begin{equation*}
\bar\TT_{\rr}\bar\pp = 
{\TT}_{(\rr - r_1\ee_1)}^{(\frac{1}{2})}
\Big(\bar\theta_{1,r_1}^-\cI\;\;\; \bar\theta_{1,r_1}^+\cI\Big)  \bar{\pp}
 = {\TT}_{(\rr - r_1\ee_1)}^{(\frac{1}{2})}
\Big(\theta_{1,r_1}^-\cI\;\;\; \theta_{1,r_1}^+\cI\Big)  \pp = \TT_{\rr}\pp.
\end{equation*}
\item Similarly, for any $\rr$ s.t. $r_1 = 0, r_2 \neq 0$, $\Big(\bar{\ttt}_2^- \otimes \cI\;\;\; \bar{\ttt}_2^+ \otimes \cI\Big)\bar{\pp} = \Big(\ttt_2^- \otimes \cI\;\;\; \ttt_2^+ \otimes \cI\Big){\pp}$ with $\bar{\ttt}_2^- = {\ttt}_2^-$
guarantees that $\bar{\TT}_{\rr}\bar{\pp} = \TT_{\rr}\pp$.
\item Similarly, for any $\rr$ s.t. $r_1 \neq 0, r_2 \neq 0$, $\Big(\bar{\ttt}_1^- \otimes \bar{\ttt}_2^- \otimes \cI\;\;\; \bar{\ttt}_1^+ \otimes \bar{\ttt}_2^+  \otimes \cI\Big)\bar{\pp} = \Big(\ttt_1^- \otimes \ttt_2^- \otimes \cI\;\;\; \ttt_1^+ \otimes \ttt_2^+ \otimes \cI\Big){\pp}$ guarantees $\bar{\TT}_{\rr}\bar{\pp} = \TT_{\rr}\pp$.
\end{itemize}

Next we construct $(\bar{\ttt}_1^+, \bar{\ttt}_2^+, \bar{\ttt}_1^-, \bar{\pp})$ s.t. $\bar{\mct} \bar{\pp} = \mct \pp$ holds.
Let $\bar{\pp}_{\bg^1} =\rho \cdot \bar{\pp}_{\bg^0}$, $\pp_{\bg^0} = u \cdot \bar{\pp}_{\bg^0}$, and ${\pp}_{\bg^1} = v \cdot \bar{\pp}_{\bg^0}$. Then $\bar{\mct} \bar{\pp} = \mct \pp$ can be simplified to
\begin{align*}
\bar{\mct} 
\begin{pmatrix}
\bar{\pp}_{\bg^0}\\
\rho \cdot \bar{\pp}_{\bg^0}
\end{pmatrix} &= \mct 
\begin{pmatrix}
u \cdot \bar{\pp}_{\bg^0}\\
v \cdot \bar{\pp}_{\bg^0}
\end{pmatrix},
\end{align*}
i.e.,
\begin{equation*}
    \begin{cases}
        \;\;\;\cI \bar{\pp}_{\bg^0} + \rho \cI \bar{\pp}_{\bg^0} = u \cI \bar{\pp}_{\bg^0} + v \cI \bar{\pp}_{\bg^0};\\
\bar{\theta}_{1,l_1}^-\cI \bar{\pp}_{\bg^0} + \rho \bar{\theta}_{1,l_1}^+\cI \bar{\pp}_{\bg^0} = u \theta_{1,l_1}^-\cI \bar{\pp}_{\bg^0} + v \theta_{1,l_1}^+\cI \bar{\pp}_{\bg^0},\;\;l_1 \in [H_1];\\
\bar{\theta}_{2,l_2}^-\cI \bar{\pp}_{\bg^0} + \rho \bar{\theta}_{2,l_2}^+ \cI \bar{\pp}_{\bg^0} = u \theta_{2,l_2}^-\cI \bar{\pp}_{\bg^0} + v \theta_{2,l_2}^+\cI \bar{\pp}_{\bg^0},\;\;l_2 \in [H_2];\\
\bar{\theta}_{1,l_1}^- \bar{\theta}_{2,l_2}^-\cI \bar{\pp}_{\bg^0} + \rho \bar{\theta}_{1,l_1}^+ \bar{\theta}_{2,l_2}^+\cI \bar{\pp}_{\bg^0} = u \theta_{1,l_1}^- \theta_{2,l_2}^- \cI \bar{\pp}_{\bg^0} + v \theta_{1,l_1}^+ \theta_{2,l_2}^+ \cI  \bar{\pp}_{\bg^0}, \;\;l_1 \in [H_1], \;\;l_2 \in [H_2].
    \end{cases}
\end{equation*}
Then it suffices to have
\begin{equation}\label{eq:four eq1}
\begin{cases}
    1+ \rho = u + v;\\
\bar{\theta}_{1, l_1}^-  + \rho \bar{\theta}_{1,l_1}^+ = u \theta_{1,l_1}^- + v \theta_{1, l_1}^+, \;\;l_1 \in [H_1];\\
\bar{\theta}_{2, l_2}^-  + \rho \bar{\theta}_{2,l_2}^+ = u \theta_{2,l_2}^- + v \theta_{2, l_2}^+, \;\;l_2 \in [H_2];\\
\bar{\theta}_{1, l_1}^-\bar{\theta}_{2, l_2}^-  + \rho \bar{\theta}_{1,l_1}^+\bar{\theta}_{2,l_2}^+ = u \theta_{1,l_1}^- \theta_{2,l_2}^- + v \theta_{1, l_1}^+ \theta_{2, l_2}^+,\;\;l_1 \in [H_1], \;\;l_2 \in [H_2].
\end{cases}
\end{equation}
Let some $\kappa \in (0,1)$ s.t
\begin{align}
\begin{pmatrix}
    \theta_{1, l_1}^-\\
    \bar{\theta}_{1, l_1}^-\\
    \theta_{1, l_1}^+\\
    \bar{\theta}_{1, l_1}^+
\end{pmatrix} = \kappa^{l_1-1}
\begin{pmatrix}
    \theta_{1, 1}^-\\
    \bar{\theta}_{1, 1}^-\\
    \theta_{1, 1}^+\\
    \bar{\theta}_{1, 1}^+
\end{pmatrix}, \;\text{and}\;
\begin{pmatrix}
    \theta_{2, l_2}^-\\
    \bar{\theta}_{2, l_2}^-\\
    \theta_{2, l_2}^+\\
    \bar{\theta}_{2, l_2}^+
\end{pmatrix} = \kappa^{l_2-1}
\begin{pmatrix}
    \theta_{2, 1}^-\\
    \bar{\theta}_{2, 1}^-\\
    \theta_{2, 1}^+\\
    \bar{\theta}_{2, 1}^+
\end{pmatrix} \;\;\;\text{for}\;\;l_1 \in [H_1],\;l_2 \in [H_2],
\end{align} 
then equations (\ref{eq:four eq1}) are reduced to the following four equations:
\begin{equation*}
\begin{cases}
    1+ \rho = u + v;\\
\bar{\theta}_{1, 1}^-  + \rho \bar{\theta}_{1,1}^+ = u \theta_{1,1}^- + v \theta_{1, 1}^+ ;\\
\bar{\theta}_{2, 1}^-  + \rho \bar{\theta}_{2,1}^+ = u \theta_{2,1}^- + v \theta_{2,1}^+;\\
\bar{\theta}_{1, 1}^-\bar{\theta}_{2, 1}^-  + \rho \bar{\theta}_{1,1}^+\bar{\theta}_{2,1}^+ = u \theta_{1,1}^- \theta_{2,1}^- + v \theta_{1, 1}^+\theta_{2, 1}^+.
\end{cases}
\end{equation*}
There are five parameters $(\rho, u, v, \,\bar \theta_{1, 1}^+,\, \bar \theta_{2, 1}^+)$ with four constraints, so there are infinite many solutions. Consequently, parameters $(\ttt^+,\ttt^-,\pp)$ are not identifiable and condition C2 is indeed necessary.
\paragraph{Necessity of condition C3.} Suppose the $\QQ$-matrix satisfies conditions C1 and C2, 
but does not satisfy condition C3. WLOG, we may write $\QQ$ as
\[
\QQ=\left(\begin{array}{c}
\cI_K \\
\QQ^*
\end{array}\right)_{J\times K}, 
 \;\;\text{where}\;\;
\QQ^* = \begin{pmatrix}
\vv\;\; \vv \; \;\vdots\; \;\vdots\;\; \vdots
\end{pmatrix}.
\]
We partition $\aaa$ into four groups according to the first and the second attributes:
\begin{align*}
    \bg^{00} &= \{\aaa: \alpha_1 = 0, \alpha_2 = 0\} = \{\aaa = (0, 0, \aaa^*),\;\aaa^* \in \{0, 1\}^{K-2}\},\\
    \bg^{10} &= \{\aaa: \alpha_1 = 1, \alpha_2 = 0\} = \{\aaa = (1, 0, \aaa^*),\;\aaa^* \in \{0, 1\}^{K-2}\},\\
    \bg^{01} &= \{\aaa: \alpha_1 = 0, \alpha_2 = 1\} = \{\aaa = (0, 1, \aaa^*),\;\aaa^* \in \{0, 1\}^{K-2}\},\\
    \bg^{11} &= \{\aaa: \alpha_1 = 1,  \alpha_2 = 1\} = \{\aaa = (1, 1, \aaa^*),\;\aaa^* \in \{0, 1\}^{K-2}\}.
\end{align*} 
So each group has $2^{K-2}$ attribute profiles. Index the entries in each group through the following: in group $\bg^{00}$, 
\begin{equation*}
    \resizebox{1.06\linewidth}{!}{$
\bg^{00}_1 = (0, 0, \zero), \; \bg^{00}_2 = (0, 0, \ee_1), \ldots, \bg^{00}_K = (0, 0, \ee_{K-2}), \; \bg^{00}_{K+1} = (0, 0, \ee_1+\ee_2), \ldots, \bg^{00}_{2^{K-2}} = \left(0, 0, \sum_{k=1}^{K-2} \ee_k\right),
$}
\end{equation*}
where $\ee_1, \ldots, \ee_{K-2} \in \{0,1\}^{K-2}$ have $K-2$ elements. 
Similarly we index the elements in $\bg^{10},\; \bg^{01},\; \bg^{11}$, so that $\bg^{00}_k$, $\bg^{10}_k$, $\bg^{01}_k$ and $\bg^{11}_k$ for $k \in [2^{K-2}]$ share the same attributes except for the first and second attributes. Index the population proportion parameters $\pp$ in the following way:
\begin{equation*}
\begin{frame}{}
    \resizebox{1.06\linewidth}{!}{$
\pp = \begin{pmatrix}
\pp_{\bg^{00}}\\
\pp_{\bg^{10}}\\
\pp_{\bg^{01}}\\
\pp_{\bg^{11}}
\end{pmatrix},\;\;\;
\text{where}\;\;
\pp_{\bg^{00}} = \begin{pmatrix}
p_{\bg^{00}_1}\\
p_{\bg^{00}_2}\\
\vdots\\
p_{\bg^{00}_{2^{K-2}}}\\
\end{pmatrix},\;
\pp_{\bg^{10}} = \begin{pmatrix}
p_{\bg^{10}_1}\\
p_{\bg^{10}_2}\\
\vdots\\
p_{\bg^{10}_{2^{K-2}}}\\
\end{pmatrix},\;
\pp_{\bg^{01}} = \begin{pmatrix}
p_{\bg^{01}_1}\\
p_{\bg^{01}_2}\\
\vdots\\
p_{\bg^{01}_{2^{K-2}}}\\
\end{pmatrix},\;
\pp_{\bg^{11}} = \begin{pmatrix}
p_{\bg^{11}_1}\\
p_{\bg^{11}_2}\\
\vdots\\
p_{\bg^{11}_{2^{K-2}}}\\
\end{pmatrix}.
$}
\end{frame}
\end{equation*}
Next, we seek to construct $(\bar{\ttt}^+,\bar{\ttt}^-,\bar{\pp})\neq (\ttt^+,\ttt^-,\pp)$ such that (\ref{lem:t-matrix-eq}) holds: take $\bar{\ttt}_j^+ = \ttt_j^+$ for all $j$, 
 $\bar{\ttt}_j^- = \ttt_j^-$ for $j > 2 $, and let $\bar{\pp}_{\bg^{11}} = \pp_{\bg^{11}}$.
If we let $\bg^{-11} = \{\bg^{00},\,\bg^{10},\,\bg^{01}\}$ denote the union of the other three groups, and denote its corresponding population proportion parameters as
\begin{equation*}
   \pp_{\bg^{-11}} = \begin{pmatrix}
       \pp_{\bg^{00}}\\
       \pp_{\bg^{10}}\\
       \pp_{\bg^{01}}
   \end{pmatrix},
\end{equation*}
then using a similar strategy we can show that if $\mct \pp_{\bg^{-11}} = \bar{\mct} \Bar{\pp}_{\bg^{-11}}$ holds, we have $\TT \pp = \bar{\TT} \bar{\pp}$, where $\mct$ is given as follows:
\begin{equation}\label{eq:mct2}
\mct = \begin{pmatrix}
1 & 1 & 1\\
\ttt_1^- & \ttt_1^+ & \ttt_1^-\\
\ttt_2^- & \ttt_2^- & \ttt_2^+\\
\ttt_1^- \otimes \ttt_2^- & \ttt_1^+ \otimes \ttt_2^- & \ttt_1^- \otimes \ttt_2^+
\end{pmatrix} \otimes \cI,
\end{equation}
with $\cI = \cI_{2^{K-2}}$ being the identity matrix of dimension $(K-2)\times(K-2)$. We now prove the claim.
\begin{itemize}
\item For any $\rr$ s.t. $r_1 = r_2 = 0$, $\TT_{\rr}$ does not involve $\ttt_1^+,\,\ttt_1^-$, $\ttt_2^+,\,\ttt_2^-$. According to the construction of the $\QQ$-matrix,
item $2, \ldots, J$ either require both $\alpha_1$ and $\alpha_2$ or require neither, therefore the response distribution in groups $\bg^{00},\,\bg^{10},\,\bg^{01}$ are the same, i.e., $t_{\rr,\, \bg^{00}_k} \equiv t_{\rr, \,\bg^{10}_k} \equiv t_{\rr, \,\bg^{01}_k}\, \; \text{for}\; k \in [2^{K-2}]$. Since $\bar{\ttt}_j^+ = \ttt_j^+$ $\bar{\ttt}_j^- = \ttt_j^-$ for $j > 2 $, we further have
\begin{equation}\label{eq:t2}
    t_{\rr,\, \bg^{00}_k} \equiv t_{\rr, \,\bg^{10}_k} \equiv t_{\rr, \,\bg^{01}_k} \equiv \bar{t}_{\rr,\, \bg^{00}_k} \equiv \bar{t}_{\rr, \,\bg^{10}_k} \equiv \bar{t}_{\rr, \,\bg^{01}_k},\;\;\text{and}\;\;t_{\rr, \,\bg^{11}_k} \equiv \bar{t}_{\rr,\, \bg^{11}_k} \;\;\text{for}\; \;k \in [2^{K-2}]. 
\end{equation}
Let $\TT_{\rr}^{(\frac{1}{4})}$ be the first quartile of the vector $\TT_{\rr}$, 
and $\TT_{\rr}^+$ be the last quartile of the vector $\TT_{\rr}$, then
$\TT_{\rr}^{(\frac{1}{4})}$ is the response of group $\bg^{00}$ and
$\TT_{\rr}^+$ is the response of group $\bg^{11}$. Then equation~\eqref{eq:t2} indicates that $\TT_{\rr}^{(\frac{1}{4})}$ is also the response of group $\bg^{10}$ and $\bg^{01}$. I.e., we have
\begin{align*}
    \TT_{\rr}^{(\frac{1}{4})} &= \underset{j: r_j \neq 0}{\prod} P\left( R_j = r_j\mid \QQ,\ttt^+,\ttt^-,\aaa = \bg^{00}\right) = \underset{j: r_j \neq 0}{\prod} P\left( R_j = r_j\mid\QQ, \ttt^+, \ttt^-, \aaa = \bg^{10}\right) \\
    &= \underset{j: r_j \neq 0}{\prod} P\left( R_j = r_j\mid\,\QQ, \ttt^+, \ttt^-, \aaa = \bg^{01}\right),\\
\TT_{\rr}^+ &= \underset{j: r_j \neq 0}{\prod} P\left( R_j = r_j\mid\,\QQ, \ttt^+, \ttt^-, \aaa = \bg^{11}\right).
\end{align*}
Then equation~\eqref{eq:t2} also indicates that $\bar{\TT}_{\rr}^{(\frac{1}{4})} = \TT_{\rr}^{(\frac{1}{4})}$ and $\TT_{\rr}^+ = \bar{\TT}_{\rr}^+$. Thus,
\begin{equation*}
    \TT_{\rr} = \left(\,\TT_{\rr}^{(\frac{1}{4})}\;\; \TT_{\rr}^{(\frac{1}{4})}\;\;\TT_{\rr}^{(\frac{1}{4})}\;\;\TT_{\rr}^+\,\right) = \left(\TT_{\rr}^{(\frac{1}{4})} \left(\,\cI\;\;\cI\;\;\cI\,\right)\; \TT_{\rr}^+\,\right)
= \left(\bar{\TT}_{\rr}^{(\frac{1}{4})} \left(\,\cI\;\;\cI\;\;\cI\,\right)\; \bar{\TT}_{\rr}^+\,\right).
\end{equation*}
combining $\Big(\cI\;\;\cI\;\;\cI\Big)\bar{\pp}_{\bg^{-11}} = \Big(\cI\;\;\cI\;\;\cI\Big) {\pp_{\bg^{-11}}}$ in \eqref{eq:mct2} and $\bar{\pp}_{\bg^{11}} = \pp_{\bg^{11}}$
gives
\begin{align*}
    \bar{\TT}_{\rr}\bar{\pp} =
    \left(\bar{\TT}_{\rr}^{(\frac{1}{4})} \left(\,\cI\;\;\cI\;\;\cI\,\right)\; \bar{\TT}_{\rr}^+\,\right) \begin{pmatrix}
        \bar{\pp}_{\bg^{-11}}\\
        \bar{\pp}_{\bg^{11}}
    \end{pmatrix} &=
    \bar{\TT}_{\rr}^{(\frac{1}{4})} \left(\,\cI\;\;\cI\;\;\cI\,\right) \bar{\pp}_{\bg^{-11}}\; + \bar{\TT}_{\rr}^+\, \bar{\pp}_{\bg^{11}}\\
     &= \TT_{\rr}^{(\frac{1}{4})} \left(\,\cI\;\;\cI\;\;\cI\,\right) \pp_{\bg^{-11}}\; + \TT_{\rr}^+\, \pp_{\bg^{11}}\\
     &= \TT_{\rr}\pp.
\end{align*}
\item For any $\rr$ s.t. $r_1 \neq 0,  r_2 = 0$, write $\rr = r_1\ee_1 + (\rr - r_1\ee_1)$,
then $\TT_{\rr} = \TT_{(\rr - r_1\ee_1)} \circ \TT_{r_1\ee_1}$. If we split the vector $\TT_{r_1\ee_1}$ via the third quartile, i.e.,
\begin{equation}\label{ll0}
    \TT_{r_1\ee_1} = \left(\TT_{r_1\ee_1}^{(\frac{3}{4})} \;\;\TT_{r_1\ee_1}^+\right),
\end{equation}
where $\TT_{r_1\ee_1}^{(\frac{3}{4})}$ is the first three-fourths of the vector and $\TT_{r_1\ee_1}^+$ is the last fourth part of the vector. Then we know that 
$\TT_{r_1\ee_1}^{(\frac{3}{4})}$ is the response of group $\bg^{-11}$ and
$\TT_{r_1\ee_1}^+$ is the response of group $\bg^{11}$, i.e.,
\begin{align*}
    \TT_{r_1\ee_1}^{(\frac{3}{4})} &=  P\left( R_1 = r_1\mid \QQ,\ttt^+,\ttt^-,\aaa = \bg^{-11}\right), \\
\TT_{r_1\ee_1}^+ &= P\left( R_1 = r_1\mid\,\QQ, \ttt^+, \ttt^-, \aaa = \bg^{11}\right).
\end{align*}
Since item 1 only requires $\alpha_1$, the entries of $\TT_{r_1\ee_1}^+$ are positive, and since $\ttt_1^+ = \bar{\ttt}_1^+$, we have $\bar{\TT}_{r_1\ee_1}^+ = \TT_{r_1\ee_1}^+$.
Furthermore, 
response $(\rr - r_1\ee_1)$ belongs to the case we analyzed previously, therefore, we have
\begin{align}\label{ll1}
   \bar{\TT}_{(\rr - r_1\ee_1)} 
    &= \left(\bar{\TT}_{(\rr - r_1\ee_1)}^{(\frac{1}{4})} \left(\,\cI\;\;\cI\;\;\cI\,\right)\;\; \bar{\TT}_{(\rr - r_1\ee_1)}^+\,\right)\\ \label{ll2}
    &=
   \left(\TT_{(\rr - r_1\ee_1)}^{(\frac{1}{4})} \left(\,\cI\;\;\cI\;\;\cI\,\right)\;\; \TT_{(\rr - r_1\ee_1)}^+\,\right).
\end{align}
Hence, 
\begin{align*}
    \bar{\TT}_{\rr}\bar{\pp} &= 
(\bar{\TT}_{(\rr - r_1\ee_1)} \circ \bar{\TT}_{r_1\ee_1})\;\bar{\pp}\\ 
&=
\bar{\TT}_{(\rr - r_1\ee_1)}\,( \bar{\TT}_{r_1\ee_1} \circ \bar{\pp} ) \\
&\stackrel{(\ref{ll1})}{=\joinrel=}  \left(\bar{\TT}_{(\rr - r_1\ee_1)}^{(\frac{1}{4})} \left(\,\cI\;\;\cI\;\;\cI\,\right)\;\; \bar{\TT}_{(\rr - r_1\ee_1)}^+\,\right) \,( \bar{\TT}_{r_1\ee_1} \circ \bar{\pp} )\\
&\stackrel{(\ref{ll2})}{=\joinrel=} \left(\left({\TT}_{(\rr - r_1\ee_1)}^{(\frac{1}{4})} \left(\,\cI\;\;\cI\;\;\cI\,\right)\;\; {\TT}_{(\rr - r_1\ee_1)}^+\,\right) \circ  \bar{\TT}_{r_1\ee_1}\right)  \bar{\pp} \\
&\stackrel{(\ref{ll0})}{=\joinrel=} 
\left({\TT}_{(\rr - r_1\ee_1)}^{(\frac{1}{4})}\left( \,\cI\;\;\cI\;\;\cI\,\right)\circ  \bar{\TT}_{r_1\ee_1}^{(\frac{3}{4})}\;\; {\TT}_{(\rr - r_1\ee_1)}^+\circ  \bar{\TT}_{r_1\ee_1}^+\,\right)  \bar{\pp} \\
&= {\TT}_{(\rr - r_1\ee_1)}^{(\frac{1}{4})}
\Big(\bar{\theta}_{1,r_1}^-\cI\;\;\; \bar{\theta}_{1,r_1}^+\cI \;\;\;\bar{\theta}_{1,r_1}^-\cI\Big) \bar{\pp}_{\bg^{-11}} + {\TT}_{(\rr - r_1\ee_1)}^+\circ  \bar{\TT}_{r_1\ee_1}^+\bar{\pp}_{\bg^{11}}\\
&= {\TT}_{(\rr - r_1\ee_1)}^{(\frac{1}{4})}
\Big(\bar{\theta}_{1,r_1}^-\cI\;\;\; \bar{\theta}_{1,r_1}^+\cI \;\;\;\bar{\theta}_{1,r_1}^-\cI\Big) \bar{\pp}_{\bg^{-11}} + {\TT}_{(\rr - r_1\ee_1)}^+\circ  {\TT}_{r_1\ee_1}^+{\pp}_{\bg^{11}}.
\end{align*}
Then $\Big(\bar{\ttt}_1^- \otimes \cI\;\;\; \bar{\ttt}_1^+ \otimes \cI \;\;\;\bar{\ttt}_1^- \otimes \cI\Big)\bar{\pp}_{\bg^{-11}} = \Big(\ttt_1^- \otimes \cI\;\;\; \ttt_1^+ \otimes \cI \;\;\;{\ttt}_1^- \otimes \cI\Big){\pp}_{\bg^{-11}}$ in equation
(\ref{eq:mct2}) guarantees that for $\forall r_1 \in [H_1]$,
\begin{equation*}
\Big(\bar{\theta}_{1,r_1}^-\cI\;\;\; \bar{\theta}_{1,r_1}^+\cI \;\;\;\bar{\theta}_{1,r_1}^-\cI\Big) \bar{\pp}_{\bg^{-11}}
 = 
\Big({\theta}_{1,r_1}^-\cI\;\;\; {\theta}_{1,r_1}^+\cI \;\;\;{\theta}_{1,r_1}^-\cI\Big) {\pp}_{\bg^{-11}}
\end{equation*}
Therefore, $\bar{\TT}_{\rr}\bar{\pp} = \TT_{\rr}\pp$.
\item Similarly, for any $\rr$ s.t. $r_1 = 0, r_2 \neq 0$, $\Big(\bar{\ttt}_2^- \otimes \cI\;\;\;\bar{\ttt}_2^- \otimes \cI\;\;\; \bar{\ttt}_2^+ \otimes \cI\Big)\bar{\pp}_{\bg^{-11}} = \Big(\ttt_2^- \otimes \cI\;\;\;\ttt_2^- \otimes \cI\;\;\; \ttt_2^+ \otimes \cI\Big){\pp}_{\bg^{-11}}$ gives $\bar{\TT}_{\rr}\bar{\pp} = \TT_{\rr}\pp$.
\item Similarly, for any $\rr$ s.t. $r_1 \neq 0, \; r_2 \neq 0$, to ensure $\bar{\TT}_{\rr}\bar{\pp} = \TT_{\rr}\pp$, it suffices to have 
\[\Big(\bar{\ttt}_1^- \otimes \bar{\ttt}_2^- \otimes \cI\;\;\;\; \bar{\ttt}_1^+ \otimes \bar{\ttt}_2^-  \otimes \cI\;\;\;\; \bar{\ttt}_1^- \otimes \bar{\ttt}_2^+  \otimes \cI\Big)\bar{\pp}_{\bg^{-11}} 
    = \Big(\ttt_1^- \otimes \ttt_2^- \otimes \cI\;\;\; \ttt_1^+ \otimes \ttt_2^- \otimes \cI\;\;\; \ttt_1^- \otimes \ttt_2^+ \otimes \cI\Big){\pp}_{\bg^{-11}}.\]
\end{itemize}
Now we construct $(\bar{\ttt}_1^-, \bar{\ttt}_2^-, \bar{\pp}_{\bg^{00}}, \bar{\pp}_{\bg^{10}}, \bar{\pp}_{\bg^{01}})$ that satisfies $\bar{\mct} \bar{\pp} = \mct \pp$
as follows: let $\pp_{\bg^{10}} = \rho_1 \cdot \pp_{\bg^{00}}, \;\pp_{\bg^{01}} = \rho_2 \cdot \pp_{\bg^{00}}$, $\bar{\pp}_{\bg^{00}} = u \cdot \pp_{\bg^{00}}$, $\bar{\pp}_{\bg^{10}} = v \cdot \pp_{\bg^{00}}$ and $\bar{\pp}_{\bg^{01}} = w \cdot \pp_{\bg^{00}}$. Then $\bar{\mct} \bar{\pp} = \mct \pp$ can be simplified to
\[
\bar{\mct} \begin{pmatrix}
u \cdot \pp_{\bg^{00}}\\
v \cdot \pp_{\bg^{00}}\\
w \cdot \pp_{\bg^{00}}
\end{pmatrix} = \mct \begin{pmatrix}
1 \cdot \pp_{\bg^{00}}\\
\rho_1 \cdot \pp_{\bg^{00}}\\
\rho_2 \cdot \pp_{\bg^{00}}
\end{pmatrix},
\]
i.e., for $l_1 \in [H_1], \;\;l_2 \in [H_2]$,
\begin{align*}
\begin{frame}{}
    \resizebox{1.08\linewidth}{!}{$
\begin{cases}
    u \cI \pp_{\bg^{00}} + v \cI \pp_{\bg^{00}} + w \cI \pp_{\bg^{00}} = \cI \pp_{\bg^{00}} + \rho_1 \cI \pp_{\bg^{00}} + \rho_2 \cI \pp_{\bg^{00}};\\
u \bar{\theta}_{1,l_1}^-\cI \pp_{\bg^{00}} + v \bar{\theta}_{1,l_1}^+\cI \pp_{\bg^{00}} + w \bar{\theta}_{1,l_1}^-\cI \pp_{\bg^{00}} = \theta_{1,l_1}^-\cI \pp_{\bg^{00}} + \rho_1 \theta_{1,l_1}^+\cI \pp_{\bg^{00}} + \rho_2 \theta_{1,l_1}^-\cI \pp_{\bg^{00}};\\
u \bar{\theta}_{2,l_2}^-\cI \pp_{\bg^{00}} + v \bar{\theta}_{2,l_2}^- \cI \pp_{\bg^{00}}  + w \bar{\theta}_{2,l_2}^+ \cI \pp_{\bg^{00}} = \theta_{2,l_2}^-\cI \pp_{\bg^{00}} + \rho_1 \theta_{2,l_2}^-\cI \pp_{\bg^{00}} + \rho_2 \theta_{2,l_2}^+\cI \pp_{\bg^{00}};\\
u \bar{\theta}_{1,l_1}^- \bar{\theta}_{2,l_2}^-\cI \pp_{\bg^{00}} + v \bar{\theta}_{1,l_1}^+ \bar{\theta}_{2,l_2}^- \cI \pp_{\bg^{00}}  + w \bar{\theta}_{1,l_1}^- \bar{\theta}_{2,l_2}^+ \cI \pp_{\bg^{00}}  = \theta_{1,l_1}^- \theta_{2,l_2}^- \cI \pp_{\bg^{00}} + \rho_1 \theta_{1,l_1}^+ \theta_{2,l_2}^-  \cI\pp_{\bg^{00}} + \rho_2 \theta_{1,l_1}^- \theta_{2,l_2}^+  \cI\pp_{\bg^{00}}.
\end{cases} 
$}
\end{frame}
\end{align*}
Then it suffices to have
\begin{align}\label{eq:tt1}
\begin{cases}
    u + v +w = 1 + \rho_1 + \rho_2;\\
u\bar{\theta}_{1, l_1}^-  + v \bar{\theta}_{1,l_1}^+ + w \bar{\theta}_{1,l_1}^- =  \theta_{1,l_1}^- + \rho_1 \theta_{1, l_1}^+ + \rho_2 \theta_{1, l_1}^-;\\
u\bar{\theta}_{2, l_2}^-  + v \bar{\theta}_{2,l_2}^- + w \bar{\theta}_{2,l_2}^+ =  \theta_{2,l_2}^- + \rho_1 \theta_{2, l_2}^- + \rho_2 \theta_{2, l_2}^+;\\
u\bar{\theta}_{1, l_1}^-\bar{\theta}_{2, l_2}^-  + v \bar{\theta}_{1,l_1}^+\bar{\theta}_{2,l_2}^- + w \bar{\theta}_{1,l_1}^-\bar{\theta}_{2,l_2}^+ = \theta_{1,l_1}^- \theta_{2,l_2}^- + \rho_1 \theta_{1, l_1}^+\theta_{2, l_2}^- + \rho_2 \theta_{1, l_1}^-\theta_{2, l_2}^+.\;\;
\end{cases}
\end{align}
Let some $\kappa \in (0,1)$ s.t.
\begin{align}
\begin{pmatrix}
    \theta_{1, l_1}^-\\
    \bar{\theta}_{1, l_1}^-\\
    \theta_{1, l_1}^+\\
    \bar{\theta}_{1, l_1}^+
\end{pmatrix} = \kappa^{l_1-1}
\begin{pmatrix}
    \theta_{1, 1}^-\\
    \bar{\theta}_{1, 1}^-\\
    \theta_{1, 1}^+\\
    \bar{\theta}_{1, 1}^+
\end{pmatrix}, \;\text{and}\;
\begin{pmatrix}
    \theta_{2, l_2}^-\\
    \bar{\theta}_{2, l_2}^-\\
    \theta_{2, l_2}^+\\
    \bar{\theta}_{2, l_2}^+
\end{pmatrix} = \kappa^{l_2-1}
\begin{pmatrix}
    \theta_{2, 1}^-\\
    \bar{\theta}_{2, 1}^-\\
    \theta_{2, 1}^+\\
    \bar{\theta}_{2, 1}^+
\end{pmatrix} \;\;\;\text{for}\;\;l_1 \in [H_1],\;l_2 \in [H_2],
\end{align} 
then equations (\ref{eq:tt1}) are now reduced to 
\begin{align*}
\begin{cases}
    u + v +w = 1 + \rho_1 + \rho_2;\\
    \bar{\theta}_{1, 1}^-  + v \bar{\theta}_{1, 1}^+ + w \bar{\theta}_{1, 1}^- =  \theta_{1, 1}^- + \rho_1 \theta_{1, 1}^+ + \rho_2 \theta_{1, 1}^-;\\
    u\bar{\theta}_{2, 1}^-  + v \bar{\theta}_{2,1}^- + w \bar{\theta}_{2,1}^+ = \theta_{2,1}^- + \rho_1 \theta_{2, 1}^- + \rho_2 \theta_{2, 1}^+;\\
    u\bar{\theta}_{1, 1}^-\bar{\theta}_{2, 1}^-  + v \bar{\theta}_{1, 1}^+\bar{\theta}_{2,1}^- + w \bar{\theta}_{1, 1}^-\bar{\theta}_{2,1}^+ = \theta_{1, 1}^- \theta_{2,1}^- + \rho_1 \theta_{1, 1}^+\theta_{2, 1}^- + \rho_2 \theta_{1, 1}^-\theta_{2, 1}^+.
\end{cases}
\end{align*}
The above equation system contains $(u, v, w, \bar{\theta}_{1, 1}^-, \bar{\theta}_{2, 1}^-)$ five parameters with four constraints, which gives infinitely many solutions. Thus  the model parameters are not identifiable and condition C3 is necessary.
\end{proof}

\subsection{Identifiability of Sequential DINA model}\label{app:seq}
\suSeq*
\begin{proof}
We begin by showing that when $\QQ^1$ meets the conditions S1-S3, $ \bar{\TT}^s\bar{\pp} = \TT^s \pp$ gives $\bar{\pp} = \pp$ and $\bar{\beta}_{j,1}^+ = \beta_{j,1}^+,\; \bar{\beta}_{j,1}^- = \beta_{j,1}^-,$ for $j \in [J]$. 
As discussed in Section~\ref{sec:seq-id}, the parameters $\left(\,\pp,\;(\beta_{j,1}^+)_{j\in [J]},\;(\beta_{j,1}^-)_{j\in [J]} \,\right)$ can be interpreted as parameters in the reduced binary DINA model. In this model, the binary item takes the form $I$(item $j\geq 1$), and the corresponding $\qq$-vectors are $(\qq_{j,1})_{j\in [J]}$, which indicates the attributes required to complete the first categories. Moreover, $\QQ^1$ is the $\QQ$-matrix for this reduced DINA model, so the item parameters only involve $(\beta_{j,1}^+)_{j\in [J]}, (\beta_{j,1}^-)_{j\in [J]}$. Since the $\QQ$-matrix for this reduced binary DINA model satisfies conditions S1-S3,  according to the sufficient condition for the binary DINA model in \citet{id-dina}, parameters $\pp,\;(\beta_{j,1}^+)_{j\in [J]},\;(\beta_{j,1}^-)_{j\in [J]}$ are identified, i.e., $\bar{\pp} = \pp$ and $\bar{\beta}_{j,1}^+ = \beta_{j,1}^+, \bar{\beta}_{j,1}^- = \beta_{j,1}^-$ for $j \in [J]$.

Next we identify $\beta_{j,l}^+, \beta_{j,l}^-, \,$ for $ l > 1$, by induction.
Suppose categories $h \,(h < l)$ of item $j$ have been identified, i.e.,
\begin{equation}\label{eq:assum}
    \bar{\beta}_{j,h}^+ =\beta_{j,h}^+,\;\;\bar{\beta}_{j,h}^- = \beta_{j,h}^-\;\;\text{for}\;h <l.
\end{equation}
For each item $j$, we will use two rows of the $\TT^s$-matrix in $\bar{\TT}^s\bar{\pp} = \TT^s \pp$ to infer that $\bar{\beta}_{j,l}^+ =\beta_{j,l}^+$ and $ \bar{\beta}_{j,l}^- = \beta_{j,l}^-$. 
For item $j$, its category $l$ requires some attribute, and WLOG, suppose this is $\alpha_1$. There are two possible cases.

\paragraph{Case 1: } category $l$ of item $j$ requires solely $\alpha_1$, i.e., $\qq_{j,l} = \ee_1^\top$. 
According to condition S2, there exists some other item $j'$ whose first category requires $\alpha_1$. 
So the $\QQ$-matrix takes the following form:
\begin{equation}\label{q-matrix1}
    \QQ = \begin{pmatrix}
\vdots & \vdots & \vdots & \vdots & \vdots\\
\hdashline[2pt/2pt]
item \;j, category\; l
& 1&0&\ldots&0\\
\hdashline[2pt/2pt]
\vdots & \vdots & \vdots & \vdots & \vdots\\
\hdashline[2pt/2pt]
    item \;j', category\; 1
& 1&&\vv^\top&\\
\hdashline[2pt/2pt]
\vdots & \vdots & \vdots & \vdots & \vdots
\end{pmatrix}.
\end{equation}
In this case, $\xi_{j',1,\aaa} \leq \xi_{j,l,\aaa}$ for all $\aaa$, consider vectors
\begin{align*}
    \TT^s_{l\ee_j} &= P( R_j \geq l\mid \QQ,\bbb^+,\bbb^-), \\
    \TT^s_{l\ee_j+\ee_{j'}} &= P( R_j \geq l, R_{j'} \geq 1\mid \QQ,\bbb^+,\bbb^-).
\end{align*}
According to the assumption in equation (\ref{eq:assum}), we have $t^s_{(l-1)\ee_j,\aaa} = \bar t^s_{(l-1)\ee_j,\aaa}$. Therefore,
\begin{align}\label{x1}
    t^s_{l\ee_j,\aaa} = \begin{cases}
        \beta_{j,l}^+ \,t^s_{(l-1)\ee_j,\aaa}, \;\;\xi_{j,l,\aaa} = 1\\
        \beta_{j,l}^- \,t^s_{(l-1)\ee_j,\aaa}, \;\;\xi_{j,l,\aaa} = 0
    \end{cases}
    ,\;\;\;\;\;
    \bar t^s_{l\ee_j,\aaa} = \begin{cases}
        \bar \beta_{j,l}^+ \,t^s_{(l-1)\ee_j,\aaa}, \;\;\xi_{j,l,\aaa} = 1\\
        \bar\beta_{j,l}^- \,t^s_{(l-1)\ee_j,\aaa}, \;\;\xi_{j,l,\aaa} = 0
    \end{cases}.
\end{align}
Since $\bar{\beta}_{j',1}^+ = \beta_{j',1}^+, \;\bar{\beta}_{j',1}^- = \beta_{j',1}^-$, we have
\begin{align}\label{x2}
    t^s_{l\ee_j+\ee_{j'},\aaa} = \begin{cases} 
{\beta}_{j',1}^+\, t^s_{l\ee_j,\aaa}
,\;\;  &\xi_{j',1,\aaa} = 1\\
{\beta}_{j',1}^-\, t^s_{l\ee_j,\aaa},\;\;  &\xi_{j',1,\aaa} = 0
\end{cases},\;\;\;\;\;
        \bar t^s_{l\ee_j+\ee_{j'},\aaa} = \begin{cases}  {\beta}_{j',1}^+\, \bar t^s_{l\ee_j,\aaa} 
    \;\;  &\xi_{j',1,\aaa} = 1\\
{\beta}_{j',1}^-\, \bar t^s_{l\ee_j,\aaa},\;\;  &\xi_{j',1,\aaa} = 0
\end{cases}.
\end{align}
Since
$\bar{\TT}^s_{\rr}\bar{\pp} = \TT^s_{\rr} \pp$ for $\rr = l\ee_j$ and $\rr = l\ee_j + \ee_{j'}$, $\bar{\pp} = \pp$, we have
\begin{align*}
        &
        \begin{cases}
            \beta_{j',1}^-(\bar{\TT}^s_{l\ee_j}- \TT^s_{l\ee_j})\pp  = 0\\
    (\bar{\TT}^s_{l\ee_j+\ee_{j'}} - \TT^s_{l\ee_j+\ee_{j'}})\pp = 0
        \end{cases}\\
    \Rightarrow\;& [({\beta}_{j',1}^-\bar{\TT}^s_{l\ee_j} - \bar{\TT}^s_{l\ee_j+\ee_{j'}}) - (\beta_{j',1}^-\TT^s_{l\ee_j} - \TT^s_{l\ee_j+\ee_{j'}})]\pp = 0.
\end{align*}
Using equations (\ref{x1}-\ref{x2}), we have $\displaystyle (\beta_{j',1}^- - \beta_{j',1}^+) \sum_{\aaa\,: \;\xi_{j',1,\aaa = 1}}(\bar t^s_{l\ee_j,\aaa} - t^s_{l\ee_j,\aaa}) \; p_{\aaa} = 0$. According to the constructed $\QQ$-matrix~(\ref{q-matrix1}), $\xi_{j',1,\aaa} = 1$ must imply $\xi_{j,l,\aaa} = 1$. Therefore,
\begin{align*}
    (\beta_{j',1}^- - \beta_{j',1}^+)(\bar{\beta}_{j,l}^+ - \beta_{j,l}^+)  \sum_{\aaa\,:\,\xi_{j',1,\aaa = 1}} t^s_{(l-1)\ee_j,\aaa}\; p_{\aaa} = 0,
\end{align*}
from which we conclude that $\bar{\beta}_{j,l}^+ = \beta_{j,l}^+$.

Next consider $\TT^s_{l\ee_j}$: 
with $\bar{\beta}_{j,l}^+ = \beta_{j,l}^+$, equation \eqref{x1} can be written as
\begin{align}\label{x3}
    t^s_{l\ee_j,\aaa} = \begin{cases}
        \beta_{j,l}^+ \,t^s_{(l-1)\ee_j,\aaa}, \;\;\xi_{j,l,\aaa} = 1\\
        \beta_{j,l}^- \,t^s_{(l-1)\ee_j,\aaa}, \;\;\xi_{j,l,\aaa} = 0
    \end{cases}
    ,\;\;\;\;\;
    \bar t^s_{l\ee_j,\aaa} = \begin{cases}
        \beta_{j,l}^+ \,t^s_{(l-1)\ee_j,\aaa}, \;\;\xi_{j,l,\aaa} = 1\\
        \bar\beta_{j,l}^- \,t^s_{(l-1)\ee_j,\aaa}, \;\;\xi_{j,l,\aaa} = 0
    \end{cases}.
\end{align}
With
$\bar{\TT}^s_{l\ee_j}\bar{\pp} = \TT^s_{l\ee_j} \pp$, 
 $\bar{\pp} = \pp$ and equation \eqref{x3},
we have
\begin{align*}
    & (\bar{\TT}^s_{l\ee_j}-\TT^s_{l\ee_j})\pp = 0\\
    \Rightarrow& (\bar{\beta}_{j,l}^- - \beta_{j,l}^-)\cdot \underset{\aaa\,:\, \xi_{j,l,\aaa} = 0}{\sum}t^s_{(l-1)\ee_j,\aaa} \;p_{\aaa} = 0\\
    \Rightarrow&
    \bar{\beta}_{j,l}^- = \beta_{j,l}^-.
\end{align*}

\paragraph{Case 2: }
category $l$ of item $j$ requires some other attribute. WLOG, assume that this is the second attribute $\alpha_2$. Then according to condition S1, 
$\alpha_1$ or $\alpha_2$ is required solely by some other item's first category. WLOG, assume that it is $\alpha_1$, and it is 
item $j'$'s ($j' \neq j$) first category that requires $\alpha_1$,  i.e., $\qq_{j',1} = \ee_1^\top$. So the $\QQ$-matrix can be written as follows:
\begin{equation}\label{q-matrix2}
    \QQ = \begin{pmatrix}
\vdots& \vdots&\vdots&\vdots&\vdots\\
\hdashline[2pt/2pt]
item \;j, category\; l \;& 1&1&\;\;\;\;\;\vv^\top\\
\hdashline[2pt/2pt]
\vdots&\vdots&\vdots&\vdots&\vdots\\
\hdashline[2pt/2pt]
 item \;j', category\; 1\;& 1&0&\ldots&0\\
\hdashline[2pt/2pt]
\vdots&\vdots&\vdots&\vdots&\vdots
\end{pmatrix}.
\end{equation}
In this case, $\xi_{j,l,\aaa} \leq \xi_{j',1,\aaa}$ for all $\aaa$. 
Consider vectors
\begin{align}
    \TT^s_{l\ee_j} &= P( R_j \geq l\mid \QQ,\bbb^+,\bbb^-), \\
    \TT^s_{l\ee_j+\ee_{j'}} &= P( R_j \geq l, R_{j'} \geq 1\mid \QQ,\bbb^+,\bbb^-). 
\end{align}
According to assumption (\ref{eq:assum}), we have $t^s_{(l-1)\ee_j,\aaa} = \bar t^s_{(l-1)\ee_j,\aaa}$. Therefore,
\begin{align}\label{y1}
    t^s_{l\ee_j,\aaa} = \begin{cases}
        \beta_{j,l}^+ \,t^s_{(l-1)\ee_j,\aaa}, \;\;\xi_{j,l,\aaa} = 1\\
        \beta_{j,l}^- \,t^s_{(l-1)\ee_j,\aaa}, \;\;\xi_{j,l,\aaa} = 0
    \end{cases}
    ,\;\;\;\;\;
    \bar t^s_{l\ee_j,\aaa} = \begin{cases}
        \bar \beta_{j,l}^+ \,t^s_{(l-1)\ee_j,\aaa}, \;\;\xi_{j,l,\aaa} = 1\\
        \bar\beta_{j,l}^- \,t^s_{(l-1)\ee_j,\aaa}, \;\;\xi_{j,l,\aaa} = 0
    \end{cases}.
\end{align}
Since $\bar{\beta}_{j',1}^+ = \beta_{j',1}^+, \;\bar{\beta}_{j',1}^- = \beta_{j',1}^-$, we have
\begin{align}\label{y2}
    t^s_{l\ee_j+\ee_{j'},\aaa} = \begin{cases} 
{\beta}_{j',1}^+\, t^s_{l\ee_j,\aaa}
,\;\;  &\xi_{j',1,\aaa} = 1\\
{\beta}_{j',1}^-\, t^s_{l\ee_j,\aaa},\;\;  &\xi_{j',1,\aaa} = 0
\end{cases},\;\;\;\;\;
        \bar t^s_{l\ee_j+\ee_{j'},\aaa} = \begin{cases}  {\beta}_{j',1}^+\, \bar t^s_{l\ee_j,\aaa} 
    \;\;  &\xi_{j',1,\aaa} = 1\\
{\beta}_{j',1}^-\, \bar t^s_{l\ee_j,\aaa},\;\;  &\xi_{j',1,\aaa} = 0
\end{cases}.
\end{align}
 Since
$\bar{\TT}^s_{\rr}\bar{\pp} = \TT^s_{\rr} \pp$ holds for $\rr = l\ee_j$ and $\rr = l\ee_j + \ee_{j'}$, and $\bar{\pp} = \pp$,
\begin{align*}
        &
        \begin{cases}
            \beta_{j',1}^+(\bar{\TT}^s_{l\ee_j}- \TT^s_{l\ee_j})\pp  = 0\\
    (\bar{\TT}^s_{l\ee_j+\ee_{j'}} - \TT^s_{l\ee_j+\ee_{j'}})\pp = 0
        \end{cases}
        \\
    \Rightarrow\;& [({\beta}_{j',1}^+\bar{\TT}^s_{l\ee_j} - \bar{\TT}^s_{l\ee_j+\ee_{j'}}) - (\beta_{j',1}^+\TT^s_{l\ee_j} - \TT^s_{l\ee_j+\ee_{j'}})]\pp = 0.
\end{align*}
Using equations (\ref{y1}-\ref{y2}), we have $\displaystyle (\beta_{j',1}^+ - \beta_{j',1}^-) \sum_{\aaa\,: \;\xi_{j',1,\aaa = 0}}(\bar t^s_{l\ee_j,\aaa} - t^s_{l\ee_j,\aaa}) \; p_{\aaa} = 0$. According to the constructed $\QQ$-matrix~(\ref{q-matrix2}), $\xi_{j',1,\aaa} = 0$ must imply $\xi_{j,l,\aaa} = 0$. Therefore, using equation \eqref{y1}, we have
\begin{align*}
     (\beta_{j',1}^+ - \beta_{j',1}^-)(\bar{\beta}_{j,l}^- - \beta_{j,l}^-)  \sum_{\aaa: \,\xi_{j',1,\aaa = 0}} t^s_{(l-1)\ee_j,\aaa}\; p_{\aaa} = 0,
\end{align*}
which we conclude that $\bar{\beta}_{j,l}^- = \beta_{j,l}^-$.
Next consider $\TT^s_{l\ee_j}$: 
with $\bar{\beta}_{j,l}^- = \beta_{j,l}^-$, equation (\ref{y1}) can be written as
\begin{align}\label{y3}
    t^s_{l\ee_j,\aaa} = \begin{cases}
        \beta_{j,l}^+ \,t^s_{(l-1)\ee_j,\aaa}, \;\;\xi_{j,l,\aaa} = 1\\
        \beta_{j,l}^- \,t^s_{(l-1)\ee_j,\aaa}, \;\;\xi_{j,l,\aaa} = 0
    \end{cases}
    ,\;\;\;\;\;
    \bar t^s_{l\ee_j,\aaa} = \begin{cases}
        \bar\beta_{j,l}^+ \,t^s_{(l-1)\ee_j,\aaa}, \;\;\xi_{j,l,\aaa} = 1\\
        \beta_{j,l}^- \,t^s_{(l-1)\ee_j,\aaa}, \;\;\xi_{j,l,\aaa} = 0
    \end{cases}.
\end{align}
Using 
$\bar{\TT}^s_{l\ee_j}\bar{\pp} = \TT^s_{l\ee_j} \pp$ and 
 $\bar{\pp} = \pp$, and equation (\ref{y3}),
we have
\begin{align*}
    & (\bar{\TT}^s_{l\ee_j}-\TT^s_{l\ee_j})\pp = 0\\
    \Rightarrow& (\bar{\beta}_{j,l}^+ - \beta_{j,l}^+)\cdot \underset{\aaa\,:\, \xi_{j,l,\aaa} = 1}{\sum}t^s_{(l-1)\ee_j,\aaa} \;p_{\aaa} = 0\\
    \Rightarrow&
    \bar{\beta}_{j,l}^+ = \beta_{j,l}^+.
\end{align*}
Therefore, for both cases,
$\beta_{j,l}^+$ and $\beta_{j,l}^-$ are identified. By induction, we conclude that all parameters $(\bbb^+, \bbb^-, \pp)$ are identifiable and conditions S1-S3 are sufficient.
\end{proof}

\neSeq*
\begin{proof}
\;\;
We have shown the necessity of condition S1 in Proposition~\ref{neSeq1}, so we may assume that the $\QQ$-matrix satisfies condition S1.

\paragraph{Necessity of Condition S2$^*$. }

Suppose the $\QQ^1$ matrix satisfies condition S1, but does not satisfy condition S2$^*$.
We first show that each attribute must be required by more than one item. Suppose there exists some attribute that is only required by one item. WLOG, assume this is attribute one $\alpha_1$, and is only required by the first item. 
So the $\QQ$-matrix can be written as
\begin{equation}
  \QQ = \begin{pmatrix}
  \text{item } 1\quad\quad
  \begin{pmatrix}
  1 & \zero^\top \\
  * & *
  \end{pmatrix}\\
  \hdashline[2pt/2pt]
  \text{item }(2:J) \;
  \begin{pmatrix}
  \;\zero \;\;& *
  \end{pmatrix}
  \end{pmatrix}. 
\end{equation}
We partition $\aaa$ into two groups according to the first attribute:
\begin{align*}
    \bg^0 = \{\aaa: \alpha_1 = 0\} = \{\aaa = (0, \aaa^*),\;\aaa^* \in \{0, 1\}^{K-1}\},\\
    \bg^1 = \{\aaa: \alpha_1 = 1\} = \{\aaa = (1, \aaa^*), \;\aaa^* \in \{0, 1\}^{K-1}\}.
\end{align*} 
So each group has $2^{K-1}$ attribute profiles, and we index the entries in each group by
\begin{align*}
    \bg^0_1 = (0, \zero), \; \bg^0_2 = (0, \ee_1), \ldots, \bg^0_K = (0, \ee_{K-1}), \; \bg^0_{K+1} = (0, \ee_1+\ee_2), \ldots, \bg^0_{2^{K-1}} = \left(0, \sum_{k=1}^{K-1} \ee_k\right), \\
    \bg^1_1 = (1, \zero), \; \bg^1_2 = (1, \ee_1), \ldots, \bg^1_K = (1, \ee_{K-1}), \; \bg^1_{K+1} = (1, \ee_1+\ee_2), \ldots, \bg^1_{2^{K-1}} = \left(1, \sum_{k=1}^{K-1} \ee_k\right),
\end{align*}
where $\ee_1, \ldots, \ee_{K-1} \in \{0,1\}^{K-1}$ have $K-1$ elements.
Therefore, the $k$-th ($k \in [2^{K-1}]$) entry of $\bg^0$ and $\bg^1$:
$\bg^0_k$ and $\bg^1_k$, share the same attributes except for the first one $\alpha_1$. Index the population proportion parameters $\pp$ in the following way:
\begin{equation*}
  \pp = \begin{pmatrix}
\pp_{\bg^0}\\
\pp_{\bg^1}\\
\end{pmatrix},\;\;\;
\text{where}\;\;
\pp_{\bg^0} = \begin{pmatrix}
p_{\bg^0_1}\\
p_{\bg^0_2}\\
\vdots\\
p_{\bg^0_{2^{K-1}}}\\
\end{pmatrix}\;\;\;\text{and}\;\;
\pp_{\bg^1} = \begin{pmatrix}
p_{\bg^1_1}\\
p_{\bg^1_2}\\
\vdots\\
p_{\bg^1_{2^{K-1}}}\\
\end{pmatrix}.
\end{equation*}
Recall that $\bbb_j^+ = \left(\beta_{j,1}^+,\; \beta_{j,1}^+\beta_{j,2}^+,\; \ldots \;\prod_{l=1}^{H_j}\, \beta_{j,l}^+\right)$, $\bbb_j^- = \left(\beta_{j,1}^-,\; \beta_{j,1}^-\beta_{j,2}^-,\; \ldots \;\prod_{l=1}^{H_j}\, \beta_{j,l}^-\right)$, for $j \in [J]$. Item parameters
$\bbb^+ = (\bbb_1^+, \bbb_2^+, \ldots, \bbb_J^+)$ and $\bbb^- = (\bbb_1^-, \bbb_2^-, \ldots, \bbb_J^-)$. We now seek to construct $(\bbb^+,\bbb^-,\pp) \neq (\bar{\bbb}^+,\bar{\bbb}^-,\bar{\pp})$ such that (\ref{eq-orig-seq}) holds:
\begin{enumerate}
    \item Take $\bbb_{j}^+ = \bar{\bbb}_{j}^+,\;\bbb_{j}^- = \bar{\bbb}_{j}^-$ for $j > 1$.
    \item $\beta_{1,1}^- = \bar{\beta}_{1,1}^- = 0$, $\bar{\beta}_{1,l_1}^+ = \beta_{1,l_1}^+$, $\bar{\beta}_{1,l_1}^- = \beta_{1,l_1}^-$ for $l_1 > 1$.
    \item  $\bar{\pp}_{\bg^1} =\rho \cdot \bar{\pp}_{\bg^0}$, ${\pp}_{\bg^0} = u \cdot \bar{\pp}_{\bg^0}$, and ${\pp}_{\bg^1} = v \cdot \bar{\pp}_{\bg^0}$.
\end{enumerate}
According to Lemma~\ref{lem:ts-matrix}, in order for equation~(\ref{eq-orig-seq}) to hold, it suffices to show that $\bar{\TT}_{\rr}\bar{\pp} = \TT_{\rr}\pp$ holds for $\forall \rr$.

\begin{itemize}
\item For any $\rr$ s.t. $r_1 = 0$, 
$t_{\rr,\, \bg^0_k} \equiv t_{\rr, \,\bg^1_k} \; \text{for}\; k \in [2^{K-1}]$.
Since $\bbb_{j}^+ = \bar{\bbb}_{j}^+,\;\bbb_{j}^- = \bar{\bbb}_{j}^-$ for $j > 1$, 
$t_{\rr,\, \bg^0_k} \equiv t_{\rr, \,\bg^1_k} \equiv \bar{t}_{\rr,\, \bg^0_k} \equiv \bar{t}_{\rr, \,\bg^1_k} \; \text{for}\; k \in [2^{K-1}]$,
to ensure $\bar{\TT}_{\rr}\bar{\pp} = \TT_{\rr}\pp$,
it suffices to have $1+\rho = u+v$.
\item For any $\rr$ s.t. $r_1 = 1$, since $\beta_{1,1}^- = 0$,
$t_{\rr,\, \bg^0_k} = 0$, it suffices to have $\rho \bar{\beta}_{1,1}^+ = v \beta_{1, 1}^+$ to
guarantee $\bar{\TT}_{\rr}\bar{\pp} = \TT_{\rr}\pp$.
\item For any $\rr$ s.t. $r_1 > 1$, since $\beta_{1,1}^- = 0$ and $\bar{\beta}_{1,l_1}^+ = \beta_{1,l_1}^+,\;\bar{\beta}_{1,l_1}^- = \beta_{1,l_1}^-$ for $l_1 > 1$, $\bar{\TT}_{\rr}\bar{\pp} = \TT_{\rr}\pp$ holds without additional conditions.
\end{itemize}
With three parameters $(u, v, \beta_{1,1}^+)$ and two constraints, the equation system has infinitely many solutions. Thus the construction exists. Therefore, each attribute must be required by more than one item.

Now suppose the $\QQ$-matrix satisfies condition S1 and each attribute is required by at least two items, and suppose that there exists some attribute which is required by at most two categories. WLOG, assume this is $\alpha_1$, and it is the first and second items that require $\alpha_1$. Assume that category $l_2^* \in [H_2]$ of item 2 requires $\alpha_1$ and other categories of item 2 do not require.
So the $\QQ$ can be written as
\begin{equation}
  \QQ = \begin{pmatrix}
  \text{item \;1} &\begin{cases}
  1  \;\;\;\;&\zero^\top \\
  \zero \;\;&* 
  \end{cases}\\
  \hdashline[2pt/2pt]
  \text{item \;2} &
  \begin{cases}
   \zero  &* \\
  1 &\vv^\top\\
  \zero  &* 
  \end{cases}\\
  \hdashline[2pt/2pt]
  \text{item \;3:J} &\begin{cases}
    \zero &\; * \;
  \end{cases}\\
  \end{pmatrix}.
\end{equation}
We now seek to construct $(\bar{\bbb}^+,\bar{\bbb}^-,\bar{\pp})\neq (\bbb^+,\bbb^-,\pp)$ such that $\bar{\TT}_{\rr}\bar{\pp} = \TT_{\rr}\pp$ holds for $\forall \rr$.
\begin{enumerate}
    \item Take $\bbb_{j}^+ = \bar{\bbb}_{j}^+,\;\bbb_{j}^- = \bar{\bbb}_{j}^-$ for $j > 2$.
    \item For $l_1 > 1$,  $\bar{\beta}_{1,l_1}^- = \beta_{1,l_1}^-$ and $\bar{\beta}_{1,l_1}^+ = \beta_{1,l_1}^+$. 
    \item For $l_2 \neq l_2^*$, $\bar{\beta}_{2,l_2}^+ = \beta_{2,l_2}^+$,
    for $l_2\in[H_2]$, $\bar{\beta}_{2,l_2}^- = \beta_{2,l_2}^-$.
    \item  $\bar{\pp}_{\bg^1} =\rho \cdot \bar{\pp}_{\bg^0}$, ${\pp}_{\bg^0} = u \cdot \bar{\pp}_{\bg^0}$, and ${\pp}_{\bg^1} = v \cdot \bar{\pp}_{\bg^0}$.
\end{enumerate}
So the remaining parameters are $(\beta_{1,1}^+, \beta_{1,1}^-, \beta_{2,l_2^*}^+,
 \rho, u, v)$. We partition $\aaa$ into two groups according to the first attribute as we did before. 

\begin{itemize}
\item For any $\rr$ s.t. $r_1 =0,  r_2 < l_2^*$, 
$t_{\rr,\, \bg^0_k} \equiv t_{\rr, \,\bg^1_k} \; \text{for}\; k \in [2^{K-1}]$.
Since $\bbb_{j}^+ = \bar{\bbb}_{j}^+,\;\bbb_{j}^- = \bar{\bbb}_{j}^-$ for $j > 1$,
$t_{\rr,\, \bg^0_k} \equiv t_{\rr, \,\bg^1_k} \equiv \bar{t}_{\rr,\, \bg^0_k} \equiv \bar{t}_{\rr, \,\bg^1_k}\; \text{for}\; k \in [2^{K-1}]$. To ensure $\bar{\TT}_{\rr}\bar{\pp} = \TT_{\rr}\pp$,
it suffices to have $1+\rho = u+v$.
\item For any $\rr$ s.t. $r_1 =0,  r_2 = l_2^*$, 
it suffices to have $\bar{\beta}_{2, l_2^*}^- + \rho \bar{\beta}_{2, l_2^*}^+ = u\beta_{2, l_2^*}^- + v \beta_{2, l_2^*}^+$. When this is met, for any $\rr$ s.t. $r_1 = 0$, $r_2 > l_2^*$, $\bar{\TT}_{\rr}\bar{\pp} = \TT_{\rr}\pp$ also holds, since for $l_2 > l_2^*$, $\bar{\beta}_{2,l_2}^+ = \beta_{2,l_2}^+,\;\bar{\beta}_{2,l_2}^- = \beta_{2,l_2}^-$.
\item For any $\rr$ s.t. $r_1 = 1$, $r_2 = 0$,
it suffices to have $\bar{\beta}_{1,1}^- + \rho \bar{\beta}_{1,1}^+ = u\beta_{1, 1}^- + v \beta_{1, 1}^+$ to
guarantee $\bar{\TT}_{\rr}\bar{\pp} = \TT_{\rr}\pp$. When this is met, for any $\rr$ s.t. $r_1 > 1$, $r_2 = 0$, $\bar{\TT}_{\rr}\bar{\pp} = \TT_{\rr}\pp$ also holds. Since for $l_1 > 1$,  $\bar{\beta}_{1,l_1}^- = \beta_{1,l_1}^-$ and $\bar{\beta}_{1,l_1}^+ = \beta_{1,l_1}^+$. 

\item For any $\rr$ s.t. $r_1 \neq 0$ and $r_2 \neq 0$, we need $\bar{\beta}_{1,1}^-\bar{\beta}_{2, l_2^*}^- + \rho \bar{\beta}_{1,1}^+ \bar{\beta}_{2, l_2^*}^+ = u\beta_{1, 1}^-\beta_{2, l_2^*}^- + v \beta_{1, 1}^+\beta_{2, l_2^*}^+$.
\end{itemize}
With six parameters and four constraints, the equation system has infinitely many solutions. Thus the construction exists. Thus the parameters are not identifiable and the condition S$2^*$ is necessary.

\paragraph{Necessity of Condition S3$^*$.}
Suppose $\QQ$-matrix satisfies condition S1 and S$2^*$, but does not satisfy condition S$3^*$, i.e., there exists two columns of the matrix
\[
\begin{pmatrix}
    \QQ_{1:K}^{-1}\\
    \QQ_{K+1:J}
    \end{pmatrix}
\]
are the same, and WLOG, assume they are columns 1 and 2. Since $\QQ$ satisfies S1, $\QQ^1_{1:K} = \cI_K$, so the $\QQ$-matrix can be written as 
\[
\QQ = 
\begin{pmatrix}
    \QQ_1\\
    \hdashline\QQ_2\\
    \hdashline
    \QQ_{3:J}
\end{pmatrix}=
\begin{pmatrix}
\text{item \;1\;\;\;\;\;\;\;\;\;\;} \begin{cases}
    1 &0 \;\; \ldots\\
    \vv_1 & \vv_1 \;\;\ldots
    \end{cases}\\
    \hdashline 
    \text{item \;2\;\;\;\;\;\;\;\;\;\;} \begin{cases}
    0 &1 \;\; \ldots\\
    \vv_2 & \vv_2 \;\;\ldots
    \end{cases}\\
    \hdashline
    \text{item }(3:J) \begin{cases}
    \vv & \vv \;\;\ldots
    \end{cases}
    \end{pmatrix}
    \begin{array}{cc}
         &  \\
         & 
    \end{array}.
\]
We partition $\aaa$ into four groups according to the first and the second attribute:
\begin{align*}
    \bg^{00} &= \{\aaa: \alpha_1 = 0, \alpha_2 = 0\} = \{\aaa = (0, 0, \aaa^*),\;\aaa^* \in \{0, 1\}^{K-2}\},\\
    \bg^{10} &= \{\aaa: \alpha_1 = 1, \alpha_2 = 0\} = \{\aaa = (1, 0, \aaa^*),\;\aaa^* \in \{0, 1\}^{K-2}\},\\
    \bg^{01} &= \{\aaa: \alpha_1 = 0, \alpha_2 = 1\} = \{\aaa = (0, 1, \aaa^*),\;\aaa^* \in \{0, 1\}^{K-2}\},\\
    \bg^{11} &= \{\aaa: \alpha_1 = 1,  \alpha_2 = 1\} = \{\aaa = (1, 1, \aaa^*),\;\aaa^* \in \{0, 1\}^{K-2}\}.
\end{align*} 
So each group has $2^{K-2}$ attribute profiles, and we index the entries in each group by
\begin{align*}
    \resizebox{1.05\linewidth}{!}{$ 
\bg^{00}_1 = (0, 0, \zero), \; \bg^{00}_2 = (0, 0, \ee_1), \ldots, \bg^{00}_K = (0, 0, \ee_{K-2}), \; \bg^{00}_{K+1} = (0, 0, \ee_1+\ee_2), \ldots, \bg^{00}_{2^{K-2}} = \left(0, 0, \sum_{k=1}^{K-2} \ee_k\right),
$}
\end{align*}
where $\ee_1, \ldots, \ee_{K-2} \in \{0,1\}^{K-2}$ have $K-2$ elements. 
Similarly we index the elements of $\bg^{10},\; \bg^{01},\; \bg^{11}$.
Therefore, $\bg^{00}_k$, $\bg^{10}_k$, $\bg^{01}_k$ and $\bg^{11}_k$ for $k \in [2^{K-2}]$ share the same attributes except for the first and second attributes. Index the population proportion parameters $\pp$ in the following way:
\begin{equation*}
    \resizebox{1.06\linewidth}{!}{$
\pp = \begin{pmatrix}
\pp_{\bg^{00}}\\
\pp_{\bg^{10}}\\
\pp_{\bg^{01}}\\
\pp_{\bg^{11}}
\end{pmatrix},\;\;\;
\text{where}\;\;
\pp_{\bg^{00}} = 
\begin{pmatrix}
p_{\bg^{00}_1}\\
p_{\bg^{00}_2}\\
\vdots\\
p_{\bg^{00}_{2^{K-2}}}\\
\end{pmatrix},\;
\pp_{\bg^{10}} = 
\begin{pmatrix}
p_{\bg^{10}_1}\\
p_{\bg^{10}_2}\\
\vdots\\
p_{\bg^{10}_{2^{K-2}}}\\
\end{pmatrix},\;
\pp_{\bg^{01}} = 
\begin{pmatrix}
p_{\bg^{01}_1}\\
p_{\bg^{01}_2}\\
\vdots\\
p_{\bg^{01}_{2^{K-2}}}\\
\end{pmatrix}\;\text{and}\;\,
\pp_{\bg^{11}} = 
\begin{pmatrix}
p_{\bg^{11}_1}\\
p_{\bg^{11}_2}\\
\vdots\\
p_{\bg^{11}_{2^{K-2}}}\\
\end{pmatrix}.\;
$}
\end{equation*}
We now seek to construct $(\bar{\bbb}^+,\bar{\bbb}^-,\bar{\pp})\neq (\bbb^+,\bbb^-,\pp)$ such that $\bar{\TT}_{\rr}\bar{\pp} = \TT_{\rr}\pp$ holds for $\forall \rr$.
\begin{enumerate}
    \item Take $\bbb_{j}^+ = \bar{\bbb}_{j}^+,\;\bbb_{j}^- = \bar{\bbb}_{j}^-$ for $j > 2$.
    \item For $l_1 > 1$,  $\bar{\beta}_{1,l_1}^- = \beta_{1,l_1}^-$ and $\bar{\beta}_{1,l_1}^+ = \beta_{1,l_1}^+$; $\bar{\beta}_{1,1}^+ = \beta_{1,1}^+$. 
    \item For $l_2 > 1$, $\bar{\beta}_{2,l_2}^+ = \beta_{2,l_2}^+$
    and $\bar{\beta}_{2,l_2}^- = \beta_{2,l_2}^-$;
    $\bar{\beta}_{2,1}^+ = \beta_{2,1}^+$.
    \item 
    $\pp_{\bg^{11}} = \bar{\pp}_{\bg^{11}}$,\;
    $\pp_{\bg^{00}} = \rho_1 \pp_{\bg^{11}}$,\;
    $\pp_{\bg^{10}} = \rho_2 \pp_{\bg^{11}}$, \; $\pp_{\bg^{01}} = \rho_3 \pp_{\bg^{11}}$, \;\\
    $\bar{\pp}_{\bg^{00}} = u_1\pp_{\bg^{11}}$,\;
    $\bar{\pp}_{\bg^{10}} = u_2\pp_{\bg^{11}}$\; and \;
    $\bar{\pp}_{\bg^{01}} = u_3\pp_{\bg^{11}}$.
\end{enumerate}
So the remaining parameters are $(\bar\beta_{1,1}^-, \bar\beta_{2,1}^-, u_1, u_2, u_3)$. 

\begin{itemize}
\item For any $\rr$ s.t. $r_1 =0,  r_2 = 0$, the response does not require $\alpha_1$ and $\alpha_2$, so 
$t_{\rr,\, \bg^{00}_k} \equiv t_{\rr, \,\bg^{10}_k} \equiv t_{\rr, \,\bg^{01}_k} \equiv t_{\rr, \,\bg^{11}_k} \; \text{for}\; k \in [2^{K-2}]$.
Since $\bbb_{j}^+ = \bar{\bbb}_{j}^+,\;\bbb_{j}^- = \bar{\bbb}_{j}^-$ for $j > 1$,
\[
t_{\rr,\, \bg^{00}_k} \equiv t_{\rr, \,\bg^{10}_k} \equiv t_{\rr, \,\bg^{01}_k} \equiv t_{\rr, \,\bg^{11}_k} \equiv 
\bar{t}_{\rr,\, \bg^{00}_k} \equiv \bar{t}_{\rr, \,\bg^{10}_k} \equiv \bar{t}_{\rr, \,\bg^{01}_k} \equiv \bar{t}_{\rr, \,\bg^{11}_k},\;\;
\text{for}\; k \in [2^{K-2}].
\]
To ensure $\bar{\TT}_{\rr}\bar{\pp} = \TT_{\rr}\pp$,
it suffices to have $\rho_1 + \rho_2 + \rho_3 = u_1 + u_2 +u_3$.
\item For any $\rr$ s.t. $r_1 =1,  r_2 = 0$, 
it suffices to have 
$\rho_1\beta_{1,1}^- + \rho_2\beta_{1,1}^+ + \rho_3\beta_{1,1}^- + \beta_{1,1}^+ = 
u_1\bar{\beta}_{1,1}^- + u_2\bar{\beta}_{1,1}^+ + u_3\bar{\beta}_{1,1}^- + \bar{\beta}_{1,1}^+ 
$ to ensure $\bar{\TT}_{\rr}\bar{\pp} = \TT_{\rr}\pp$. Since $\beta_{1,1}^+ = \bar{\beta}_{1,1}^+$, it suffices to have
$\rho_1\beta_{1,1}^- + \rho_2\beta_{1,1}^+ + \rho_3\beta_{1,1}^- = 
u_1\bar{\beta}_{1,1}^- + u_2\bar{\beta}_{1,1}^+ + u_3\bar{\beta}_{1,1}^-
$. 
When this is met, for any $\rr$ s.t. $r_1 > 1$, $r_2 = 0$, $\bar{\TT}_{\rr}\bar{\pp} = \TT_{\rr}\pp$ also holds, since for $l_1 > 1$,  $\bar{\beta}_{1,l_1}^- = \beta_{1,l_1}^-$ and $\bar{\beta}_{1,l_1}^+ = \beta_{1,l_1}^+$.
\item For any $\rr$ s.t. $r_1 = 0$, $r_2 = 1$,
it suffices to have $\rho_1\beta_{2,1}^- + \rho_2\beta_{2,1}^- + \rho_3\beta_{2,1}^+ + \beta_{2,1}^+ = 
u_1\bar{\beta}_{2,1}^- + u_2\bar{\beta}_{2,1}^- + u_3\bar{\beta}_{2,1}^+ + \bar{\beta}_{2,1}^+ 
$ to ensure $\bar{\TT}_{\rr}\bar{\pp} = \TT_{\rr}\pp$. Since $\beta_{2,1}^+  = \bar{\beta}_{2,1}^+$, it suffices to have 
$\rho_1\beta_{2,1}^- + \rho_2\beta_{2,1}^- + \rho_3\beta_{2,1}^+ = 
u_1\bar{\beta}_{2,1}^- + u_2\bar{\beta}_{2,1}^- + u_3\bar{\beta}_{2,1}^+$.
When this is met, for any $\rr$ s.t. $r_1 = 0$, $r_2 > 1$, $\bar{\TT}_{\rr}\bar{\pp} = \TT_{\rr}\pp$ also holds, since for $l_2 > 1$,  $\bar{\beta}_{2,l_2}^- = \beta_{2,l_2}^-$ and $\bar{\beta}_{2,l_2}^+ = \beta_{2,l_2}^+$. 

\item For any $\rr$ s.t. $r_1 \neq 0$ and $r_2 \neq 0$, we also need $\rho_1\beta_{1,1}^-\beta_{2,1}^- + \rho_2\beta_{1,1}^+\beta_{2,1}^- + \rho_3\beta_{1,1}^-\beta_{2,1}^+ = 
u_1\bar{\beta}_{1,1}^-\bar{\beta}_{2,1}^- + u_2\bar{\beta}_{1,1}^+\bar{\beta}_{2,1}^- + u_3\bar{\beta}_{1,1}^-\bar{\beta}_{2,1}^+$.
\end{itemize}
With five parameters and four constraints, the equation system has infinitely many solutions, thus the construction exists. So the model parameters are not identifiable and condition S$3^*$ is necessary.

\end{proof}

\subsection{Proofs of Examples}\label{app:ex}
\vspace{8mm}
Our proofs utilize certain results from existing literature, which we have summarized as lemmas below.
\begin{lemma}\label{lem-ex1}
    When $K = 1$, the parameters of the binary DINA model with the following $\QQ$-matrix are identifiable.
    \begin{equation}
        \QQ = \begin{pmatrix}
            1\\
            \hdashline[2pt/2pt]
            1\\
            \hdashline[2pt/2pt]
            1
        \end{pmatrix}
    \end{equation}
\end{lemma}
This lemma is a direct result from Theorem 1 in \citet{id-dina}.
\begin{lemma}\label{lem-ex2}
    Given that the item parameters are known(identified), i.e., only the population proportion parameters $\pp$ need to be identified, 
    and if the $\QQ$-matrix contains an identity matrix, then the binary DINA model is identifiable.
\end{lemma}
This lemma is a result from  Theorem 1 in \citet{xu2016}.
\subsubsection{Proof of Example \ref{counterex1}}
Assuming that $0 < \beta_{j,l}^- < \beta_{j,l}^+ \leq 1$, the Sequential DINA model parameters with the following $\QQ$-matrix are identifiable: 
\begin{equation}\label{joke2}
    \QQ = \begin{pmatrix}
    item 1 \begin{cases}
        1 & 1\\
        0 & 1
    \end{cases}\\
        \hdashline[2pt/2pt]
        item 2 \begin{cases}
        1 & 1\\
        \end{cases}\\
        \hdashline[2pt/2pt]
        item 3 \begin{cases}
        1 & 1\\
        \end{cases}\\
        \hdashline[2pt/2pt]
        item 4 \begin{cases}
        1 & 0\\
        \end{cases}\\
        \hdashline[2pt/2pt]
        item 5 \begin{cases}
        1 & 0\\
        \end{cases}\\
        \hdashline[2pt/2pt]
        item 6 \begin{cases}
        1 & 0\\
        \end{cases}\\
\end{pmatrix} \text{and}\;\;
    \QQ^1 = \begin{pmatrix}
         1 & 1\\
        \hdashline[2pt/2pt]
        1 & 1\\
        \hdashline[2pt/2pt]
        1 & 1\\
        \hdashline[2pt/2pt]
        1 & 0\\
        \hdashline[2pt/2pt]
        1 & 0\\
        \hdashline[2pt/2pt]
        1 & 0\\
    \end{pmatrix}.
\end{equation}
\begin{proof}
If we consider the first categories of the first three items, whose $\QQ$-matrix can be written as
\[
\QQ^1_{1:3} = \begin{pmatrix}
    1 & 1\\
    \hdashline[2pt/2pt]
    1 & 1\\
    \hdashline[2pt/2pt]
    1 & 1
\end{pmatrix}.
\]
and we regard $(1,1)$ together as ``one attribute", then the above $\QQ$-matrix can be viewed as the $\QQ$-matrix for the binary DINA model with only one attribute $(1,1)$. Using Lemma~\ref{lem-ex1} and the equations $\TT^s_{\rr} \pp = \bar\TT^s_{\rr} \bar\pp$ with $\rr = \ee_j,\;j \in [3]$ and their sums, we have $p_{(11)} = \bar p_{(11)}$, $\beta_{j,1}^+ = \bar \beta_{j,1}^+$  and $\beta_{j,1}^- = \bar \beta_{j,1}^-$ for $j \in [3]$. Similarly if we consider the first categories of the last three items, we obtain  $p_{(10)} + p_{(11)} = \bar p_{(10)} + \bar p_{(11)}$ and $\beta_{j,1}^+ = \bar \beta_{j,1}^+$, $\beta_{j,1}^- = \bar \beta_{j,1}^-$ for $j \in \{4,5,6\}$. 
Therefore, $p_{(10)} = \bar p_{(10)}$.
Next we identify $\beta_{1,2}^+,\,\beta_{1,2}^-,\,p_{(00)},\,p_{(01)}$: consider equations $\TT^s_{\rr} \pp = \bar\TT^s_{\rr} \bar\pp$ with $\rr\in \{ 2\ee_1,\;\ee_2, \; \ee_4\}$ and their sums.
With $p_{(11)} = \bar p_{(11)}$, $p_{(10)} = \bar p_{(10)}$ and $\beta_{j,1}^+ = \bar \beta_{j,1}^+$, $\beta_{j,1}^- = \bar \beta_{j,1}^-$ for $j \in \{1,2,4\}$, these equations give
\begin{equation*}
    \resizebox{\linewidth}{!}{$
    \begin{pmatrix}
    0 & 1& 0& 0 & 0 & 1& 0& 0 \\
    0 & 0& 0& 1 & 0 & 0& 0& 1 \\
    \beta_{1,1}^- & \beta_{1,1}^-& \beta_{1,1}^-& \beta_{1,1}^+ & \beta_{1,1}^- & \beta_{1,1}^-& \beta_{1,1}^-& \beta_{1,1}^+\\
    \beta_{1,1}^-\beta_{1,2}^- & \beta_{1,1}^-\beta_{1,2}^-& \beta_{1,1}^-\beta_{1,2}^+& \beta_{1,1}^+\beta_{1,2}^+ &\beta_{1,1}^-\bar\beta_{1,2}^- & \beta_{1,1}^-\bar\beta_{1,2}^-& \beta_{1,1}^-\bar\beta_{1,2}^+& \beta_{1,1}^+\bar\beta_{1,2}^+ \\
    \beta_{2,1}^-\beta_{1,1}^- & \beta_{2,1}^-\beta_{1,1}^-& \beta_{2,1}^-\beta_{1,1}^-& \beta_{2,1}^+\beta_{1,1}^+ &\beta_{2,1}^-\beta_{1,1}^- & \beta_{2,1}^-\beta_{1,1}^-& \beta_{2,1}^-\beta_{1,1}^-& \beta_{2,1}^+\beta_{1,1}^+ \\
    \beta_{4,1}^-\beta_{1,1}^-& \beta_{4,1}^+\beta_{1,1}^-& \beta_{4,1}^-\beta_{1,1}^-& \beta_{4,1}^+\beta_{1,1}^+ &\beta_{4,1}^-\beta_{1,1}^- & \beta_{4,1}^+\beta_{1,1}^-& \beta_{4,1}^-\beta_{1,1}^-& \beta_{4,1}^+\beta_{1,1}^+ \\
    \beta_{2,1}^-\beta_{1,1}^-\beta_{1,2}^- & \beta_{2,1}^-\beta_{1,1}^-\beta_{1,2}^-& \beta_{2,1}^-\beta_{1,1}^-\beta_{1,2}^+& \beta_{2,1}^+\beta_{1,1}^+\beta_{1,2}^+ &\beta_{2,1}^-\beta_{1,1}^-\bar\beta_{1,2}^- & \beta_{2,1}^-\beta_{1,1}^-\bar\beta_{1,2}^-& \beta_{2,1}^-\beta_{1,1}^-\bar\beta_{1,2}^+& \beta_{2,1}^+\beta_{1,1}^+\bar\beta_{1,2}^+ \\
    \beta_{4,1}^-\beta_{1,1}^-\beta_{1,2}^- & \beta_{4,1}^+\beta_{1,1}^-\beta_{1,2}^-& \beta_{4,1}^-\beta_{1,1}^-\beta_{1,2}^+& \beta_{4,1}^+\beta_{1,1}^+\beta_{1,2}^+ &\beta_{4,1}^-\beta_{1,1}^-\bar\beta_{1,2}^- & \beta_{4,1}^+\beta_{1,1}^-\bar\beta_{1,2}^-& \beta_{4,1}^-\beta_{1,1}^-\bar\beta_{1,2}^+& \beta_{4,1}^+\beta_{1,1}^+\bar\beta_{1,2}^+
    \end{pmatrix}
\begin{pmatrix}
    p_{(00)}\\
    p_{(10)}\\
    p_{(01)}\\
    p_{(11)}\\
    -\bar p_{(00)}\\
    -\bar p_{(10)}\\
    -\bar p_{(01)}\\
    -\bar p_{(11)}\\
    \end{pmatrix} = \zero,
    $}
\end{equation*}
if we abbreviate 
\begin{equation}\label{tilde pp}
    \tilde \pp = (   
    p_{(00)}\;\;
    p_{(10)}\;\;
    p_{(01)}\;\;
    p_{(11)}\;\;
    -\bar p_{(00)}\;
    -\bar p_{(10)}\;
    -\bar p_{(01)}\;
    -\bar p_{(11)}\;
    -\bar p_{(11)})^\top,
\end{equation}
then
\begin{align*}
    \resizebox{\linewidth}{!}{$
    \begin{pmatrix}
    0 & 0& 0& 1 & 0 & 0& 0& 1 \\
    0 & 0& \beta_{1,1}^-(\beta_{1,2}^+ - \beta_{1,2}^-)& \beta_{1,1}^+(\beta_{1,2}^+ - \beta_{1,2}^-) &\beta_{1,1}^-(\bar\beta_{1,2}^- - \beta_{1,2}^-) & \beta_{1,1}^-(\bar\beta_{1,2}^- - \beta_{1,2}^-)& \beta_{1,1}^-(\bar\beta_{1,2}^+ - \beta_{1,2}^-)& \beta_{1,1}^+(\bar\beta_{1,2}^+ -\beta_{1,2}^-) \\
    0 & 0& \beta_{2,1}^-\beta_{1,1}^-(\beta_{1,2}^+ - \beta_{1,2}^-)& \beta_{2,1}^+\beta_{1,1}^+(\beta_{1,2}^+ - \beta_{1,2}^-) &\beta_{2,1}^-\beta_{1,1}^-(\bar\beta_{1,2}^- - \beta_{1,2}^-)& \beta_{2,1}^-\beta_{1,1}^-(\bar\beta_{1,2}^- - \beta_{1,2}^-)& \beta_{2,1}^-\beta_{1,1}^-(\bar\beta_{1,2}^+ - \beta_{1,2}^-)& \beta_{2,1}^+\beta_{1,1}^+(\bar\beta_{1,2}^+ - \beta_{1,2}^-)\\
    0 & 0 & \beta_{4,1}^-\beta_{1,1}^-(\beta_{1,2}^+ - \beta_{1,2}^-)& \beta_{4,1}^+\beta_{1,1}^+(\beta_{1,2}^+ - \beta_{1,2}^-) &\beta_{4,1}^-\beta_{1,1}^-(\bar\beta_{1,2}^- - \beta_{1,2}^-)& \beta_{4,1}^+\beta_{1,1}^-(\bar\beta_{1,2}^- - \beta_{1,2}^-)& \beta_{4,1}^-\beta_{1,1}^-(\bar\beta_{1,2}^+ - \beta_{1,2}^-)& \beta_{4,1}^+\beta_{1,1}^+(\bar\beta_{1,2}^+ - \beta_{1,2}^-)
    \end{pmatrix}
\tilde\pp = \zero, 
$}
\end{align*}
thus
\begin{align}
    \resizebox{\linewidth}{!}{$
    \begin{pmatrix}
    0 & 0& \beta_{1,1}^-(\beta_{1,2}^+ - \beta_{1,2}^-)& 0 &\beta_{1,1}^-(\bar\beta_{1,2}^- - \beta_{1,2}^-) & \beta_{1,1}^-(\bar\beta_{1,2}^- - \beta_{1,2}^-)& \beta_{1,1}^-(\bar\beta_{1,2}^+ - \beta_{1,2}^-)& \beta_{1,1}^+(\bar\beta_{1,2}^+ -\beta_{1,2}^+) \\
    0 & 0& \beta_{2,1}^-\beta_{1,1}^-(\beta_{1,2}^+ - \beta_{1,2}^-)& 0 &\beta_{2,1}^-\beta_{1,1}^-(\bar\beta_{1,2}^- - \beta_{1,2}^-)& \beta_{2,1}^-\beta_{1,1}^-(\bar\beta_{1,2}^- - \beta_{1,2}^-)& \beta_{2,1}^-\beta_{1,1}^-(\bar\beta_{1,2}^+ - \beta_{1,2}^-)& \beta_{2,1}^+\beta_{1,1}^+(\bar\beta_{1,2}^+ - \beta_{1,2}^+)\\\label{fun1}
    0 & 0 & \beta_{4,1}^-\beta_{1,1}^-(\beta_{1,2}^+ - \beta_{1,2}^-)& 0 &\beta_{4,1}^-\beta_{1,1}^-(\bar\beta_{1,2}^- - \beta_{1,2}^-)& \beta_{4,1}^+\beta_{1,1}^-(\bar\beta_{1,2}^- - \beta_{1,2}^-)& \beta_{4,1}^-\beta_{1,1}^-(\bar\beta_{1,2}^+ - \beta_{1,2}^-)& \beta_{4,1}^+\beta_{1,1}^+(\bar\beta_{1,2}^+ - \beta_{1,2}^+)
    \end{pmatrix}
\tilde\pp = \zero,
    $}
\end{align}
which gives $(\beta_{2,1}^+- \beta_{2,1}^-)\beta_{1,1}^+(\bar\beta_{1,2}^+ - \beta_{1,2}^+) = 0$, so we have $\bar\beta_{1,2}^+ = \beta_{1,2}^+$. Taking this back to \eqref{fun1}, we have
\begin{equation}\label{fun2}
    \beta_{1,1}^-\begin{pmatrix}
    0 & 0 & (\beta_{1,2}^+ - \beta_{1,2}^-)& 0 &(\bar\beta_{1,2}^--\beta_{1,2}^-) & (\bar\beta_{1,2}^- - \beta_{1,2}^-)& (\beta_{1,2}^+ - \beta_{1,2}^-)& 0\\
         0& 0& \beta_{4,1}^-(\beta_{1,2}^+ - \beta_{1,2}^-)& 0 &\beta_{4,1}^-(\bar\beta_{1,2}^- - \beta_{1,2}^-) & \beta_{4,1}^+(\bar\beta_{1,2}^- - \beta_{1,2}^-)& \beta_{4,1}^-(\beta_{1,2}^+ - \beta_{1,2}^-)& 0
    \end{pmatrix}
    \tilde\pp = 0,
\end{equation}
which gives 
\begin{equation*}
    \beta_{1,1}^-\begin{pmatrix}
         0& 0& 0 & 0 & 0 & (\beta_{4,1}^+ - \beta_{4,1}^-)(\bar\beta_{1,2}^- - \beta_{1,2}^-)&  0 & 0
    \end{pmatrix}\tilde\pp = 0.
\end{equation*}
Assuming $\beta_{1,1}^- > 0$, we have $\bar\beta_{1,2}^- = \beta_{1,2}^-$. Taking this back to \eqref{fun2}, we have 
\[
\begin{pmatrix}
    0 & 0 & (\beta_{1,2}^+ - \beta_{1,2}^-)& 0 &0 & 0& (\beta_{1,2}^+ - \beta_{1,2}^-)& 0
\end{pmatrix} \tilde \pp = 0.
\]
Thus $p_{(01)} = \bar p_{(01)}$, and also $p_{(00)} = \bar p_{(00)}$. Therefore, the parameters of the Sequential DINA model with the above $\QQ$-matrix (\ref{joke2}) are identifiable.
\end{proof}

\subsubsection{Proof of Example \ref{notid1}}

The Sequential DINA model parameters with the following $\QQ$-matrix are not identifiable: 
\begin{equation*}
    \QQ = 
    \begin{pmatrix}
    item \; 1
    \begin{cases}
    1 & 0\\   
    \end{cases}\\
    \hdashline[2pt/2pt]
    item \; 2
    \begin{cases}
    0 & 1\\       
    \end{cases}\\
    \hdashline[2pt/2pt]
    item \; 3
    \begin{cases}
    0 & 1\\        
    \end{cases}\\
    \hdashline[2pt/2pt]
    item \; 4
    \begin{cases}
    1 & 1\\     
        1 & 0\\
    \end{cases}
    \end{pmatrix} \;\text{and}\;\;\QQ^1 = \begin{pmatrix}
    1 & 0\\
    \hdashline[2pt/2pt]
    0 & 1\\
    \hdashline[2pt/2pt]
    0 & 1\\
    \hdashline[2pt/2pt]
    1 & 1
    \end{pmatrix}. 
\end{equation*}
\begin{proof}
    Let $\bar\beta_{1,1}^+ = \beta_{1,1}^+,\; \bar\beta_{2,1}^+ = \beta_{2,1}^+\;\bar\beta_{3,1}^+ = \beta_{3,1}^+, \;\bar\beta_{4,2}^+ = \beta_{4,2}^+$, $\bar\beta_{2,1}^- = \beta_{2,1}^-,\;\bar\beta_{3,1}^- = \beta_{3,1}^-$, $\bar\beta_{4,1}^- = \beta_{4,1}^- = 0$. Since $\beta_{4,1}^- = 0$, $\beta_{4,2}^-$ is not defined (or equals 0). Then $\bar\TT\bar\pp = \TT\pp$ holds if and only if the following equations hold:
\begin{equation*}
\resizebox{1\hsize}{!}{$
\begin{cases}
 \bar p_{(00)} + \bar p_{(10)} + \bar p_{(01)} +  \bar p_{(11)} 
    = p_{(00)} + p_{(10)} + p_{(01)}+ p_{(11)};\\
    \bar p_{(01)} +  \bar p_{(11)} 
    =  p_{(01)} +   p_{(11)};\\
\bar \beta_{1,1}^-[\bar p_{(00)} +  \bar p_{(01)}] + \beta_{1,1}^+[\bar p_{(10)} +   \bar p_{(11)}] 
    = \beta_{1,1}^-[ p_{(00)} +   p_{(01)}] + \beta_{1,1}^+[ p_{(10)} +   p_{(11)}];\\
\bar\beta_{1,1}^-\beta_{2,1}^-\bar p_{(00)} +  \bar\beta_{1,1}^-\beta_{2,1}^+\bar p_{(01)} +  
\beta_{1,1}^+\beta_{2,1}^-\bar p_{(10)} + 
\beta_{1,1}^+\beta_{2,1}^+ \bar p_{(11)}
    =  \beta_{1,1}^-\beta_{2,1}^- p_{(00)} +   \beta_{1,1}^-\beta_{2,1}^+p_{(01)} +   \beta_{1,1}^+\beta_{2,1}^-p_{(10)} +   \beta_{1,1}^+\beta_{2,1}^+p_{(11)};\\
    \bar \beta_{4,1}^+ \bar p_{(11)} = \beta_{4,1}^+ p_{(11)}.
\end{cases}
$}
\end{equation*}
There are five equations with six parameters ($p_{(00)}, p_{(10)}, p_{(01)}, p_{(11)}, \beta_{1,1}^-, \beta_{4,1}^+$), thus there are infinitely many solutions and the parameters are not identifiable.
\end{proof}

\subsubsection{Proof of Example \ref{notid2}}
The Sequential DINA model parameters with the following $\QQ$-matrix are not identifiable: 
\begin{equation*}
    \QQ = 
    \begin{pmatrix}
    item \; 1
    \begin{cases}
    1 & 0\\   
    \end{cases}\\
    \hdashline[2pt/2pt]
    item \; 2
    \begin{cases}
    0 & 1\\       
    \end{cases}\\
    \hdashline[2pt/2pt]
    item \; 3
    \begin{cases}
    1 & 1\\        
    1 & 0\\
    \end{cases}\\
    \hdashline[2pt/2pt]
    item \; 4
    \begin{cases}
    1 & 1\\        
    \end{cases}
    \end{pmatrix} \;\text{and}\;\;\QQ^1 = \begin{pmatrix}
    1 & 0\\
    \hdashline[2pt/2pt]
    0 & 1\\
    \hdashline[2pt/2pt]
    1 & 1\\
    \hdashline[2pt/2pt]
    1 & 1
    \end{pmatrix}. 
\end{equation*}
\begin{proof}
Let $\bar\beta_{j,l}^+ = \beta_{j,l}^+$ for all $j$ and $l \in [H_j]$, $\bar\beta_{4,1}^- = \beta_{4,1}^-$, $\bar\beta_{3,1}^- = \beta_{3,1}^- = 0$, and $\bar p_{(11)} = p_{(11)}$. Since $\beta_{3,1}^- = 0$, $\beta_{3,2}^-$ is not defined (or equals 0). Then $\bar\TT\bar\pp = \TT\pp$ holds if and only if the following equations hold:
\begin{equation*}
\resizebox{1\hsize}{!}{$
\begin{cases}
 \bar p_{(00)} + \bar p_{(10)} + \bar p_{(01)}
    = p_{(00)} + p_{(10)} + p_{(01)};\\
\bar \beta_{1,1}^-[\bar p_{(00)} +  \bar p_{(01)}] + \beta_{1,1}^+[\bar p_{(10)} +   p_{(11)}] 
    = \beta_{1,1}^-[ p_{(00)} +   p_{(01)}] + \beta_{1,1}^+[ p_{(10)} +   p_{(11)}];\\
\bar \beta_{2,1}^-[\bar p_{(00)} +  \bar p_{(10)}] + \beta_{2,1}^+[\bar p_{(01)} +  p_{(11)}] 
    = \beta_{2,1}^-[ p_{(00)} +   p_{(10)}] + \beta_{2,1}^+[ p_{(01)} +   p_{(11)}];\\
\bar\beta_{1,1}^-\bar \beta_{2,1}^-\bar p_{(00)} +  \bar\beta_{1,1}^-\bar \beta_{2,1}^+\bar p_{(01)} +  
\beta_{1,1}^+\beta_{2,1}^-\bar p_{(10)} + 
\beta_{1,1}^+\beta_{2,1}^+ p_{(11)}
    =  \beta_{1,1}^-\beta_{2,1}^- p_{(00)} +   \beta_{1,1}^-\beta_{2,1}^+p_{(01)} +   \beta_{1,1}^+\beta_{2,1}^-p_{(10)} +   \beta_{1,1}^+\beta_{2,1}^+p_{(11)}.
\end{cases}
$}
\end{equation*}
There are four equations with five parameters ($p_{(00)}, p_{(10)}, p_{(01)}, \beta_{1,1}^-, \beta_{2,1}^-$), thus there are infinitely many solutions and the parameters are not identifiable.
\end{proof}

\subsubsection{Proof of Example~\ref{counterex2}}
The Sequential DINA model parameters with the following $\QQ$-matrix are identifiable: 
\begin{equation}
    \QQ = \begin{pmatrix}
    item \;1
    \begin{cases}
        1 & 0\\
    0 & 1\\
    \end{cases}\\
    \hdashline[2pt/2pt]
    item \;2
    \begin{cases}
    1 & 0\\
    0 & 1\\        
    \end{cases}\\
    \hdashline[2pt/2pt]
    item \;3
    \begin{cases}
    0 & 1\\
    1 & 0\\
    \end{cases}\\
    \hdashline[2pt/2pt]
    item \; 4
    \begin{cases}
    0 & 1\\
    1 & 0\\        
    \end{cases}
    \end{pmatrix} \;\text{and}\;\;\QQ^1 = \begin{pmatrix}
    1 & 0\\
    \hdashline[2pt/2pt]
    1 & 0\\
    \hdashline[2pt/2pt]
    0 & 1\\
    \hdashline[2pt/2pt]
    0 & 1\\
    \end{pmatrix}. 
\end{equation}
\begin{proof}
Consider equations $\TT^s_{\rr} \pp = \bar\TT^s_{\rr} \bar\pp$ with $\rr \in \{ \ee_1,\,2\ee_1,\,\;\ee_3, \; \ee_4 \}$ and their sums. Using a similar calculation, 
these equations give
\begin{align}\label{kk1}
\resizebox{1.1\hsize}{!}{$
    \begin{pmatrix} 
    \zero_4^\top & \bar\beta_{1,1}^-(\bar\beta_{1,2}^- - \beta_{1,2}^-)(\bar \beta_{3,1}^- - \beta_{3,1}^+) & \bar\beta_{1,1}^+(\bar\beta_{1,2}^- - \beta_{1,2}^-)(\bar \beta_{3,1}^- - \beta_{3,1}^+)& \bar\beta_{1,1}^-(\bar\beta_{1,2}^+ - \beta_{1,2}^-)(\bar \beta_{3,1}^+ - \beta_{3,1}^+)& \bar\beta_{1,1}^+(\bar\beta_{1,2}^+ - \beta_{1,2}^-) (\bar \beta_{3,1}^+ - \beta_{3,1}^+)\\
    \zero_4^\top &\bar\beta_{1,1}^-(\bar\beta_{1,2}^- - \beta_{1,2}^-)(\bar \beta_{3,1}^- - \beta_{3,1}^+)\bar \beta_{4,1}^-& \bar\beta_{1,1}^+(\bar\beta_{1,2}^- - \beta_{1,2}^-)(\bar \beta_{3,1}^- - \beta_{3,1}^+)\bar \beta_{4,1}^-& \bar\beta_{1,1}^-(\bar\beta_{1,2}^+ - \beta_{1,2}^-)(\bar \beta_{3,1}^+ - \beta_{3,1}^+)\bar \beta_{4,1}^+& \bar\beta_{1,1}^+(\bar\beta_{1,2}^+ - \beta_{1,2}^-) (\bar \beta_{3,1}^+ - \beta_{3,1}^+)\bar \beta_{4,1}^+
    \end{pmatrix}
    \Tilde{\pp} = \zero,
$}
\end{align}
where $\zero_4$ denotes all-zero vectors of dimension four and $\tilde\pp$ is \eqref{tilde pp}. Thus,
\begin{equation*}
        \resizebox{\linewidth}{!}{$ 
    \begin{pmatrix} 
\zero_5^\top & \bar\beta_{1,1}^-(\bar\beta_{1,2}^+ - \beta_{1,2}^-)(\bar \beta_{3,1}^+ - \beta_{3,1}^+)(\bar \beta_{4,1}^+ - \bar \beta_{4,1}^-)& \bar\beta_{1,1}^+(\bar\beta_{1,2}^+ - \beta_{1,2}^-) (\bar \beta_{3,1}^+ - \beta_{3,1}^+)(\bar \beta_{4,1}^+ - \bar \beta_{4,1}^-)
    \end{pmatrix}
    \Tilde{\pp} = \zero.
$}
\end{equation*}
Therefore, \[(\bar\beta_{1,1}^-\bar p_{01} + \bar\beta_{1,1}^+\bar p_{11})
(\bar\beta_{1,2}^+ - \beta_{1,2}^-) (\bar \beta_{3,1}^+ - \beta_{3,1}^+)(\bar \beta_{4,1}^+ - \bar \beta_{4,1}^-) = 0.\] Since $0 \leq \beta_{j,l}^- < \beta_{j,l}^+ \leq 1$, $(\bar\beta_{1,2}^+ - \beta_{1,2}^-) (\bar \beta_{3,1}^+ - \beta_{3,1}^+) = 0$. But if $\bar \beta_{3,1}^+ \neq \beta_{3,1}^+$ and $\bar\beta_{1,2}^+ = \beta_{1,2}^-$, we can swap the parameters $\bbb$ and $\bar\bbb$, and show that $\beta_{1,2}^+ = \bar\beta_{1,2}^-$ by symmetry. Yet this will indicate that $\beta_{1,2}^+ = \bar\beta_{1,2}^- < \bar\beta_{1,2}^+ = \beta_{1,2}^-$, which leads to a contradiction, thus we must have $\bar \beta_{3,1}^+ = \beta_{3,1}^+$. By symmetry we can also show that $\bar\beta_{4,1}^+ = \beta_{4,1}^+$, and similarly for $\bar\beta_{1,1}^+ = \beta_{1,1}^+$ and $\bar \beta_{2,1}^+ = \beta_{2,1}^+$. Taking these back to equation (\ref{kk1}), we have $\bar \beta_{1,2}^- = \beta_{1,2}^-$. Samely we obtain $\bar \beta_{j,2}^- = \beta_{j,2}^-$ for $j \in [4]$. Thus all the item parameters are identified, and according to Lemma~\ref{lem-ex2},
we know that when these parameters are known, the completeness of the $\QQ$-matrix will suffice to identify parameters $\pp$. Therefore, we must have $(\bbb^+, \bbb^-, \pp) = (\bar\bbb^+, \bar\bbb^-, \bar\pp)$, i.e., the model parameters are identifiable.
\end{proof}

\subsubsection{Proof of Example~\ref{counterex3}}
The Sequential DINA model parameters with the following $\QQ$-matrix are identifiable: 

\begin{equation*}
    \QQ = 
    \begin{pmatrix}
    item \; 1
    \begin{cases}
    1 & 0\\
    0 & 1\\        
    \end{cases}\\
    \hdashline[2pt/2pt]
    item \; 2
    \begin{cases}
    0 & 1\\
    1 & 0\\        
    \end{cases}\\
    \hdashline[2pt/2pt]
    item \; 3
    \begin{cases}
    1 & 1\\        
    \end{cases}\\
    \hdashline[2pt/2pt]
    item \; 4
    \begin{cases}
    1 & 1\\        
    \end{cases}\\
    \hdashline[2pt/2pt]
    item \; 5
    \begin{cases}
    1 & 1\\        
    \end{cases}
    \end{pmatrix} \;\text{and}\;\;\QQ^1 = \begin{pmatrix}
    1 & 0\\
    \hdashline[2pt/2pt]
    0 & 1\\
    \hdashline[2pt/2pt]
    1 & 1\\
    \hdashline[2pt/2pt]
    1 & 1\\
    \hdashline[2pt/2pt]
    1 & 1\\
    \end{pmatrix}. 
\end{equation*}
\begin{proof}
Consider the last three items, similar to the proof in Example \ref{counterex1}, 
according to Lemma~\ref{lem-ex1},
we have $p_{(11)} = \bar p_{(11)}$, and $\beta_{j,1}^+ = \bar\beta_{j,1}^+$, $\beta_{j,1}^- = \bar\beta_{j,1}^-$ for $j \in \{3,4,5\}$. 
Next using equations $\TT^s_{\rr} \pp = \bar\TT^s_{\rr} \bar\pp$ with $\rr = \ee_1, \ee_1 + \ee_3$, we have
\begin{equation*}
    \begin{pmatrix}
    \beta_{1,1}^- & \beta_{1,1}^+& \beta_{1,1}^-& \beta_{1,1}^+&\bar\beta_{1,1}^- & \bar\beta_{1,1}^+& \bar\beta_{1,1}^-& \bar\beta_{1,1}^+\\ 
    \beta_{1,1}^- \beta_{3,1}^-& \beta_{1,1}^+ \beta_{3,1}^-& \beta_{1,1}^- \beta_{3,1}^-& \beta_{1,1}^+ \beta_{3,1}^+ &\bar\beta_{1,1}^-\beta_{3,1}^-& \bar\beta_{1,1}^+ \beta_{3,1}^-& \bar\beta_{1,1}^-\beta_{3,1}^-& \bar\beta_{1,1}^+ \beta_{3,1}^+
    \end{pmatrix}
    \tilde \pp = \zero,
\end{equation*}
which gives \[
\begin{pmatrix}
   0&0& 0& \beta_{1,1}^+ (\beta_{3,1}^+ - \beta_{3,1}^-)&0& 0& 0& \bar\beta_{1,1}^+ (\beta_{3,1}^+ - \beta_{3,1}^-)
    \end{pmatrix}
    \Tilde{\pp} = 0,
\] 
where $\tilde \pp$ is given by \eqref{tilde pp}.
Since $p_{(11)} = \bar p_{(11)}$, we must have $\beta_{1,1}^+ = \bar\beta_{1,1}^+$. Next combining $\rr = 2\ee_1$ and $\rr = 2\ee_1 + \ee_3$, 
\begin{equation*}
\resizebox{.99\textwidth}{!}{$\begin{pmatrix}
    \beta_{1,1}^-\beta_{1,2}^- & \beta_{1,1}^+\beta_{1,2}^-& \beta_{1,1}^-\beta_{1,2}^+& \beta_{1,1}^+\beta_{1,2}^+ &\bar\beta_{1,1}^-\bar\beta_{1,2}^- & \bar\beta_{1,1}^+\bar\beta_{1,2}^-& \bar\beta_{1,1}^-\bar\beta_{1,2}^+& \bar\beta_{1,1}^+\bar\beta_{1,2}^+\\ 
    \beta_{1,1}^-\beta_{1,2}^- \beta_{3,1}^-& \beta_{1,1}^+\beta_{1,2}^- \beta_{3,1}^-& \beta_{1,1}^-\beta_{1,2}^+ \beta_{3,1}^-& \beta_{1,1}^+\beta_{1,2}^+ \beta_{3,1}^+ &\bar\beta_{1,1}^-\bar\beta_{1,2}^- \beta_{3,1}^-& \bar\beta_{1,1}^+\bar\beta_{1,2}^- \beta_{3,1}^-& \bar\beta_{1,1}^-\bar\beta_{1,2}^+\beta_{3,1}^-& \bar\beta_{1,1}^+\bar\beta_{1,2}^+ \beta_{3,1}^+
    \end{pmatrix}
    \tilde \pp = \zero,$}
\end{equation*}
which gives $\beta_{1,2}^+ = \bar\beta_{1,2}^+$. Similarly we can obtain 
$\beta_{2,1}^+ = \bar\beta_{2,1}^+$ and $\beta_{2,2}^+ = \bar\beta_{2,2}^+$. Next consider $\rr \in \{ \ee_1, 2\ee_1,\, \ee_1+\ee_2, 2\ee_1+\ee_2\}$, using similar calculations, we have,
\begin{equation*}
\resizebox{1.1\hsize}{!}{$
    \begin{pmatrix}
    0 & 0& \beta_{1,1}^-(\beta_{1,2}^+ - \beta_{1,2}^-)& \beta_{1,1}^+(\beta_{1,2}^+ - \beta_{1,2}^-)&\bar\beta_{1,1}^-(\bar\beta_{1,2}^- - \beta_{1,2}^-)& \beta_{1,1}^+(\bar\beta_{1,2}^- - \beta_{1,2}^-)& \bar\beta_{1,1}^-(\beta_{1,2}^+ - \beta_{1,2}^-)& \beta_{1,1}^+(\beta_{1,2}^+ - \beta_{1,2}^-)\\ 
    0& 0& \beta_{1,1}^-(\beta_{1,2}^+ - \beta_{1,2}^-)\beta_{2,1}^+& \beta_{1,1}^+(\beta_{1,2}^+ - \beta_{1,2}^-)\beta_{2,1}^+ &\bar\beta_{1,1}^-(\bar\beta_{1,2}^- - \beta_{1,2}^-)\bar \beta_{2,1}^-& \beta_{1,1}^+(\bar\beta_{1,2}^- - \beta_{1,2}^-)\bar \beta_{2,1}^-& \bar\beta_{1,1}^-(\beta_{1,2}^+ - \beta_{1,2}^-)\beta_{2,1}^+& \beta_{1,1}^+(\beta_{1,2}^+ - \beta_{1,2}^-) \beta_{2,1}^+
    \end{pmatrix}
    \Tilde{\pp} = \zero,
$}
\end{equation*}
which gives \[
(\bar\beta_{1,1}^- \bar p_{(00)} +
\bar\beta_{1,1}^+ \bar p_{(10)})(\bar\beta_{1,2}^- - \beta_{1,2}^-)(\bar \beta_{2,1}^- - \beta_{2,1}^+) = 0,
\] 
thus we have $\bar\beta_{1,2}^- = \beta_{1,2}^-$.
Finally we identify $\beta_{2,1}^-$,
\begin{equation*}
\resizebox{1.1\hsize}{!}{$
    \begin{pmatrix}
    \beta_{1,1}^- & \beta_{1,1}^+& \beta_{1,1}^-& \beta_{1,1}^+&\bar\beta_{1,1}^- & \beta_{1,1}^+& \bar\beta_{1,1}^-& \beta_{1,1}^+\\ 
    \beta_{1,1}^-\beta_{1,2}^- & \beta_{1,1}^+\beta_{1,2}^-& \beta_{1,1}^-\beta_{1,2}^+& \beta_{1,1}^+\beta_{1,2}^+ &\bar\beta_{1,1}^-\bar\beta_{1,2}^- & \beta_{1,1}^+\bar\beta_{1,2}^-& \bar\beta_{1,1}^-\beta_{1,2}^+& \beta_{1,1}^+\beta_{1,2}^+\\ 
    \beta_{1,1}^- \beta_{2,1}^-& \beta_{1,1}^+ \beta_{2,1}^-& \beta_{1,1}^- \beta_{2,1}^+& \beta_{1,1}^+ \beta_{2,1}^+ &\bar\beta_{1,1}^-\bar \beta_{2,1}^-& \beta_{1,1}^+ \bar \beta_{2,1}^-& \bar\beta_{1,1}^-\beta_{2,1}^+& \beta_{1,1}^+ \beta_{2,1}^+\\
    \beta_{1,1}^-\beta_{1,2}^- \beta_{2,1}^-& \beta_{1,1}^+\beta_{1,2}^- \beta_{2,1}^-& \beta_{1,1}^-\beta_{1,2}^+ \beta_{2,1}^+& \beta_{1,1}^+\beta_{1,2}^+ \beta_{2,1}^+ &\bar\beta_{1,1}^-\bar\beta_{1,2}^- \bar \beta_{2,1}^-& \beta_{1,1}^+\bar\beta_{1,2}^- \bar \beta_{2,1}^-& \bar\beta_{1,1}^-\beta_{1,2}^+\beta_{2,1}^+& \beta_{1,1}^+\beta_{1,2}^+ \beta_{2,1}^+
    \end{pmatrix}
    \tilde \pp = \zero
$}
\end{equation*}
\vspace{-9mm}
\begin{equation*}
\resizebox{1.1\hsize}{!}{$
\Rightarrow
\begin{pmatrix}
    0& 0& \beta_{1,1}^- (\beta_{2,1}^+ - \beta_{2,1}^-)& \beta_{1,1}^+ (\beta_{2,1}^+ - \beta_{2,1}^-)&\bar\beta_{1,1}^- (\bar \beta_{2,1}^- - \beta_{2,1}^-)& \beta_{1,1}^+(\bar \beta_{2,1}^- - \beta_{2,1}^-)& \bar\beta_{1,1}^-(\beta_{2,1}^+ - \beta_{2,1}^-)& \beta_{1,1}^+ (\beta_{2,1}^+ - \beta_{2,1}^-)\\
    0& 0& \beta_{1,1}^-\beta_{1,2}^+ (\beta_{2,1}^+ - \beta_{2,1}^-)& \beta_{1,1}^+\beta_{1,2}^+ (\beta_{2,1}^+ - \beta_{2,1}^-)&\bar\beta_{1,1}^-\bar\beta_{1,2}^- (\bar \beta_{2,1}^- - \beta_{2,1}^-)& \beta_{1,1}^+\bar\beta_{1,2}^- (\bar \beta_{2,1}^- - \beta_{2,1}^-)& \bar\beta_{1,1}^-\beta_{1,2}^+(\beta_{2,1}^+ - \beta_{2,1}^-)& \beta_{1,1}^+\beta_{1,2}^+ (\beta_{2,1}^+ - \beta_{2,1}^-)
    \end{pmatrix}
    \Tilde{\pp} = \zero.
$}
\end{equation*}
Hence 
\[
\begin{pmatrix}
    0& 0& 0& 0&\bar\beta_{1,1}^-(\bar\beta_{1,2}^- - \beta_{1,2}^+)(\bar \beta_{2,1}^- - \beta_{2,1}^-)& \beta_{1,1}^+(\bar\beta_{1,2}^- - \beta_{1,2}^+) (\bar \beta_{2,1}^- - \beta_{2,1}^-)& 0& 0
    \end{pmatrix}
    \Tilde{\pp} = 0.
\] 
Therefore, we have \[(\bar\beta_{1,1}^- \bar p_{(00)} +
\bar\beta_{1,1}^+ \bar p_{(10)})(\bar\beta_{1,2}^- - \beta_{1,2}^+) (\bar \beta_{2,1}^- - \beta_{2,1}^-) = 0,\] thus $\bar \beta_{2,1}^- = \beta_{2,1}^-$.
Using the same strategy we can show that $\bar \beta_{2,2}^- = \beta_{2,2}^-$ and $\bar \beta_{1,1}^- = \beta_{1,1}^-$, which completes the proof.
\end{proof}

\end{document}